\documentclass[12pt,a4paper]{article}

\usepackage[utf8]{inputenc}      
\usepackage[T1]{fontenc}
\usepackage{lmodern}             
\usepackage{textcomp}
\usepackage{microtype}           
\usepackage[english]{babel}
\usepackage{csquotes}

\usepackage[a4paper,left=2.5cm,right=2.5cm,top=2.5cm,bottom=2.5cm]{geometry}
\usepackage{setspace}
\usepackage{amsmath,amssymb,amsfonts}
\usepackage{mathtools}           
\usepackage{bm}                  
\usepackage{amsthm}
\allowdisplaybreaks              

\usepackage{array}
\usepackage{booktabs}
\usepackage{tabularx}
\usepackage{dcolumn}
\usepackage{multirow}
\usepackage{makecell}
\usepackage{longtable}           
\usepackage{threeparttable}
\usepackage{adjustbox}
\usepackage{pdflscape}
\usepackage{rotating}

\newcolumntype{d}[1]{D{.}{.}{#1}}

\newcommand{\sym}[1]{\ifmmode^{#1}\else\(^{#1}\)\fi}

\usepackage{graphicx}            
\usepackage{epstopdf}
\usepackage{float}
\usepackage{placeins}            
\usepackage{subcaption}
\usepackage[font=small,labelfont=bf,justification=justified,
            singlelinecheck=false]{caption}
\graphicspath{{figures/}{Figures/}{graphs/}}

\usepackage{tikz}
\usetikzlibrary{shapes.geometric,arrows,arrows.meta,positioning,calc}

\usepackage{enumitem}
\usepackage{comment}
\usepackage{xcolor}
\usepackage{authblk}             
\usepackage[title]{appendix}

\usepackage[authoryear,round,longnamesfirst]{natbib}
\bibpunct{(}{)}{;}{a}{,}{,}
\usepackage{hyperref}
\hypersetup{
    colorlinks = true,
    linkcolor  = blue,
    citecolor  = blue,
    filecolor  = magenta,
    urlcolor   = blue,
    breaklinks = true,
    pdftitle   = {Complements or Substitutes? Technology Adoption and Clinical Care Utilization: Evidence from Automated Insulin Delivery},
    pdfauthor  = {Moslem Rashidi},
    pdfkeywords= {automated insulin delivery; task-based automation;
                  staggered difference-in-differences; healthcare utilization}
}
\usepackage{doi}                 
\usepackage[capitalise,noabbrev]{cleveref}   

\usepackage{thmtools}
\usepackage{thm-restate}

\theoremstyle{plain}

\newtheorem{proposition}{Proposition}

\newtheorem{corollary}{Corollary}

\theoremstyle{definition}
\newtheorem{assumption}{Assumption}

\theoremstyle{remark}

\newtheorem{repcorollary}{Corollary}

\title{Complements or Substitutes? Technology Adoption and Clinical Care Utilization: Evidence from Automated Insulin Delivery}
\author{
Moslem Rashidi\thanks{Corresponding author. Email: \texttt{moslem.rashidi2@unibo.it}.}
\quad
Cristina Ugolini\thanks{Email: \texttt{cristina.ugolini@unibo.it}.}
\quad
Gianluca Fiorentini\thanks{Email: \texttt{gianluca.fiorentini@unibo.it}.}
\\
\small Department of Economics, University of Bologna \\
\small Piazza Scaravilli, 40126 Bologna, Italy
}
\date{\today}

\begin{document}

\maketitle

\begin{abstract}
\noindent Whether medical technology reduces or increases demand for professional care is central to understanding healthcare costs. We study automated insulin delivery (AID) adoption among 1,608 patients with type 1 diabetes treated in four Italian National Health Service clinics, including 283 adopters. To address clinically targeted and staggered adoption, we combine risk-set coarsened exact matching with a staggered event-study design and adjust for differential pre-adoption trends through event-time extrapolation. After adjustment, AID adoption is associated with increased routine outpatient engagement. The estimated increase in the probability of a diabetologist visit is 16.3 percentage points one semester after adoption and 39.3 percentage points four semesters after adoption, relative to a pre-adoption visit probability of 55.2\%. HbA1c testing and a process-of-care index show similar but weaker evidence. The visit estimates are sensitive to the extrapolation assumption and should be interpreted conditionally on that restriction. The results are consistent with a task-based view of medical technology: automation may replace routine dosing decisions while increasing demand for clinical interpretation, adjustment, and supervision. Thus, the economic consequences of patient-facing automation may extend beyond device costs and clinical effectiveness to the professional resources required to support its use.
\end{abstract}

\begin{center}

\bigskip


\end{center}

\vfill

\thispagestyle{empty}

\section{Introduction}
\label{sec:introduction}
New medical technologies are typically evaluated on clinical grounds: do they improve health outcomes? From a health economics perspective, however, the question is broader. Do they reduce the need for professional clinical services, or do they create new forms of demand for specialist input, monitoring, and follow-up? A technology that improves outcomes while also increasing specialist contact changes the cost structure of care, not just its quality. Technological change and evolving medical practice are major drivers of healthcare expenditure growth, in part because they alter the intensity and composition of care delivered to patients (Chandra and Skinner, 2012; Laudicella et al., 2022). Assessing the economic consequences of medical technology therefore requires understanding not only whether it improves health, but also how it changes the demand for professional care.

The economics of automation provides a useful framework for understanding this trade-off. Acemoglu and Restrepo (2019) develop a framework in which automation can displace workers from tasks it takes over but can also create new tasks in which labor keeps a comparative advantage, an effect that raises rather than lowers labor demand. Both displacement from automated tasks and reinstatement through the creation of new ones have been documented in manufacturing and digital services (Autor, Levy, and Murnane, 2003; Acemoglu and Restrepo, 2019). Applied to healthcare, this framework implies that medical automation may reduce some forms of professional input while increasing the need for others. Yet whether this reinstatement mechanism operates in clinical care remains less well understood.

Automated insulin delivery (AID) systems combine a continuous glucose monitor, an insulin pump, and a control algorithm that automatically adjusts insulin delivery in response to real-time glucose data. The system modulates basal insulin delivery and may also administer automatic correction boluses, while patients remain responsible for mealtime boluses and for responding to alarms and technical problems. Many patients eligible for AID already use a glucose sensor or an insulin pump, so continuous glucose data and insulin-delivery records are not new in themselves. What distinguishes AID is the closed-loop integration of these components: glucose readings, insulin delivery, and the algorithm’s dosing decisions are recorded together as part of the same treatment process. This integrated record gives clinicians information not available under conventional management, even for patients with prior sensor or pump experience, and may therefore create new needs for specialist interpretation, treatment adjustment, and follow-up.

The opposite prediction is also plausible. In Grossman’s (1972) health-production framework, a technology that raises the productivity of the patient’s own time in managing their health may reduce their demand for purchased medical care. A related mechanism is risk compensation (Peltzman, 1975): if technology lowers the consequences of imperfect self-management, patients may rationally reduce other forms of precaution, including scheduled visits and monitoring. Both mechanisms therefore predict fewer contacts with professional care after adoption, although behavioral frictions may weaken or complicate this response (Baicker, Mullainathan, and Schwartzstein, 2015). Consistent with this substitution channel, Schmid (2015) finds that greater access to consumer health information lowers demand for physician visits. In that setting, however, the information is used directly by patients, whereas AID generates data and algorithmic decisions that may themselves require clinical interpretation. 

Empirical evidence on whether medical technology substitutes for professional care is mixed. In the broader health-IT literature, Agha (2014) finds that hospital information technology raised spending without commensurate quality gains, Lee, McCullough, and Town (2013) find only modest productivity contributions from IT investment, and B\"ockerman et al. (2025a) find that nationwide e-prescribing had no evident average effect on harmful drug co-prescribing, but substantially improved prescribing quality in rural regions, where information integration facilitated coordination across physicians. Evidence from patient-facing technologies likewise provides little support for a simple substitution effect. Roubos et al.\ (2025) study a telemonitoring platform used as a substitute for outpatient visits in COPD and as a complement to them in asthma: visit frequency held steady in both groups, while remote consultations rose substantially. Similarly, Zeltzer et al.\ (2023) find that adoption of a home diagnostic device raised primary-care utilization by 12 percent while reducing use of more intensive settings, with no increase in total spending. In a related setting, B\"ockerman et al.\ (2025b) show that e-prescribing reduced the hassle costs of prescription renewal and increased medication use, illustrating how healthcare technology can expand utilization along some margins rather than simply displace existing care. Accordingly, Horn, Sacarny, and Zhou (2022) show that adopting robotic surgery can expand utilization and shift patient allocation across providers. These findings suggest that new technologies may reallocate care composition rather than simply reduce its use.

For AID specifically, existing evidence has focused on clinical outcomes and acute healthcare utilization: adopters experience lower rates of emergency and inpatient care (Manjelievskaia et al., 2025), alongside the improvements in glycemic control documented in the clinical literature (Burnside et al., 2022; St{\aa}hl et al., 2026). What remains less known is how AID affects routine outpatient engagement: the specialist visits and laboratory monitoring that constitute the recurring component of diabetes care. 

Two features of our institutional setting may further favor complementarity between AID and professional care, although neither determines the direction of the utilization response. First, prescribing physicians retain responsibility for monitoring treatment and device use after adoption. Second, within the publicly funded Italian NHS, routine diabetes care entails little or no out-of-pocket cost for patients covered by chronic-disease exemptions. Together, these features may make additional specialist monitoring more likely to be incorporated into the care pathway following AID adoption.

Evaluating how AID adoption affects healthcare utilization is challenging because its causal effect is difficult to identify for two main reasons. First, adoption is clinically targeted rather than randomly assigned. Physicians prescribe AID in response to a patient’s evolving clinical trajectory, so adoption timing is correlated with within-patient changes in disease severity that also affect healthcare utilization. Adopters and untreated comparison patients may therefore differ not only in their clinical history at the time of adoption, but also in the trajectory of utilization leading up to it. Second, adoption occurs at different times across patients. When treatment effects are heterogeneous across cohorts and over time, conventional two-way fixed-effects event-study estimators may assign negative weight to some cohort-time effects and distort the estimated dynamic path (de Chaisemartin and D'Haultf\oe uille, 2020; Goodman-Bacon, 2021).

The design addresses these problems in three steps. First, risk-set coarsened exact matching pairs adopters with untreated comparison patients who share a comparable pre-adoption clinical history within each adoption risk set, retaining adopters whose history has support among eligible controls and reducing observed differences in clinical characteristics at the time of adoption. Second, the interaction-weighted estimator of Sun and Abraham (2021) builds the dynamic treatment path from comparisons between newly treated and not-yet-treated or never-treated patients only, avoiding the negative weighting problem that can arise in conventional two-way fixed-effects event studies. Third, following Dobkin et al.\ (2018) and Freyaldenhoven et al.\ (2025a, 2025b), we project the pre-adoption trajectory into the post-adoption window and assess whether the estimated effect remains after accounting for that trend. This projection relies on an explicit identifying restriction that residual confounding follows a low-dimensional path in event time. We treat this restriction as an assumption to be made transparent and assessed, not as a property delivered by the event study design itself.  

We find that adoption raises routine outpatient engagement once differential pre-adoption trends are taken into account. The estimated effect on the probability of a diabetologist visit is 16.3 percentage points one semester after adoption and rises to 39.3 percentage points four semesters later against a pre-adoption mean in the reference semester immediately before adoption of 55.2 percent. HbA1c testing and the broader process-of-care index move in the same direction, although the evidence for these outcomes is weaker. Because the visit estimates remain sensitive to departures from the extrapolation assumption, we interpret them as conditional on that identifying restriction rather than as design-based causal effects (Rambachan and Roth, 2023). Overall, the results indicate that AID does not reduce routine specialist involvement. This pattern is consistent with medical automation complementing, rather than simply substituting for, professional clinical inputs. More broadly, patient-facing technologies may alter the composition of healthcare production and the demand for clinical labor, with implications for the cost structure of chronic care and specialist capacity. 

The paper makes two contributions. First, we provide evidence on how AID adoption affects routine outpatient engagement, a margin not addressed by the existing AID literature. This complements previous evidence on improvements in glycemic outcomes and reductions in acute and inpatient utilization among AID users (Burnside et al., 2022; Manjelievskaia et al., 2025; Ståhl et al., 2026) and provides a more complete picture of how AID may affect the composition of healthcare use. The closest evidence concerns continuous glucose monitoring (CGM), a component technology rather than the closed-loop system itself. Karter et al.\ (2021) find that initiation of real-time CGM was associated with fewer outpatient visits and more telephone contacts, suggesting that technology may alter the mode of clinical interaction. AID differs because it integrates glucose data, insulin delivery, and algorithmic dosing decisions into a single treatment record, potentially generating additional information for clinical review. We examine whether this more integrated form of technology changes the intensity of routine specialist care. 

Second, we extend the task-based framework of automation to a patient-facing therapeutic technology, where automation occurs in patients’ day-to-day disease management while professional input is supplied by clinicians. This differs from settings in which algorithms assist clinicians directly. Evidence from those settings is mixed: Agarwal et al.\ (2023), for example, find that AI assistance does not improve radiologists’ diagnostic accuracy on average and that human and algorithmic inputs behave more like substitutes than complements on a given case. Our findings are instead consistent with a setting in which automation may complement professional clinical inputs rather than simply displace them. The estimates do not allow us to distinguish task reinstatement from an expansion in the scale or intensity of care, but, conditional on the identifying assumptions and on inelastic demand for glycemic control, they are difficult to reconcile with a pure substitution or risk-compensation mechanism. More broadly, this implies that the economic consequences of patient-facing medical technology may extend beyond device costs and clinical effectiveness to the professional resources required to support its use.

Section~\ref{sec:institutional background and conceptual framework} sets out the institutional background and the conceptual framework; Section~\ref{sec:Data} describes the data; Section~\ref{sec:empirical_strategy} details the empirical strategy; Section~\ref{sec:results} reports results and robustness; Section~\ref{sec:Conclusion} concludes.

\section{Institutional Background and Conceptual Framework}
\label{sec:institutional background and conceptual framework}

\subsection{Clinical pathway and AID adoption}
\label{subsec:background}
Our institutional setting is the Italian National Health Service (NHS), with a focus on the Emilia-Romagna region. The NHS is a universal, tax-funded system organized regionally. Regions are responsible for planning and financing healthcare, while Local Health Authorities organize service delivery through hospitals, community services, and specialist outpatient clinics. In Emilia-Romagna, type~1 diabetes (T1D) is managed primarily through specialist outpatient services, where salaried diabetologists oversee insulin therapy, glycemic control, diabetes technologies, laboratory monitoring, and complication screening according to national clinical standards. The regional public healthcare system provides visits, laboratory tests, and prescribed devices, with patient co-payments governed by chronic-disease exemption rules.

As described in the Introduction, automated insulin delivery (AID) integrates continuous glucose monitoring, insulin delivery, and algorithmic control. Hybrid closed-loop AID systems were commercially available in Italy by early 2019. Clinical standards strongly recommend AID for suitable patients with T1D, based on evidence of improvements in HbA1c, time in range, and glycemic variability (Pintaudi et al., 2026). Access to AID is clinically targeted. National guidelines in place since 2022 recommended automated pump–sensor systems for patients with inadequate glycemic control despite existing pump and continuous glucose monitoring technologies (Pintaudi et al., 2022). Regional guidance issued in December 2024 subsequently formalized eligibility criteria in Emilia-Romagna, including poor glycemic control, problematic hypoglycaemia, and high glycemic variability (Regione Emilia-Romagna, 2024). The 2024 regional guidance therefore largely formalized eligibility principles that were already reflected in national clinical recommendations during the main period of AID diffusion in our sample.

Eligibility does not, however, translate automatically into adoption. In practice, prescription reflects individual clinical assessment, including glycemic history, complication burden, prior experience with pump or continuous glucose monitoring technology, patient readiness, and capacity to manage the device (Pintaudi et al., 2026). Our empirical analysis focuses on AUSL Romagna, the largest Local Health Authority in Emilia-Romagna by population served, covering over one million residents. Within this authority, access to AID also requires a formal administrative authorization process accompanied by documentation of patient suitability and training. Only regulated, commercially registered AID systems are eligible for reimbursement; open-source and do-it-yourself systems are excluded (Regione Emilia-Romagna, 2024).

The patients in our study received care at four specialist diabetes clinics within AUSL Romagna, serving the areas of Ravenna, Lugo, Faenza, and Rimini. Patients are followed longitudinally through periodic diabetology visits, HbA1c testing, laboratory monitoring, and complication screening. Recorded AID adoptions in AUSL Romagna are concentrated in the first half of the calendar year, with first-half counts exceeding second-half counts in each year from 2021 onward (Table~\ref{tab:adoption_timing} and Figure~\ref{fig:aid_adoption_timing}). This pattern is particularly pronounced in 2022, when AID diffusion accelerated in this setting (Figure~\ref{fig:aid_penetration_calendar}). Since initiating AID involves clinical assessment, administrative authorization, and patient training, the observed adoption date may reflect both clinical and organizational factors. This timing pattern is descriptive and is not used as a source of exogenous variation. Our identification strategy instead relies on matched risk-set comparisons and the event-time restriction described in Section~\ref{sec:empirical_strategy}.

\subsection{Conceptual framework and empirical predictions}
\label{subsec:conceptual_framework}
Whether medical automation substitutes for or complements clinical care is theoretically ambiguous. In Grossman’s health-production framework (1972), technology change can increase the productivity of health-producing inputs and thereby reduce the resources required to achieve a given health outcome. Acemoglu and Restrepo's (2019) task-based framework offers a complementary perspective. Automation can reassign existing tasks from human input to technology, generating a \emph{displacement effect}, while also creating new tasks in which human judgment retains a comparative advantage, generating a \emph{reinstatement effect}. Applied to healthcare, these mechanisms generate competing predictions for the demand for clinical input. AID may reduce professional involvement by automating routine management tasks, but it may also increase it if the technology creates new needs for interpretation, adjustment, and supervision. Which mechanism dominates is therefore an empirical question.

AID automates part of day-to-day insulin management, including basal-rate adjustments and some correction decisions, thereby reducing the need for clinical input tied to conventional dose management. At the same time, it generates an integrated record of glucose readings, insulin delivery, and algorithmic responses that may require clinical interpretation, treatment adjustment, patient training, and troubleshooting. AID may therefore displace some routine clinical tasks while creating or expanding others in which professional judgment remains important.

To organize these mechanisms, we distinguish between the share of management tasks requiring clinical input and the overall scale of resources devoted to glycemic control. Let $I$ denote the extent of automation and $N-I$ the measure of diabetes-management tasks still performed with clinical input; a rise in $I$ moves existing tasks onto the device, a rise in $N$ adds tasks requiring clinical input. Let $M$ denote total demand for clinical input, $c_H(I,N)$ the unit resource cost of producing glycemic control, $H^{*}$ the desired level of glycemic control, and $r_M$ the opportunity cost of clinical input. The task-production model in Appendix~\ref{app:task_model} yields
\begin{equation}
M
=
(N-I)\frac{c_H(I,N)H^{*}}{r_M}.
\label{eq:clinical_demand_simple}
\end{equation}
Equation~\eqref{eq:clinical_demand_simple} splits clinical demand into two margins: the clinical task content of diabetes management, $N-I$, and a scale term, $c_HH^{*}/r_M$, reflecting total resources spent on control relative to the cost of clinical input.

Let
\begin{equation}
\varepsilon
\equiv
-
\frac{d\ln H^{*}}{d\ln c_H}
>0
\label{eq:control_elasticity_simple}
\end{equation}
denote the elasticity of desired control with respect to its unit cost. This is not a conventional patient price elasticity: under the Italian NHS, chronic patients pay no money price for routine diabetes care. Rather, $\varepsilon$ captures how the desired intensity of glycemic control responds when achieving a given level of control becomes less resource-intensive. 

The scale channel is relevant in this setting. For example, national guidelines recommend that patients spend at least 70 percent of time within the target glucose range (70–180 mg/dL) (Pintaudi et al., 2022), an outcome that can be routinely measured through continuous glucose monitoring. This illustrates how technological change can expand the set of clinically measurable and attainable treatment targets, potentially increasing the desired intensity of glycemic control. We do not infer from this whether $\varepsilon$ exceeds one; the example simply motivates the possibility of a scale response.

Under a constant-elasticity approximation, the change in clinical demand
between the pre-adoption and post-adoption states is
\begin{equation}
\Delta\ln M
=
\underbrace{
\ln\left(
\frac{N_{1}-I_{1}}{N_{0}-I_{0}}
\right)
}_{\text{change in clinical task content}}
+
\underbrace{
(1-\varepsilon)\Delta\ln c_H
}_{\text{scale effect}}
-
\Delta\ln r_M.
\label{eq:net_clinical_demand}
\end{equation}
Because adoption happens inside the same publicly financed clinical system, we treat the short-run opportunity cost $r_M$ as locally fixed in the patient-level event study. The first term in equation~\eqref{eq:net_clinical_demand} carries the displacement-reinstatement trade-off: an outward move in the automation frontier, $I_1>I_0$, shrinks the clinical task content of care, while new interpretation, adjustment, and supervision tasks, $N_1>N_0$, expand it.

The second term works through the scale of control. Write the unit-cost
semi-elasticities as
\begin{equation}
\pi_I
\equiv
\frac{\partial\ln c_H}{\partial I},
\qquad
\pi_N
\equiv
\frac{\partial\ln c_H}{\partial N},
\qquad\text{so that}\qquad
\Delta\ln c_H
\approx
\pi_I\,\Delta I
+
\pi_N\,\Delta N .
\label{eq:cost_elasticities}
\end{equation}
Automation lowers the unit cost of achieving a given level of control,
so Assumption~\ref{ass:cost_advantage}, stated formally in Appendix~\ref{app:task_model},
implies $\pi_I \leq 0$. The sign of $\pi_N$ is not restricted, because it compares the cost of the task entering at the upper end of the task range with the cost of the task retired at its lower end. If demand for control is elastic enough ($\varepsilon>1$), a cost reduction expands the desired level of control by enough to raise demand for the remaining clinical tasks. If $\varepsilon<1$, the same cost reduction instead shrinks the scale component. The model does not sign the total effect on its own.

Equation~\eqref{eq:net_clinical_demand} yields two competing predictions.
If AID mostly reallocates existing management tasks to the device, with little new clinical work and no strong scale response, routine clinical contact should fall or hold steady---the substitution hypothesis.
If the interpretation, adjustment, training, and supervision tasks AID creates are large enough, routine clinical contact should rise---the reinstatement hypothesis.

Two features of the institutional setting may reinforce complementarity.
The prescribing diabetologist keeps responsibility for monitoring treatment
and device use after adoption, and chronic patients face no direct monetary
price for routine diabetes care, though time and attendance costs remain.
Neither feature determines the sign of the utilization effect on its own,
but both push toward integrating AID into an intensive outpatient care pathway rather than away from it.

The timing of the response carries information too. Some new tasks appear
immediately at adoption, device setup and patient training among them.
Others grow more important gradually, as clinicians and patients build up
experience with the device and use repeated data histories to guide
treatment. We represent this as
\begin{equation}
N_m
=
N_0
+
\Delta N^{\mathrm{impact}}
+
\Delta N^{\mathrm{learning}}
\left(1-e^{-\Psi m}\right),
\qquad
m\geq 0,
\quad
\Psi>0,
\label{eq:new_tasks_dynamics_simple}
\end{equation}
where $\Delta N^{\mathrm{impact}}$ denotes tasks that appear at initiation
and $\Delta N^{\mathrm{learning}}>0$ denotes tasks that build up through
repeated use. This formulation assumes the device remains in place over the
post-adoption window, since the learning term describes clinicians
accumulating experience with a given patient's device record. Where use is
interrupted, the observed path understates the accumulation the mechanism
describes.

\begin{corollary}[Adoption dynamics]
\label{cor:impact-dynamics} 
Suppose equation~\eqref{eq:new_tasks_dynamics_simple} holds with $\Delta N^{\mathrm{learning}}>0$ and let $\Delta I>0$, so that automation occurs at adoption while the complementary clinical tasks generated by AID develop gradually with experience. If the net effect of task creation on visit demand is positive and task creation at initiation is small relative to displacement, then the causal path implied by Proposition~\ref{prop:net} is flat or negative at $m=0$ and rises monotonically toward its long-run level, with $\Psi$ governing how quickly clinicians learn to convert the device record into treatment changes. Appendix~\ref{app:task_model} states these two sufficient conditions, (C1) and (C2), formally.
\end{corollary}

The corollary gives the event study a second prediction, sharper than the sign of the long-run effect: a coefficient path that starts near zero, climbs over the first several post-adoption semesters, and flattens once learning is exhausted. The prediction is conditional. Equation~\eqref{eq:new_tasks_dynamics_simple} fixes only the timing of task creation; the sign at impact and the direction of the path depend on the relative size of task creation and displacement at initiation, which the data cannot separate. A path matching the predicted shape is evidence for the joint hypothesis of reinstatement with gradual learning under the conditions of Appendix~\ref{app:task_model}, while a path departing from it does not by itself refute reinstatement, since it could instead indicate that task creation at initiation was large relative to displacement.

The predicted dynamic path is not unique to the reinstatement mechanism. A gradually increasing post-adoption effect is consistent with the accumulation of complementary clinical tasks, but the shape alone cannot distinguish a path generated by this mechanism from one generated by a time-varying confound evolving around adoption; Appendix~\ref{app:identification} states this formally as a non-identification result. Section~\ref{subsec:trend_extrapolation} develops the restriction we use to separate the two, and Section~\ref{subsec:diagnostics} reports diagnostics and sensitivity analyses quantifying how much departure from that restriction the results can tolerate.

Our empirical outcomes measure diabetologist contact, HbA1c measurement,
and a broader index of recommended diabetes monitoring. These are observed
measures of routine outpatient engagement rather than separate structural
measures of individual tasks. None of them directly measures the desired
level of control, $H^{*}$, or the scale elasticity, $\varepsilon$. HbA1c
testing records whether a measurement was taken during the semester, not
the value obtained or the target toward which the care plan is directed. It
therefore enters the model on the same side as the visit outcome, as a
clinical input into a task above the automation frontier.

What distinguishes HbA1c testing from the visit outcome is its position within the task range. HbA1c testing is not itself a task on the automation frontier, since AID does not determine whether the test is ordered. Rather, it is a complementary monitoring input that clinicians may use to assess glycemic control and treatment effectiveness alongside the information generated by the device. Laboratory demand may therefore increase with $N$ if AID is accompanied by a more intensive clinical monitoring pathway, even though testing is neither automated nor itself a reinstated task. The composite monitoring index captures whether adoption comes with a broader intensification of diabetes care rather than a change along any single margin. The event study estimates the reduced-form effect of AID
adoption on these outcomes. It does not separately identify the changes in $I$ and $N$, the scale elasticity $\varepsilon$, or the effect of adoption on the unit cost $c_H$. What it provides is an economic interpretation of the sign and timing of the response.

\section{Data}
\label{sec:Data}
We use linked administrative and clinical data for patients with type~1 diabetes followed by the four diabetes clinics of Ravenna, Lugo, Faenza, and Rimini. The clinics serve catchment areas covering approximately 730,000 residents, about two-thirds of the population served by AUSL Romagna. The data combine three sources: the Meteda electronic medical record system, the Auxilium diabetes-device registry, and the population registry. 

Meteda provides longitudinal clinical records since 2016, including specialist visits, laboratory tests, treatments, and diabetes complications and comorbidities. Auxilium records reimbursed diabetes-device supplies and is used to identify the first observed semester of AID-related device provision. Because Auxilium records device categories rather than individual commercial models or algorithm activation, treatment should be interpreted as the first observed supply of an AID-related integrated pump system rather than direct observation of closed-loop activation. The linked data separately identify conventional pump and other diabetes-device use. Importantly, the AID-related indicator follows a markedly different diffusion pattern from diabetes-device use more generally: broad device coverage expanded rapidly during 2020–2021, whereas AID-related adoption increased later, particularly from 2022 onward (Table~\ref{tab:a6a_sample_device_trends}). Only two of the 283 first observed AID-related supplies occur before 2021 (Table~\ref{tab:adoption_timing}). These patterns support the interpretation of the treatment indicator as a distinct technology transition, although some residual treatment misclassification cannot be ruled out.

The clinical and device records are linked to the population registry, which provides year of birth, sex, district of residence, and mortality information. We restrict the sample to patients with type~1 diabetes who were alive on January 1, 2020, were followed by one of the four diabetology services, and had linkable records across the three data sources. The resulting sample includes 1,608 patients, of whom 283 are observed adopting AID and 1,325 never adopt during the study period; Table~\ref{tab:sample_construction} reports the corresponding sample-construction details.

We reorganize the underlying clinical records into a patient-semester panel spanning 2016--2025, yielding 28,056 active patient-semester observations (Table~\ref{tab:sample_construction}). We aggregate the data at the semester level to balance temporal precision around adoption against the relatively low frequency of many routine monitoring activities. The final semester, 2025H2, is potentially incomplete and is therefore excluded from the estimation window.

The treatment is adoption of an automated insulin delivery system. As described in Section~\ref{subsec:background}, AID adoption became substantially more frequent from 2021 onward and was particularly concentrated in 2022H1 (Table~\ref{tab:adoption_timing}, Figures~\ref{fig:aid_adoption_timing} and~\ref{fig:aid_penetration_calendar}). Among AID adopters, median observed follow-up after adoption is 2.8 years (interquartile range: 0.9–3.5), while approximately one quarter have less than one year of post-adoption observation.

The outcomes measure routine outpatient engagement with diabetes care.\label{subsec:outcomes} Our primary outcome is an indicator for at least one recorded diabetologist visit during the semester, observed in 59.8 percent of the active patient-semesters. The second is an indicator for at least one HbA1c measurement, observed in 49.0 percent of active semesters. We also build a six-component process-of-care index intended to capture the set of routine examinations and specialist visits included in the diabetes clinical care pathway. The index counts recorded HbA1c, LDL cholesterol, albuminuria, kidney-function testing, eye examination, and diabetologist visits. The index averages 2.959 activities per semester and is intended to capture broader monitoring intensity without placing undue weight on any single, relatively infrequent process. Table~\ref{tab:main_outcomes} provides the complete outcome definition and descriptive means. 

\section{Empirical strategy}
\label{sec:empirical_strategy}
AID adoption is clinically targeted and staggered over time, creating two main challenges for causal identification. First, patients who adopt AID may differ from untreated patients not only in observed clinical characteristics, but also in the trajectory of disease management leading up to adoption. In diabetes care, glycemic control and healthcare utilization may be jointly driven by dynamic, time-varying unobserved heterogeneity, so ignoring this endogeneity can bias estimated relationships between the two (Gil, Li Donni, and Zucchelli, 2019). Second, because patients adopt at different dates, conventional two-way fixed-effects event-study estimators may be biased when treatment effects vary across cohorts and over time.

Our empirical strategy addresses these challenges in three steps. First, within each adoption risk set, we use coarsened exact matching to compare adopters with patients who remain untreated and have similar pre-adoption clinical histories. Second, we estimate cohort-specific dynamic treatment effects using an interaction-weighted event-study design that relies on not-yet-treated and never-treated patients as controls and aggregates cohort-specific effects with nonnegative weights. Third, because matching need not eliminate differential pre-adoption trajectories, we adjust the event-study estimates by extrapolating the pre-adoption trend in event time under an explicit identifying restriction. We report both the raw matched event-study estimates and their trend-adjusted counterparts.

The resulting estimand is a dynamic average treatment effect on the treated for adopters within the matched common support. Identification therefore relies on two distinct ingredients: comparability in observed pre-adoption histories within each risk set and a restriction on the evolution of residual time-varying confounding. Sections~\ref{subsec:model_identification}--~\ref{subsec:trend_extrapolation} describe the estimand, matching procedure, and trend-extrapolation assumption in turn.

\subsection{Estimand and event-study design}
\label{subsec:model_identification}
Let $E_i$ denote the semester of patient $i$'s first observed AID-related device supply, with $E_i=\infty$ for patients who do not adopt during the observation window. As discussed in Section~\ref{sec:Data}, this records device provision rather than confirmed closed-loop activation, so $E_i$ should be read as the first semester in which patient $i$ is observed to have received the technology. The estimation window ends in 2025H1. We exclude 2025H2 because process and device information for that semester is administratively truncated; patients first supplied in 2025H2 do not form an adoption cohort, but they remain available as not-yet-treated controls for earlier cohorts.

Let $g$ index adoption cohorts, so that cohort $g$ contains the patients with $E_i=g$, and define $T_{ig}=\mathbf 1\{E_i=g\}$. Let $\mathcal R_g$ denote the matched risk set for cohort $g$, consisting of those adopters together with their matched not-yet-treated and never-treated controls, constructed as described in Section~\ref{subsec:riskset_matching}. Later-treated controls contribute observations only while they remain untreated, so the risk-set observation set is
$$
\mathcal O_g=\bigl\{(i,t): i\in\mathcal R_g,\; T_{ig}=1 \text{ or } t<E_i\bigr\},
$$
where $t<E_i$ holds automatically for never-treated patients. Let $\mathcal G$ collect the cohorts with a nonempty matched risk set.

Event time for cohort $g$ is $m=t-g$, and we use a window of four semesters on either side of adoption. For $-4\le m\le 4$ define $D^m_{itg}=T_{ig}\mathbf 1\{t-g=m\}$, with binned endpoints
$$
D^{-5}_{itg}=T_{ig}\mathbf 1\{t-g\le -5\},
\qquad
D^{5}_{itg}=T_{ig}\mathbf 1\{t-g\ge 5\}.
$$
The semester immediately preceding adoption, $m=-1$, is the reference period and is omitted.

The risk-set observation sets are stacked, so a patient belonging to several risk sets appears once in each. On the stacked panel we estimate
\begin{equation}
y_{itg}
=
\alpha_{ig}+\gamma_t
+\sum_{\substack{m=-5\\ m\neq-1}}^{5}
\beta_{gm}\,D^m_{itg}
+C_{it}+u_{itg},
\qquad (i,t)\in\mathcal O_g,\ g\in\mathcal G,
\label{eq:eventstudy}
\end{equation}
weighted by the CEM weights $\omega_{ig}$ of equation~\eqref{eq:cemweights}. Here $\alpha_{ig}$ is a patient-by-risk-set fixed effect, so each copy of a patient carries its own time-invariant level, while $\gamma_t$ is a calendar-semester fixed effect common across risk sets. Because $D^m_{itg}$ is indexed by the risk set $g$ to which observation $(i,t)$ belongs, it equals one only for cohort-$g$ adopters in their own risk set. Every control copy in $\mathcal R_g$, including a later adopter serving as a not-yet-treated control, has $T_{ig}=0$ and hence $D^m_{itg}=0$ for all $m$.

Equation~\eqref{eq:eventstudy} is therefore a stacked event-study regression estimated with the interaction-weighted estimator of Sun and Abraham (2021). The cohort-by-event-time coefficients $\beta_{gm}$ are estimated jointly by CEM-weighted least squares and then aggregated across cohorts with the share weights in equation~\eqref{eq:aggregation} below. Because treated copies carry an event-time indicator in every semester except the reference semester, whose level is absorbed by $\alpha_{ig}$, the calendar effects $\gamma_t$ are identified from the untreated control copies, pooled across all risk sets. The risk sets are thus not estimated as separate regressions, and the counterfactual trend for each cohort comes from this pooled, CEM-weighted control group. The clean-control structure is preserved by construction: a patient who adopts in a later semester serves as a control for an earlier cohort only while untreated, and observations at and after that patient's own adoption are excluded from the control copy. No post-adoption observation therefore enters as a control.

The coefficient $\beta_{gm}$ is the cohort-$g$ event-time effect relative to the omitted semester $m=-1$. Suppose three conditions hold: no anticipation; parallel trends, meaning that absent adoption the untreated outcomes of matched cohort-$g$ adopters would have evolved, relative to $m=-1$, like those of the CEM-weighted pool of untreated control copies; and no residual time-varying confounding. Then $\beta_{gm}$ identifies the cohort-specific average treatment effect on the treated among matched adopters,
\begin{equation}
\tau_{gm}
=
\mathbb E\bigl[
y_{i,g+m}(g)-y_{i,g+m}(\infty)
\,\big|\,
E_i=g,\ i\in\mathcal M_g
\bigr],
\label{eq:catt}
\end{equation}
where $\mathcal M_g$ is the set of cohort-$g$ adopters retained on matched common support, $y_{it}(g)$ is the potential outcome under first adoption in semester $g$, and $y_{it}(\infty)$ is the potential outcome under no adoption.

We aggregate the cohort-specific coefficients using nonnegative cohort-share weights. Let $\mathcal G_m$ be the set of cohorts with treated support at relative time $m$, and let $N_{gm}$ be the number of matched treated patient-semester observations from cohort $g$ with $D^m_{itg}=1$. At interior event times each treated patient contributes at most one observation, so $N_{gm}$ is the number of matched cohort-$g$ adopters observed at $m$; at the binned endpoints a patient may contribute several. Define
\begin{equation}
\widetilde{w}_{gm}=\frac{N_{gm}}{\sum_{g'\in\mathcal G_m}N_{g'm}},
\qquad
\theta_m=\sum_{g\in\mathcal G_m}\widetilde{w}_{gm}\,\beta_{gm}.
\label{eq:aggregation}
\end{equation}
The weights are nonnegative and sum to one at each $m$, so $\theta_m$ places no negative weight on any cohort-by-period effect. They are distinct from the CEM weights: $\omega_{ig}$ operates within a risk set and determines how matched controls contribute to $\beta_{gm}$, whereas $\widetilde{w}_{gm}$ operates across cohorts after the $\beta_{gm}$ have been estimated. The corresponding causal aggregate is $\tau_m=\sum_{g\in\mathcal G_m}\widetilde{w}_{gm}\tau_{gm}$.

Because the panel ends in 2025H1, $\mathcal G_m$ is not constant across $m$. At $m=0$, all matched adoption cohorts with observed support at that event time can contribute to $\theta_0$, whereas at $m=4$ only cohorts adopting by 2023H1 can contribute to $\theta_4$. The aggregated path therefore tracks a narrowing and increasingly early-adopting mix of cohorts as $m$ grows, rather than a fixed population followed over a fixed horizon. For the two binary outcomes, any recorded visit and HbA1c measurement, equation~\eqref{eq:eventstudy} is a linear probability model and $\tau_m$ is in percentage points. The third outcome, the composite process index, counts HbA1c, eGFR, albuminuria, LDL, eye examination, and visit within the semester; we analyze it as $\log(1+\text{count})$ to retain zero-activity semesters. Because adoption changes the share of zeros, the magnitude and possibly the sign of an effect on this scale depend on the added constant (Chen and Roth, 2024), so we interpret only its timing. Under the reinstatement mechanism of Section~\ref{subsec:conceptual_framework}, $\{\tau_m: m\ge 0\}$ should be positive and, with learning, should rise over the first several semesters.

The term $C_{it}$ collects latent time-varying confounding. It is useful to write
\[
C_{it}=\lambda_i'F_t+\xi\,\eta_{it},
\]
where $\lambda_i'F_t$ is a low-dimensional common factor with patient-specific loadings and $\eta_{it}$ is a patient-specific component and $\xi$ is a coefficient that measures how much the patient-specific component contributes to the confounding term. Patient-by-risk-set fixed effects absorb permanent differences across patients, and semester fixed effects absorb additive calendar shocks common to all patients in a semester. Neither absorbs an interactive factor whose loading $\lambda_i$ varies across patients, nor a component of $\eta_{it}$ that evolves systematically as adoption approaches. In this setting the route to AID runs through individual eligibility review and the administrative queue rather than through an aggregate shock whose intensity differs across patients in a way correlated with adoption timing, which makes the patient-specific channel the more plausible of the two. Our baseline restriction, stated in Section~\ref{subsec:trend_extrapolation}, therefore acts on the event-time path of the aggregate residual bias. Left unrestricted, $C_{it}$ leaves the path $\{\tau_m\}$ unidentified, because a suitably chosen confound rationalizes any observed event study; Appendix~\ref{app:identification} states this formally. Matching improves comparability in observed pre-adoption clinical histories at the time of adoption, while trend extrapolation addresses residual differences that evolve systematically with event time. Both components rely on assumptions motivated by the clinical setting rather than on the event-study design itself.

\subsection{Risk-set matching}
\label{subsec:riskset_matching}
For each adoption semester $g$, we compare treated patients with $E_i=g$ only to patients still untreated at $g$, namely those with $E_i>g$ or $E_i=\infty$, so no already-treated patient enters as a control. AID eligibility is not fixed at baseline. It is produced by a patient's own evolving glycemic history, so the right comparison group for an adopter is the set of patients who are, at that same moment, equally eligible and equally untreated. A single baseline match would treat propensity to adopt as fixed and would compare patients on characteristics measured years before the decision that selected them into treatment.

Within each semester's risk set, treated patients and eligible controls are coarsened-exact-matched on pre-adoption history over $[g-4,\,g-1]$. The matching variables are mean HbA1c, mean eGFR, mean seven-item measured- and covered-process counts (HbA1c, LDL, blood pressure, BMI, albuminuria, eGFR, and eye examination), prior pump use, any complication, and prior acute care, with exact matching on sex and age band. We coarsen into clinically meaningful bins chosen in advance, which fixes the tolerated imbalance ex ante and lets a reader see what comparability means in each matched set (Iacus, King and Porro, 2012). We keep only strata with at least one treated patient and one eligible control, so common support holds by construction. Of the 283 AID adopters observed in the full sample, 20 first adopt in 2025H2 and therefore do not form an adoption cohort because the estimation window ends in 2025H1. Among the remaining 263 adopters, 82 are excluded during the matching procedure: 25 because they lack four active pre-adoption semesters and 57 because they have no eligible untreated control in their matching stratum and therefore fall outside common support. The remaining 181 adopters (68.8\% of the 263 within-window adopters) are retained in the matched risk sets and contribute to the matched event-study estimates. The estimand is therefore an average treatment effect on the treated defined over adopters who satisfy the required pre-adoption history and whose pre-adoption characteristics have support among eligible untreated patients in the same risk set, rather than over the full population of AID adopters. Because retention depends on both satisfying the pre-adoption history requirement and finding an eligible untreated match on sex, age band, and pre-adoption process history, the estimates apply to the matched adopter population and should not be extrapolated to adopters excluded from that population.

Within a retained stratum $s$ of cohort $g$, each treated patient receives weight one and each matched control the treated-to-control ratio,
\begin{equation}
\omega_{ig}
=
\begin{cases}
1,
&
\text{if } E_i=g,
\\[0.6em]
\dfrac{n^T_{s(i,g)}}{n^C_{s(i,g)}},
&
\text{if } E_i>g \text{ or } E_i=\infty,
\end{cases}
\label{eq:cemweights}
\end{equation}
where $n^T_{s(i,g)}$ and $n^C_{s(i,g)}$ are the treated and control counts in patient $i$'s stratum for cohort $g$. These weights set the reweighted control distribution of the coarsened covariates equal to the treated distribution inside each risk set. The identifying content is that, after conditioning on the risk set and the matched history, any remaining difference in the level of latent severity between adopters and controls is unrelated to adoption status. This does not require adopters and controls to be identical.

We estimate the event study on the matched, CEM-reweighted, stacked panel. Patients who adopt AID in later semesters may serve as controls for earlier adoption cohorts only until their own adoption. Consequently, the same patient may appear in multiple cohort-specific comparisons. We therefore cluster standard errors at the patient level, allowing for dependence across all observations contributed by the same patient. The CEM weights ($\omega_{ig}$) are constructed before estimation and treated as fixed. Accordingly, the reported confidence intervals capture sampling uncertainty in the event-study estimation conditional on the matched sample and weights, but they do not incorporate additional uncertainty arising from the matching procedure.

\subsection{Trend extrapolation and identifying assumptions}
\label{subsec:trend_extrapolation}

Risk-set CEM substantially improves covariate balance (Table~\ref{tab:cem_balance}). The mean absolute standardized difference falls from 0.544 to 0.013 for the exact-match variables (Panel~A) and from 0.208 to 0.043 for the coarsened clinical-history variables that define the matching strata (Panel~B). CEM does not require balance in the continuous variables that underlie the coarsened categories, and the standardized difference in mean pre-adoption HbA1c is essentially unchanged ($-0.125$ before matching, $-0.130$ after). Figure~\ref{fig:loveplot_cem} summarizes these changes. Panel~C reports baseline-semester levels of the outcome variables, which are not matching targets. Their mean absolute standardized difference falls only from 0.093 to 0.085. This does not violate the matching criterion. Level differences of this kind are absorbed by the patient-by-risk-set fixed effects in \eqref{eq:eventstudy}. The threat to identification is a difference in how outcomes evolve before adoption. Balance in levels cannot rule that out, and the raw pre-adoption event-study coefficients reported in Section~\ref{sec:results} show that it is present. Because AID eligibility evolves with a patient's clinical history, we allow the remaining confounding component to vary systematically with event time and restrict that variation to a smooth, low-dimensional path.

Let $\beta_{gm}$ denote the population coefficient on $D^m_{itg}$ in \eqref{eq:eventstudy}, and write $\beta_{gm}=\tau_{gm}+b_{gm}$, where $b_{gm}$ is the residual bias at event time $m$, measured, like $\beta_{gm}$, relative to $m=-1$. Aggregating with the cohort shares $\widetilde{w}_{gm}$ used to construct $\theta_m$ gives
\begin{equation}
\theta_m=\tau_m+c(m),
\qquad
c(m)=\sum_{g\in\mathcal G_m}\widetilde{w}_{gm}\,b_{gm}.
\label{eq:decomposition}
\end{equation}
The normalization sets $\beta_{g,-1}=0$. Under no anticipation, discussed below, $\tau_{g,-1}=0$ as well, so $b_{g,-1}=0$ for every cohort and $c(-1)=0$. Any restriction on the confounding path must respect this normalization, and we therefore state it with a reference-normalized basis, $\widetilde f(m)=f(m)-f(-1)$.

\begin{assumption}[Event-time extrapolation]
\label{ass:event_time_extrapolation}
After risk-set matching and conditioning on the patient-by-risk-set and calendar-semester fixed effects in \eqref{eq:eventstudy}, the aggregate residual confounding component follows
\begin{equation}
c(m)=\phi'\widetilde f(m),
\qquad
\widetilde f(m)=f(m)-f(-1),
\label{eq:confoundingpath}
\end{equation}
over the event-time window used for extrapolation, where $f(\cdot)$ is a known low-dimensional basis.
\end{assumption}

The assumption restricts the aggregate path, and the set of contributing cohorts $\mathcal G_m$ narrows as $m$ increases. Suppose each cohort's bias is linear with its own slope, $b_{gm}=\rho_g(m+1)$. Then $c(m)=(m+1)\sum_{g\in\mathcal G_m}\widetilde w_{gm}\rho_g$, which is linear in $m$ only if the share-weighted slope does not change as later cohorts leave the post-adoption window. Assumption~\ref{ass:event_time_extrapolation} therefore also requires that the cohorts observed at long horizons do not have systematically different pre-adoption slopes from the cohorts that drop out.

The baseline specification is linear in event time. Taking $f(m)=m$ gives $c(m)=\rho(m+1)$ for a scalar slope $\rho$, which satisfies $c(-1)=0$ by construction. Under no anticipation, $\tau_m=0$ for $m<0$, so equation~\eqref{eq:decomposition} implies $\theta_m=c(m)$ in the pre-adoption periods, and the pre-adoption coefficients identify $\rho$. We fit the slope on the coefficients at $m=-3$ and $m=-2$, with the line constrained to pass through zero at $m=-1$, so that the fitted slope reflects the trajectory closest to adoption. The coefficient at $m=-4$ lies outside the fitting window and serves as a held-out check on the fitted line. We then extrapolate the line into the post-adoption window. The trend-adjusted effect is
\begin{equation}
\widehat\tau_m=\widehat\theta_m-\widehat c(m),
\qquad m\ge 0.
\label{eq:trendadjusted}
\end{equation}
The event-study coefficients and the slope are estimated jointly by GMM, following the trend-extrapolation estimator of Dobkin et al.\ (2018) as implemented in the \texttt{xtevent} package (Freyaldenhoven et al., 2025b). Joint estimation carries the sampling uncertainty in $\widehat\rho$ into $\widehat\tau_m$. Standard errors are clustered by patient, which accounts for correlation across the stacked copies of patients who appear in several risk sets. Appendix~\ref{app:identification} gives the formal derivation.

Assumption~\ref{ass:event_time_extrapolation} is an identifying restriction, not a consequence of the event-study design or of the institutional setting. Neither matching nor fixed effects guarantee that residual untreated differences would have continued linearly after adoption. Institutional features can motivate the restriction, but they cannot establish that the extrapolated counterfactual is correct.

Nevertheless, the clinical context may help explain the observed pre-adoption pattern in routine diabetology contact. AID prescribing responds to patients’ evolving clinical histories, which may also affect the timing and type of healthcare contacts before adoption. The pre-adoption coefficients suggest declining routine diabetology contact relative to matched controls, although the pattern is not uniformly monotone. This need not imply a decline in overall healthcare use: it could reflect changes in the composition of care across settings or providers. However, our analysis does not establish whether reduced diabetology contact is offset by greater use of other services. Such a shift should therefore be regarded as a possible interpretation of the observed pattern, rather than evidence validating the extrapolated counterfactual.

For identification, the relevant distinction is between a persistent pre-adoption process and a transitory, mean-reverting disturbance. Trend extrapolation is less credible when treatment follows a temporary shock that would have reversed without treatment, as in an Ashenfelter-type dip (Ashenfelter, 1978; Heckman and Smith, 1999). The direction of the resulting bias is known here. For visits, the pre-adoption gap falls toward the reference period, and the fitted line continues that decline after adoption. If the decline were transitory, the true counterfactual would be flatter than the fitted line or would rebound, and the trend-adjusted estimates would overstate the effect of adoption. Eligibility is assessed on accumulated clinical history, which gives some reason to expect a persistent process, but it does not rule out mean reversion. We therefore treat linear extrapolation as a maintained restriction and assess departures from it directly in Section~\ref{subsec:diagnostics}.

The decomposition in \eqref{eq:decomposition} also requires no anticipation. Behavioral anticipation is likely limited for routine care, because adoption follows a clinical assessment of past glycemic history. Patients who pass eligibility review may nonetheless know that a device is coming while they wait in the administrative queue, so anticipation cannot be excluded on institutional grounds alone. The most direct threat is administrative onboarding. Training and documentation before initiation may raise visits or tests in the semester before recorded adoption. Because that semester is the reference period, onboarding is a violation of no anticipation at $m=-1$ that shifts every coefficient and enters the fitted trend through the normalization. Appendix~\ref{app:admin_effect} shows that, under the baseline fitting window and provided that \(\kappa>0\), an onboarding effect concentrated in the reference period biases the trend-adjusted estimates downward, with the bias increasing in absolute magnitude at longer horizons. We assess it by re-estimating the raw event study normalized to $m=-4$ and comparing the periods immediately before adoption with a path anchored further back in event time. A separate timing issue arises because the data record device supply rather than confirmed closed-loop activation. Actual activation may therefore occur after the recorded supply semester, although the data do not allow us to determine the exact timing of activation relative to the recorded supply date. If activation is delayed for some patients, the recorded adoption indicator may assign treatment before effective exposure begins. Under this interpretation, early post-adoption estimates could be attenuated and the estimated path could appear to rise gradually as a larger share of patients become actively exposed. We therefore cannot rule out delayed activation as one explanation for part of the observed post-adoption dynamics.

We assess the extrapolation restriction with three diagnostics. First, the raw pre-trends test asks whether the unadjusted pre-adoption coefficients are jointly zero. Flat pre-trends are compatible with Assumption~\ref{ass:event_time_extrapolation}, as the case $\rho=0$, but they are not required, and passing this test would not by itself validate the post-adoption extrapolation. Second, the adjusted joint pre-period test asks whether the trend-adjusted coefficients at $m=-4$, $-3$, and $-2$ are jointly zero, and bears more directly on linearity. Because one slope is fitted to the two coefficients inside the fitting window, the test has little power against curvature within that window, and much of its content comes from the held-out coefficient at $m=-4$. Third, following Rambachan and Roth (2023), we bound the change in the slope of the differential trend and report how large that change can be before the post-adoption conclusions change. Their smoothness restriction nests exact linearity at $M=0$. It is applied to the raw event-study coefficients over the same pre-adoption window as the baseline trend fit, the coefficients at $m=-3$ and $m=-2$ together with the normalization at $m=-1$, so at $M=0$ it imposes the linear restriction of Assumption~\ref{ass:event_time_extrapolation}. We also report the re-normalization exercise described above. Section~\ref{subsec:diagnostics} reports the results, and Appendix~\ref{app:diagnostic_tests} gives the formal definitions.

The leveling-off test, which asks whether the post-adoption path stabilizes, is conceptually distinct from these identification diagnostics. It concerns the dynamic prediction of the task-based model rather than the validity of the extrapolated counterfactual, and we discuss it with the post-adoption results in Section~\ref{sec:results}.

\section{Results}
\label{sec:results}
We first report the raw matched event-study estimates and then the estimates obtained after adjusting for differential pre-adoption trends. Our primary outcome is the probability of a recorded diabetologist visit;  HbA1c measurement and the composite process-of-care index provide complementary evidence on broader outpatient monitoring. Table~\ref{tab:aid-sa-raw-process-2025h1} and Table~\ref{tab:aid-sa-trend-process-2025h1} report the raw and trend-adjusted estimates, respectively, while Figures~\ref{fig:aid_visit_eventstudy} through~\ref{fig:aid_process_index_eventstudy} show the corresponding dynamic paths.

Our main finding is that, conditional on the trend-extrapolation assumption, AID adoption is followed by a substantial increase in specialist contact. The estimated effect on the probability of at least one diabetologist visit is small and statistically insignificant at adoption, rises to 16.3 percentage points one semester later (95\% CI: 2.2 to 30.4) and reaches 39.3 percentage points four semesters after adoption (95\% CI: 12.1 to 66.5). The joint test of post-adoption coefficients rejects the null hypothesis of no effect (p = 0.044). HbA1c measurement moves broadly in the same direction, although the evidence is less robust, while estimates for the composite process-of-care index are positive but imprecise. 

The raw estimates in Table~\ref{tab:aid-sa-raw-process-2025h1} reveal substantial pre-adoption differences between future adopters and matched controls. For visits, the joint pre-period test rejects flat pre-trends (p = 0.001): future adopters were 9–15 percentage points more likely to have had a visit in the three pre-adoption semesters ($m=-4,-3,-2$). Although the pattern is not fully monotonic, it indicates that future adopters followed a different utilization trajectory before AID initiation.  The raw post-adoption estimates therefore cannot be interpreted under conventional parallel-trends assumptions and motivate the trend-adjustment strategy described in Section~\ref{subsec:trend_extrapolation}.

The raw pre-adoption patterns are less pronounced for the other two outcomes. Neither HbA1c measurement ($p=0.158$) nor the composite process index ($p=0.435$) rejects the null of flat pre-trends, although the limited number of pre-adoption coefficients reduces the power of these tests (Roth, 2022). The two outcomes, however, differ importantly once the trend adjustment is applied. For the composite index, the linear specification absorbs the observed pre-adoption dynamics well, and the adjusted joint pre-period test does not reject ($p=0.834$). For HbA1c, by contrast, the adjusted pre-period test rejects strongly ($p < 0.001$), indicating that the linear extrapolation does not adequately capture its pre-adoption path. We therefore treat the trend-adjusted evidence for HbA1c with particular caution.

The raw post-adoption estimates provide no clear evidence of a systematic increase in routine care. The joint post-period test does not reject at the 5\% level for any of the three outcomes ($p=0.269$, $0.053$, and $0.338$, respectively), although the result for HbA1c is close to conventional significance levels. As shown in Panel~A, the estimates are generally imprecise and the post-adoption paths are not uniform across outcomes. Together with the pre-adoption dynamics documented above, this reinforces the need for the trend-adjusted analysis.

Table~\ref{tab:aid-sa-trend-process-2025h1} removes the linear pre-adoption trajectory under Assumption~\ref{ass:event_time_extrapolation}, reporting each coefficient as a deviation from its extrapolation; Panel~B of each figure plots the fitted pre-adoption trend, and Panel~C plots the corresponding trend-adjusted estimates. For visits, the adjustment transforms the picture: the contemporaneous effect is statistically indistinguishable from zero and then rises monotonically to 39.3 percentage points by $m=4$ (Table~\ref{tab:aid-sa-trend-process-2025h1}), with the joint post-period test rejecting at the 5\% level ($p=0.044$). Relative to the 55.2\% reference mean, the estimate at $m=4$ corresponds to an increase of approximately 71\%. The linear trend does not fully absorb the pre-adoption dynamics in the visit outcome. The adjusted joint pre-period test still rejects a flat residual ($p=0.027$, versus $p=0.001$ raw), with the largest remaining deviation occurring at $m=-2$. The raw path is itself non-monotonic, so a linear specification provides only an approximate fit. This limitation is particularly important for visits, since the estimated effect at $m=4$ increases from 0.123 in the raw event study to an adjusted 0.393. The main visit result therefore relies most heavily on the extrapolation restriction, and Section~\ref{subsec:diagnostics} assesses how sensitive it is to departures from linearity.

In the baseline trend-adjusted specification, HbA1c measurement also shows positive post-adoption estimates, reaching 15.9 percentage points at $m=4$, relative to a 29.8\% reference-period mean. The joint post-adoption test rejects the null of no effect at the 5\% level ($p=0.043$). However, as noted above, the adjusted pre-period test strongly rejects the linear specification, and Section~\ref{subsec:diagnostics} shows that the HbA1c estimates are sensitive to alternative pre-trend choices. We therefore interpret this evidence cautiously.

The composite process-of-care index is imprecisely estimated. The trend-adjusted coefficient at $m=4$ is 0.262 on the $\log(1+\text{count})$ scale, no individual post-adoption coefficient is statistically significant, and the joint post-period test does not reject ($p=0.265$). Given the scale-dependence discussed in Section 4.1, we do not assign this outcome a magnitude and use it only to examine the timing of the response.

Corollary~\ref{cor:impact-dynamics} predicts a dynamic pattern rather than only the sign of the effect: if automation occurs at adoption while complementary clinical tasks develop gradually with experience, the effect should initially be small and increase over subsequent semesters before eventually stabilizing. 

Conditional on the trend-adjustment specification, the visit estimates are consistent with this prediction: the effect is small and statistically insignificant at adoption and increases monotonically through $m=4$. The composite index is broadly consistent with the same pattern but is imprecisely estimated, while the HbA1c path is less consistent with the predicted dynamics. As emphasized in Section~\ref{subsec:conceptual_framework}, however, the shape of the event-study path alone cannot identify the underlying mechanism. Appendix~\ref{app:identification} discusses formally how a one-time onboarding effect around treatment initiation would affect the trend-adjusted estimates.

The leveling-off test provides no clear evidence that the post-adoption response has stabilized within the observed window. For visits, the comparison between early and later post-adoption effects yields $p=0.090$, while the corresponding tests for HbA1c and the composite index yield $p=0.280$ and $p=0.590$, respectively. Confidence intervals also widen at longer horizons, limiting the precision of the estimates. We therefore interpret the estimates at $m=4$ as medium-run effects and do not infer from them that the response has reached its long-run level.
\subsection{Diagnostics, sensitivity, and robustness checks}
\label{subsec:diagnostics}
\textit{Testing the reference-period convention.} Section~\ref{subsec:trend_extrapolation} proposes a re-normalization check to assess whether an onboarding effect concentrated around the reference period affects the extrapolation; Appendix~\ref{app:admin_effect} derives the corresponding cancellation result. Re-referencing the event study to $m=-4$ allows the coefficients at $m=-2$ and $m=-1$ to be assessed relative to a trend fitted on earlier pre-adoption periods. Figure~\ref{fig:renormalization_diagnostic} shows no significant departure at $m=-2$ for any outcome, while at $m=-1$ only HbA1c departs marginally from the fitted line ($p=0.080$). For visits and the composite index, the alternative pre-trend window preserves the sign of the post-adoption estimates, although statistical evidence weakens somewhat. For HbA1c, by contrast, the alternative extrapolation reverses the sign of the post-adoption estimates and yields a statistically significant negative joint effect ($p=0.034$). This reinforces the conclusion that the HbA1c results are highly sensitive to the choice of pre-trend specification, consistent with the strong rejection of the adjusted pre-period test in the baseline model ($p<0.001$, Table~\ref{tab:aid-sa-trend-process-2025h1}).  

\textit{Sensitivity to the linearity assumption.} Section~\ref{sec:empirical_strategy} treats linear extrapolation as a maintained assumption rather than a property guaranteed by the design. Following Rambachan and Roth (2023), we construct robust confidence sets from the raw Sun--Abraham estimates (2021) under their smoothness restriction, which bounds by $M$ the change in the slope of the differential trend between consecutive semesters. The restriction uses the pre-adoption coefficients at $m=-3$ and $m=-2$ and the normalization at $m=-1$, the same window as the baseline trend fit, and the target parameter is the average effect over $m=0,\ldots,4$. At $M=0$ the differential trend is exactly linear, which is the restriction in Assumption 1. For diabetologist visits, the confidence set for the average effect excludes zero up to approximately $M=0.015$--$0.020$, so the conclusion survives a change in slope of about 1.5--2.0 percentage points per semester. This small breakdown value, together with the rejected adjusted pre-period test ($p=0.027$), indicates that the visit result is informative under linearity but sensitive to small departures from it. The other outcomes are weaker, for different reasons. For HbA1c, the interval includes zero already at $M=0$ and the adjusted pre-period test strongly rejects linearity ($p<0.001$; Table~\ref{tab:aid-sa-trend-process-2025h1}), so the average effect is indistinguishable from zero and the linear pre-fit is rejected. For the composite index, the interval also includes zero at $M=0$, but the linear pre-period specification is not rejected ($p=0.834$), so the limitation is imprecision rather than misspecification. Conditional on linearity, we therefore give most weight to the visit outcome, on grounds of precision and a nonzero breakdown value despite its weaker pre-fit, and treat the HbA1c and composite-index estimates as suggestive (Figure~\ref{fig:aid_honestdid_three_outcomes}).

\textit{Robustness to the aggregation estimator.}
Figures~\ref{fig:cs_visit}, \ref{fig:cs_hba}, and~\ref{fig:cs_process} plot raw and trend-adjusted Callaway-Sant’Anna paths (2021), using the same risk set CEM sample and extrapolation logic, with the pre-adoption trend fit on $m=-3,-2,-1$ and the normalized reference at $m=-1$. The visit path reproduces the Sun--Abraham (2021) pattern closely: small and statistically insignificant at adoption, then rising monotonically through $m=4$. The HbA1c and composite paths are also larger at $m=4$ than at adoption, but neither rises monotonically; both dip at $m=3$. At event time $m=4$, the Callaway–Sant’Anna estimates are of the same order but do not match exactly: $\approx29$ versus 39.3 points for visits, $\approx 26$ versus 15.9 points for HbA1c, $\approx34$ versus 26.2 log points for the composite index; Callaway-Sant’Anna (2021) intervals are wider throughout, especially for the composite index. The two estimators locate the raw pre-trend problem differently: Sun–Abraham rejects flat pre-trends for visits ($p=0.001$) but not HbA1c ($p=0.158$); Callaway–Sant’Anna shows the reverse ($p=0.500$, $p=0.010$). This reflects how each estimator weights the not-yet-treated comparison group, not an artifact specific to either. Callaway–Sant’Anna (2021) estimates also continue to rise through the window, though leveling-off tests do not reject stabilization ($p=0.470$ visits, $p=0.160$ HbA1c, $p=0.230$ composite), so the tests cannot confirm that the paths have flattened. For visits, the shape predicted by Corollary 1, flat or negative at impact and rising through $m=4$, is reproduced under both estimators, which indicates that the visit result is not an artifact of cohort aggregation. We do not read the agreement as equally informative for the remaining outcomes. The two estimators disagree about where the pre-trend problem lies. For HbA1c, the evidence against the linear counterfactual comes from the Sun--Abraham (2021) adjusted pre-period test, which rejects at $p<0.001$. Under Callaway--Sant'Anna (2021), the raw test rejects flat pre-trends ($p=0.010$), but rejecting a flat trend does not by itself reject a linear one. The quadratic exercise reported below does, however, find significant curvature in the HbA1c pre-adoption path under Callaway--Sant'Anna ($p=0.003$), which bears directly on linearity. The similarity of the two HbA1c panels is therefore not independent evidence that the predicted dynamics hold for that outcome.

\textit{Robustness to nonlinear trend extrapolation.} A concern with linear extrapolation is that the true pre-adoption trend might curve, in which case a straight-line counterfactual would misattribute that curvature to the estimated treatment effect. Figures~\ref{fig:cs_quad_visit},~\ref{fig:cs_quad_hba}, and~\ref{fig:cs_quad_process} examine one such alternative. Starting from the Callaway--Sant'Anna (2021) estimates, we replace the linear counterfactual with a quadratic path fitted over $m=-4,-3,-2$ and constrained to equal zero at the normalized period $m=-1$, so that any curvature in the pre-period is free to propagate into the extrapolated counterfactual after adoption. The curvature term is statistically indistinguishable from zero for visits ($p=0.197$) and the composite index ($p=0.189$), but is significant for HbA1c ($p=0.003$). Unlike the earlier linear approximation, however, the fitted quadratic paths bend sharply downward after the normalization period. Subtracting these extrapolated counterfactuals therefore generates increasingly large positive adjusted coefficients, especially at longer horizons. Two features of the exercise limit what it establishes. The curvature term is identified from only three pre-adoption coefficients fitted to a two-parameter curve, so both the estimated curvature and its continuation beyond the observed pre-period are weakly disciplined by the data. The quadratic specification itself is not rejected by the fitted pre-adoption coefficients ($p=0.187$ for visits, $0.852$ for HbA1c, and $0.949$ for the composite index), but this in-sample fit does not validate its steep post-adoption extrapolation. Indeed, for the binary outcomes, the resulting adjusted effects eventually exceed economically feasible probability changes, indicating that the long-run counterfactual is being driven by polynomial extrapolation rather than credible identifying variation. We therefore read this exercise narrowly: allowing unrestricted quadratic curvature does not provide a reliable alternative counterfactual in this short pre-treatment window. The broader extrapolation concern is better addressed through the Rambachan and Roth (2023) sensitivity analysis reported above, which directly quantifies how much departure from linearity the visit result can tolerate without relying on polynomial continuation alone.
\subsection{Interpretation and implications}
\label{subsec:interpretation}
The estimated increase in specialist contact is large relative to baseline. Conditional on the trend-extrapolation assumption, the probability of a diabetologist visit rises by approximately 30\% relative to the reference-period mean one semester after adoption and by about 70\% by $m=4$. The corresponding estimates for HbA1c measurement are positive in the baseline specification but, as shown above, are substantially less robust to alternative pre-trend assumptions. Throughout, the estimates describe the 181 adopters retained within the matched common support rather than all 263 adopters observed in the panel.

For the health system, the main implication concerns the composition of care. Our finding of increased routine specialist contact, considered alongside previous evidence that AID adoption is associated with lower emergency and inpatient utilization (Manjelievskaia et al., 2025), is consistent with the possibility that AID reallocates healthcare use toward planned outpatient management and away from acute care. Similar changes in the composition of utilization have been documented for other patient-facing and digital technologies (Zeltzer et al., 2023; Rodrigues et al., 2025). Our data do not allow us to assess whether such a reallocation reduces total healthcare spending.

The estimates identify utilization rather than the content or marginal value of additional contacts. We therefore cannot determine whether increased specialist involvement reflects high-value monitoring, interpretation of device-generated information, tighter clinical targets, documentation and renewal requirements, or other forms of care. Nor can we identify which specific tasks account for the additional professional input. The narrower conclusion is that automating part of day-to-day insulin management did not reduce recorded specialist involvement. This pattern is consistent with the task-based reinstatement mechanism discussed in Section~\ref{subsec:conceptual_framework}, but it does not distinguish reinstatement from other mechanisms capable of generating greater clinical engagement after adoption.

\section{Conclusion}
\label{sec:Conclusion}
This paper examines whether patient-facing medical automation substitutes for or complements professional clinical care. Using the adoption of automated insulin delivery among patients with type 1 diabetes, we find that, conditional on the event-time extrapolation assumption, AID adoption is followed by a substantial increase in routine specialist contact. The probability of a recorded diabetologist visit rises gradually after adoption, reaching 39.3 percentage points above the extrapolated counterfactual four semesters later, relative to a 55.2\% reference-period mean. Evidence from the other monitoring outcomes is weaker: HbA1c measurement is sensitive to the specification of the pre-adoption trend, while the composite process-of-care index is estimated imprecisely. The most robust conclusion is therefore that AID adoption does not reduce recorded specialist involvement.

This finding is difficult to reconcile with a simple view of automation as substituting technology for professional input. AID automates part of day-to-day insulin management, but its adoption may simultaneously generate or expand activities in which clinical judgment remains valuable, including the interpretation of device-generated information, treatment adjustment, supervision, and follow-up. Alternatively, by increasing the productivity of glycemic management, AID may increase the desired scale or intensity of care. Our data do not allow us to distinguish between task-based reinstatement and such a scale response. They do show, however, that automating activities previously performed in patients’ day-to-day disease management need not translate into lower demand for professional care.

The interpretation is subject to an important identification qualification. AID adoption is clinically targeted: physicians prescribe the technology in response to an evolving patient-specific clinical trajectory, so the factors that lead to adoption may themselves also affect subsequent healthcare use. Given this selection process, a fully comparable untreated group is unlikely to be available in observational data, even after detailed matching on observed pre-adoption histories. Risk-set matching substantially improves comparability in observed characteristics, but it cannot fully eliminate differential pre-adoption dynamics. The trend-adjusted estimates therefore rely on the additional assumption that residual confounding evolves smoothly in event time. This assumption is particularly consequential for the visit outcome: the linear specification does not fully absorb the pre-adoption path, and the sensitivity analysis shows that the estimated effect tolerates only a relatively limited departure from linearity. We therefore view the visit estimates as informative but assumption-sensitive rather than as unconditional causal effects. Absent random or plausibly exogenous variation in adoption, this identification challenge is inherent to observational evaluations of clinically targeted technologies rather than specific to the matching procedure used here. A further qualification concerns the scope of the estimand rather than its identification. Risk-set matching retains 181 of the 263 adopters in estimable cohorts (68.8\%), so the estimates apply to adopters whose pre-adoption characteristics are supported by contemporaneously eligible, untreated patients. We cannot assess with these data whether the excluded adopters would respond similarly.

Our data also do not reveal the content or marginal value of the additional specialist contacts, nor whether patients remain continuously exposed to AID after initiation. Greater outpatient engagement may reflect clinically valuable interpretation and treatment adjustment, but it may also include documentation, device renewal, or other administrative activities. Similarly, the sign and timing of the utilization response cannot by themselves identify the underlying mechanism. Because adoption is dated by recorded device supply rather than by confirmed activation, the gradual rise in the estimated path cannot be fully separated from staggered activation following that supply. Distinguishing among these explanations would require information on visit content, device continuation, and downstream health outcomes. The findings therefore speak to healthcare utilization rather than welfare: more specialist contact is neither necessarily beneficial nor necessarily inefficient.

These qualifications are important for technology assessment and capacity planning. The resource consequences of AID extend beyond the acquisition cost of the device if its use also requires additional specialist time, monitoring, and follow-up. Such additional outpatient care may be worthwhile if it improves glycemic outcomes or prevents acute events and long-run complications. Indeed, previous evidence suggests that AID adoption may reduce emergency and inpatient utilization (Manjelievskaia et al., 2025). Our results nevertheless indicate that savings in routine specialist care should not be assumed when assessing the overall resource implications of the technology.

More broadly, the effects of medical technology on healthcare use depend on the institutional environment in which it is deployed. As emphasized by Baker (2001), utilization responses reflect not only a technology’s clinical properties but also the incentives and organizational arrangements surrounding its adoption. In a publicly financed system in which routine diabetes care carries little or no point-of-use cost for eligible chronic patients and prescribing specialists retain responsibility for treatment after adoption, patient-facing automation may complement rather than replace professional input. Whether the same pattern extends to remote monitoring, algorithm-supported treatment, and other forms of technology-enabled chronic care remains an open empirical question. For publicly financed health systems, the broader implication is that medical automation may change the composition of clinical work and increase demands on specialist capacity even when it successfully automates tasks performed by patients.

\section*{Funding}

This research was co-funded by the Italian Ministry of University and Research (MUR) under the Complementary National Plan to the National Recovery and Resilience Plan (PNC-I.1), ``Research initiatives for innovative technologies and pathways in the health and welfare sector'', within the initiative DARE -- Digital Lifelong Prevention (Project code PNC0000002).

\section*{Acknowledgements}

The research is the result of a joint project of the Local Health Authority (LHA) of Romagna and the Department of Economics, University of Bologna on the Evaluation and Monitoring of Health Policies. The authors gratefully acknowledge Paolo Di Bartolo, Head of the Unit ``Diabetology'', and Roberto Grilli, Head of the Unit ``Evaluative Research and Health Services Policy'' of the LHA for their valuable contributions to the design of the research project within which this study was developed, and Simona Rosa, School of Health Policy, University of Bologna, for her assistance with data extraction. The views expressed remain exclusively those of the authors.

\section*{Declaration of competing interest}

The authors declare that they have no competing interests.

\section*{Data availability}

The data used in this study were made available by the Romagna Local Health Authority under specific data-sharing agreements and are not publicly available because of privacy and data-protection restrictions. The agreements do not permit the authors to share or publicly release the underlying individual-level data. Aggregated data supporting the findings are available from the corresponding author upon reasonable request. Interested parties may request the data directly from the data controller subject to its eligibility requirements, approval procedures, and applicable data-protection legislation.

\section*{Compliance with ethical standards}

The study was conducted in accordance with applicable national and regional data-protection regulations and with the principles of the Declaration of Helsinki. It relied on routinely collected administrative, clinical, and device data from the Local Health Authority of Romagna. Clinical and device information recorded in the Meteda and Auxilium systems was collected in the course of routine care in accordance with the applicable consent and data-protection requirements. Before being made available to the authors, administrative data and device information were anonymized by the Local Health Authority, and each patient was assigned a unique study identifier. The researchers had no access to directly identifying information and could not trace the identity of individual patients. The study involved the secondary analysis of anonymized routinely collected healthcare data and did not require any additional intervention or direct contact with patients. On this basis, no additional individual patient consent was obtained for the present analysis, and the study was conducted under the data-governance arrangements applicable to retrospective analyses within the Local Health Authority.

\section*{Declaration of generative AI and AI-assisted technologies in the manuscript preparation process} During the preparation of this work, the authors used OpenAI ChatGPT and Anthropic Claude to assist with language editing, improving clarity and organization, and refining the presentation of the manuscript. After using these tools, the authors reviewed and edited the content as needed and take full responsibility for the content of the publication.

\section*{CRediT authorship contribution statement}

\textbf{Moslem Rashidi:} Conceptualization, Data curation, Formal analysis, Investigation, Methodology, Software, Visualization, Writing -- original draft, Writing -- review and editing.

\textbf{Cristina Ugolini:} Conceptualization, Investigation, Supervision, Writing -- review and editing.

\textbf{Gianluca Fiorentini:} Project administration, Conceptualization, Investigation, Supervision, Writing -- review and editing.

\section*{References}
\begin{description}

\item Acemoglu, D., Restrepo, P., 2019. Automation and new tasks: how technology displaces and reinstates labor. \emph{Journal of Economic Perspectives}, 33, 3--30. \\ \doi{https://doi.org/10.1257/jep.33.2.3}

\item Agarwal, N., Moehring, A., Rajpurkar, P., Salz, T., 2023. Combining human expertise with artificial intelligence: experimental evidence from radiology. NBER Working Paper No.\ 31422. \doi{https://doi.org/10.3386/w31422}

\item Agha, L., 2014. The effects of health information technology on the costs and quality of medical care. \emph{Journal of Health Economics}, 34, 19--30. \\
\doi{https://doi.org/10.1016/j.jhealeco.2013.12.005}

\item Ashenfelter, O., 1978. Estimating the effect of training programs on earnings. \emph{Review of Economics and Statistics}, 60, 47--57. \doi{https://doi.org/10.2307/1924332}

\item Autor, D.H., Levy, F., Murnane, R.J., 2003. The skill content of recent technological change: an empirical exploration. \emph{Quarterly Journal of Economics}, 118, 1279--1333. \doi{https://doi.org/10.1162/003355303322552801}

\item Baicker, K., Mullainathan, S., Schwartzstein, J., 2015. Behavioral hazard in health insurance. \emph{Quarterly Journal of Economics}, 130, 1623--1667. \\
\doi{https://doi.org/10.1093/qje/qjv029}

\item Baker, L.C., 2001. Managed care and technology adoption in health care:
evidence from magnetic resonance imaging. \emph{Journal of Health Economics},
20, 395–-421. \\ \doi{https://doi.org/10.1016/S0167-6296(01)00072-8}

\item B\"ockerman, P., Laine, L.T., Nurminen, M., Saxell, T., 2025a. Information integration, coordination failures, and quality of prescribing. \emph{Journal of Human Resources}, 60, 1055--1093. \doi{https://doi.org/10.3368/jhr.0921-11910R2}

\item B\"ockerman, P., Kortelainen, M., Laine, L.T., Nurminen, M., Saxell, T., 2025b. Information technology, improved access, and use of prescription drugs. \emph{Journal of the European Economic Association}, 23, 396--430. \doi{https://doi.org/10.1093/jeea/jvae034}

\item Burnside, M.J., Lewis, D.M., Crocket, H.R., Meier, R.A., Williman, J.A., Sanders, O.J., Jefferies, C.A., Faherty, A.M., Paul, R.G., Lever, C.S., Price, S.K.J., Frewen, C.M., Jones, S.D., Gunn, T.C., Lampey, C., Wheeler, B.J., de Bock, M.I., 2022. Open-source automated insulin delivery in type 1 diabetes. \emph{New England Journal of Medicine}, 387, 869--881. \doi{https://doi.org/10.1056/NEJMoa2203913}

\item Callaway, B., Sant'Anna, P.H.C., 2021. Difference-in-differences with multiple time periods. \emph{Journal of Econometrics}, 225, 200--230. \doi{https://doi.org/10.1016/j.jeconom.2020.12.001}

\item Chandra, A., Skinner, J., 2012. Technology growth and expenditure growth in health care. \emph{Journal of Economic Literature}, 50, 645--680. \doi{https://doi.org/10.1257/jel.50.3.645}

\item de Chaisemartin, C., D'Haultf\oe{}uille, X., 2020. Two-way fixed effects estimators with heterogeneous treatment effects. \emph{American Economic Review}, 110, 2964--2996. \\
\doi{https://doi.org/10.1257/aer.20181169}

\item Chen, J., and Roth, J., 2024. Logs with Zeros? Some Problems and Solutions. \emph{Quarterly Journal of Economics}, 139, 891--936. \doi{10.1093/qje/qjad054}

\item Dobkin, C., Finkelstein, A., Kluender, R., Notowidigdo, M.J., 2018. The economic consequences of hospital admissions. \emph{American Economic Review}, 108, 308--352. \\
\doi{https://doi.org/10.1257/aer.20161038}

\item Freyaldenhoven, S., Hansen, C., Pérez Pérez, J., and Shapiro, J. M., 2025a. Visualization, Identification, and Estimation in the Linear Panel Event-Study Design. In V. Chernozhukov, J. Hörner, E. La Ferrara, \& I. Werning (eds.), Advances in Economics and Econometrics: Twelfth World Congress, Vol. 2, pp. 225–268. \emph{Cambridge University Press}. \doi{10.1017/9781009589727.011}

\item Freyaldenhoven, S., Hansen, C.B., Pérez, J.P., Shapiro, J.M. and Carreto, C., 2025b. xtevent: Estimation and visualization in the linear panel event-study design. \emph{The Stata Journal}, 25, pp.97--135. \doi{10.1177/1536867X251322964}

\item Gil, J., Li Donni, P., Zucchelli, E., 2019. Uncontrolled diabetes and
health care utilisation: a bivariate latent Markov model approach.
\emph{Health Economics}, 28, 1262–-1276. \doi{https://doi.org/10.1002/hec.3939}

\item Goodman-Bacon, A., 2021. Difference-in-differences with variation in treatment timing. \emph{Journal of Econometrics}, 225, 254--277. \doi{https://doi.org/10.1016/j.jeconom.2021.03.014}

\item Grossman, M., 1972. On the concept of health capital and the demand for health. \emph{Journal of Political Economy}, 80, 223--255. \doi{https://doi.org/10.1086/259880}

\item Heckman, J.J., Smith, J.A., 1999. The pre-programme earnings dip and the determinants of participation in a social programme: implications for simple programme evaluation strategies. \emph{Economic Journal}, 109, 313--348. \doi{https://doi.org/10.1111/1468-0297.00451}

\item Horn, D., Sacarny, A., Zhou, A., 2022. Technology adoption and market allocation: the case of robotic surgery. \emph{Journal of Health Economics}, 86, 102672. \\ \doi{https://doi.org/10.1016/j.jhealeco.2022.102672}

\item Iacus, S.M., King, G., Porro, G., 2012. Causal inference without balance checking: coarsened exact matching. \emph{Political Analysis}, 20, 1--24. \doi{https://doi.org/10.1093/pan/mpr013}

\item Karter, A.J., Parker, M.M., Moffet, H.H., Gilliam, L.K., Dlott, R., 2021. Association of real-time continuous glucose monitoring with glycemic control and acute metabolic events among patients with insulin-treated diabetes. \emph{JAMA}, 325, 2273-–2284. \doi{https://doi.org/10.1001/jama.2021.6530}

\item Laudicella, M., Li Donni, P., Olsen, K.R., Gyrd-Hansen, D., 2022. Age,
morbidity, or something else? A residual approach using microdata to
measure the impact of technological progress on health care expenditure.
\emph{Health Economics}, 31, 1184–-1201. \doi{https://doi.org/10.1002/hec.4500}

\item Lee, J., McCullough, J.S., Town, R.J., 2013. The impact of health information technology on hospital productivity. \emph{RAND Journal of Economics}, 44, 545--568. \\
\doi{https://doi.org/10.1111/1756-2171.12030}

\item Manjelievskaia, J., Shah, V.N., Carlson, A.L., Isaacs, D., Wang, S.M., Pinsker, J.E., Messer, L.H., McDermott, K.W., Lavelle, K., Wall, S., Brixner, D., Malone, D.C., Stemple, C.A., Vaidya, N., Patel, B.V., 2025. Retrospective analysis of impact of automated insulin delivery technology on acute care utilisation from 2019 to 2021. \emph{Diabetic Medicine}, 42, e70147. \doi{https://doi.org/10.1111/dme.70147}

\item Peltzman, S., 1975. The effects of automobile safety regulation. \emph{Journal of Political Economy}, 83, 677--725. \doi{https://doi.org/10.1086/260352}

\item Pintaudi, B., Bruttomesso, D., Candido, R., Girelli, A., Indelicato, L., Mannucci, E., Pizzini, A., Schiaffini, R., Spandonaro, F., Speese, K., Stara, R., Targher, G., \& Vitale, M., 2022. La terapia del diabete mellito di tipo 1. Linea Guida della Associazione dei Medici Diabetologi (AMD), della Società Italiana di Diabetologia (SID) e della Società Italiana di Endocrinologia e Diabetologia Pediatrica (SIEDP). \emph{Journal of AMD (JAMD)}, 25, 45–-54. \doi{https://doi.org/10.36171/jamd22.25.1.9}

\item Pintaudi, B., Bruttomesso, D., Girelli, A., Indelicato, L., Mannucci, E., Pizzini, A., Anelli, V., Romeo, E.L., Schiaffini, R., Spandonaro, F., Migliore, A., Orso, M., D'Angela, D., Polistena, B., Speese, K., Stara, R., Targher, G., Vitale, M., Candido, R., 2026. Italian guidelines for the treatment of type 1 diabetes. \emph{Acta Diabetologica}, 63, 759--773. \doi{https://doi.org/10.1007/s00592-025-02569-1}

\item Rambachan, A., Roth, J., 2023. A more credible approach to parallel trends. \emph{Review of Economic Studies}, 90, 2555–2591. \doi{https://doi.org/10.1093/restud/rdad018}

\item Regione Emilia-Romagna, Commissione Regionale Dispositivi Medici. (2024). \textit{Linee di indirizzo regionali sull’appropriato utilizzo dei dispositivi medici per l’automonitoraggio e l’autogestione del Diabete mellito}. Direzione Generale Cura della Persona, Salute e Welfare, Bologna, dicembre 2024. \\
Available at: https://salute.regione.emilia-romagna.it

\item Rodrigues, D., Kreif, N., Darzi, A., Barahona, M., Mayer, E., 2025.
Digitalization of access to primary care: is there an equity–efficiency
trade-off? \emph{Health Economics}, 34, 1943-–1962.
\doi{https://doi.org/10.1002/hec.70014}

\item Roth, J., 2022. Pretest with caution: event-study estimates after testing for parallel trends. \emph{American Economic Review: Insights}, 4, 305–322. \doi{https://doi.org/10.1257/aeri.20210236}

\item Roubos, L.A.C., Westland, H., Hulstein-Brink, N.L., Visser, R.C., van den Berg, J.W.K., Leenen, J.P.L., 2025. Real-world comparison of telemonitoring versus conventional care in patients with chronic obstructive pulmonary disease and those with asthma: impact on clinical outcomes and patient characteristics. Retrospective cohort study. \emph{Journal of Medical Internet Research}, 27, e66743. \doi{https://doi.org/10.2196/66743}

\item Schmid, C., 2015. Consumer health information and the demand for
physician visits. \emph{Health Economics}, 24, 1619–1631.
\doi{https://doi.org/10.1002/hec.3117}

\item St{\aa}hl, F., Hellman, J., Ekelund, C., 2026. Automated insulin delivery associated with superior glycemic outcomes in type 1 diabetes: a Swedish national registry analysis. \emph{Diabetes Technology \& Therapeutics}. \doi{https://doi.org/10.1177/15209156251414976}

\item Sun, L., Abraham, S., 2021. Estimating dynamic treatment effects in event studies with heterogeneous treatment effects. \emph{Journal of Econometrics}, 225, 175--199. \\
\doi{https://doi.org/10.1016/j.jeconom.2020.09.006}

\item Zeltzer, D., Einav, L., Rashba, J., Waisman, Y., Haimi, M., Balicer,
R.D., 2023. Adoption and utilization of device-assisted telemedicine.
\emph{Journal of Health Economics}, 90, 102780.
\doi{https://doi.org/10.1016/j.jhealeco.2023.102780}

\end{description}
\clearpage

\section*{Tables and Figures}

\begin{table}[!htbp]
\centering
\caption{Timing of first observed AID adoption}
\label{tab:adoption_timing}
\begin{threeparttable}
\begin{tabular*}{\textwidth}{@{\extracolsep{\fill}}rrrrr}
\toprule
Semester & Year & Half & Treated patients & Share treated \\
\midrule
8  & 2019 & 2 & 1  & 0.004 \\
10 & 2020 & 2 & 1  & 0.004 \\
11 & 2021 & 1 & 16 & 0.057 \\
12 & 2021 & 2 & 7  & 0.025 \\
13 & 2022 & 1 & 97 & 0.343 \\
14 & 2022 & 2 & 37 & 0.131 \\
15 & 2023 & 1 & 24 & 0.085 \\
16 & 2023 & 2 & 19 & 0.067 \\
17 & 2024 & 1 & 20 & 0.071 \\
18 & 2024 & 2 & 10 & 0.035 \\
19 & 2025 & 1 & 31 & 0.110 \\
20 & 2025 & 2 & 20 & 0.071 \\
\bottomrule
\end{tabular*}
\begin{tablenotes}[flushleft]
\footnotesize
\item Shares are computed relative to all ever-treated patients.
\item The semester 20 (2025H2) cohort coincides with the semester flagged
as potentially incomplete at data extraction
(Tables~\ref{tab:sample_construction} and~\ref{tab:a6a_sample_device_trends}). Its own event-time observation
at $m=0$ falls outside the calendar window used in the trend-adjusted
event studies, which end at 2025H1.
\end{tablenotes}
\end{threeparttable}
\end{table}


\begin{table}[!htbp]
\centering
\caption{AID Adoption and Diabetes-Care Processes:
Raw Sun--Abraham Event-Study Estimates through 2025H1}
\label{tab:aid-sa-raw-process-2025h1}

\begin{threeparttable}
\small
\setlength{\tabcolsep}{6pt}
\renewcommand{\arraystretch}{1.05}
\def\sym#1{\ifmmode^{#1}\else\(^{#1}\)\fi}

\begin{tabular}{lccc}
\toprule
& (1) & (2) & (3) \\
& \shortstack{Any diabetologist\\visit}
& \shortstack{HbA1c\\measured}
& \shortstack{Log process-of-care\\index} \\
\midrule

Event time \(m=-4\)
    &  0.137\sym{***} &  0.016         &  0.084         \\
    & (0.044)         & (0.033)        & (0.059)        \\

Event time \(m=-3\)
    &  0.094\sym{**}  &  0.051\sym{**} &  0.062         \\
    & (0.044)         & (0.025)        & (0.052)        \\

Event time \(m=-2\)
    &  0.151\sym{***} &  0.033         &  0.066         \\
    & (0.041)         & (0.022)        & (0.050)        \\

Event time \(m=0\)
    &  0.011           &  0.046\sym{*}  &  0.046         \\
    & (0.043)          & (0.027)        & (0.056)        \\

Event time \(m=1\)
    &  0.056           &  0.000         &  0.018         \\
    & (0.043)          & (0.029)        & (0.055)        \\

Event time \(m=2\)
    &  0.022           &  0.027         &  0.051         \\
    & (0.052)          & (0.027)        & (0.057)        \\

Event time \(m=3\)
    &  0.045           & -0.056\sym{*}  & -0.049         \\
    & (0.054)          & (0.032)        & (0.064)        \\

Event time \(m=4\)
    &  0.123\sym{**}  &  0.029         &  0.099         \\
    & (0.056)         & (0.031)        & (0.065)        \\

\addlinespace
\midrule
Observations
    & 19,012 & 19,012 & 19,012 \\

Patient--risk-set units
    & 1,096 & 1,096 & 1,096 \\

Unique matched patients
    & 692 & 692 & 692 \\

Matched AID adopters
    & 181 & 181 & 181 \\

Outcome mean at \(m=-1\)
    & 0.552 & 0.298 & 0.935 \\

Joint pre-period F-test \(p\)-value
    & 0.001 & 0.158 & 0.435 \\

Joint post-period \(p\)-value
    & 0.269 & 0.053 & 0.338 \\

\bottomrule
\end{tabular}
\begin{tablenotes}[flushleft]
\footnotesize
\item \emph{Notes:}
Sun--Abraham estimates use risk-set CEM ATT weights, patient-by-risk-set fixed effects, and common calendar-semester fixed effects. Standard errors are clustered by patient. The panel ends in 2025H1, and $m=-1$ is omitted. Column (3) is $\log(1+\text{six-component process count})$. $^{}p<0.10$, $^{}p<0.05$, and $^{}p<0.01$.
\end{tablenotes}
\end{threeparttable}
\end{table}


\begin{table}[!htbp]
\centering
\caption{AID Adoption and Diabetes-Care Processes:
Trend-Adjusted Sun--Abraham Event-Study Estimates through 2025H1}
\label{tab:aid-sa-trend-process-2025h1}

\begin{threeparttable}
\small
\setlength{\tabcolsep}{6pt}
\renewcommand{\arraystretch}{1.05}
\def\sym#1{\ifmmode^{#1}\else\(^{#1}\)\fi}

\begin{tabular}{@{}lccc@{}}
\toprule
& (1) & (2) & (3) \\
& \shortstack{Any diabetologist\\visit}
& \shortstack{HbA1c\\measured}
& \shortstack{Log process-of-care\\index} \\
\midrule

Event time \(m=-4\)
    & -0.024          & -0.062\sym{***} & -0.014         \\
    & (0.023)         & (0.016)         & (0.030)        \\

Event time \(m=-3\)
    & -0.014          & -0.001          & -0.003         \\
    & (0.043)         & (0.026)         & (0.053)        \\

Event time \(m=-2\)
    &  0.097\sym{***} &  0.007          &  0.033         \\
    & (0.037)         & (0.021)         & (0.044)        \\

Event time \(m=0\)
    &  0.065          &  0.072\sym{*}   &  0.079         \\
    & (0.055)         & (0.038)         & (0.076)        \\

Event time \(m=1\)
    &  0.163\sym{**}  &  0.052          &  0.083         \\
    & (0.072)         & (0.052)         & (0.097)        \\

Event time \(m=2\)
    &  0.184\sym{*}   &  0.106\sym{*}   &  0.149         \\
    & (0.099)         & (0.063)         & (0.125)        \\

Event time \(m=3\)
    &  0.261\sym{**}  &  0.048          &  0.082         \\
    & (0.117)         & (0.080)         & (0.151)        \\

Event time \(m=4\)
    &  0.393\sym{***} &  0.159\sym{*}   &  0.262         \\
    & (0.139)         & (0.093)         & (0.178)        \\

\addlinespace
\midrule
Observations
    & 19,012 & 19,012 & 19,012 \\

Patient--risk-set units
    & 1,096 & 1,096 & 1,096 \\

Unique matched patients
    & 692 & 692 & 692 \\

Matched AID adopters
    & 181 & 181 & 181 \\

Outcome mean at \(m=-1\)
    & 0.552 & 0.298 & 0.935 \\

Joint adjusted F-test pre-period \(p\)-value
    & 0.027 & \(<0.001\) & 0.834 \\

Joint adjusted post-period \(p\)-value
    & 0.044 & 0.043 & 0.265 \\

\bottomrule
\end{tabular}

\begin{tablenotes}[flushleft]
\footnotesize
\item \emph{Notes:}
Trend-adjusted Sun--Abraham estimates use a GMM linear pre-trend over $m=-3,-2,-1$, risk-set CEM ATT weights, patient-by-risk-set fixed effects, and common calendar-semester fixed effects. Standard errors are clustered by patient. The panel ends in 2025H1, and $m=-1$ is normalized to zero. Column (3) is $\log(1+\text{six-component process count})$. $^{}p<0.10$, $^{}p<0.05$, and $^{}p<0.01$.
\end{tablenotes}

\end{threeparttable}
\end{table}

\begin{table}[!htbp]
\centering
\caption{Sample construction}
\footnotesize
\begin{threeparttable}
\begin{adjustbox}{max width=\textwidth}
\begin{tabular}{p{0.40\textwidth}rp{0.42\textwidth}}
\toprule
Sample-construction step & Count/value & Definition or comment \\
\midrule
Total theoretical patient-semester rows & 32,160 & Balanced patient-semester frame before exposure restrictions \\
Rows with positive exposure or active follow-up & 28,056 & Rows used to describe active clinical follow-up \\
Effective person-semesters & 27,159.289 & Sum of exposure shares across patient-semester rows \\
Rows in Auxilium/device era & 19,296 & 2020 onward, per c\_aux2020 \\
Effective person-semesters in Auxilium/device era & 17,615.236 & Exposure-weighted support in the device era \\
Rows with observed current device status & 27,802 & Nonmissing current AID/device status \\
Unique linked patients & 1,608 & Patients linked across clinical, device, and registry sources \\
Patients with at least one active semester & 1,608 & Patient-level active-follow-up support \\
Patients active at least once in Auxilium/device era & 1,608 & Patient-level device-era support \\
Patients ever observed with AID & 283 & Ever-treated patients \\
Patients never observed with AID & 1,325 & Never-treated comparison patients \\
Patients with identified first AID semester & 283 & Staggered-adoption cohorts with nonmissing adoption timing \\
Rows in potentially incomplete final semester & 1,608 & Flagged as 2025H2 \\
\bottomrule
\end{tabular}
\end{adjustbox}
\begin{tablenotes}[flushleft]
\footnotesize
\item Notes: The table distinguishes theoretical patient-semester rows from active follow-up and exposure-weighted person-semesters. Panel period: 2016--2025. Patient-level denominators coincide at 1,608 because no patient is dropped entirely by the active-follow-up or Auxilium-era restrictions; only within-patient exposure varies. Of the 283 adopters, 181 are retained in the matched risk sets used for estimation. The other 102 are not retained for three reasons: 20 first adopt in 2025H2, outside the estimation window, and do not form an adoption cohort; 25 lack four active pre-adoption semesters; and 57 have no eligible untreated control in their matching stratum and fall outside common support. The 181 retained adopters are 68.8\% of the 263 adopters in estimable cohorts (2019H2--2025H1).
\end{tablenotes}
\end{threeparttable}
\label{tab:sample_construction}
\end{table}

\begin{table}[!htbp]
\centering
\caption{Year-half sample and device trends}
\label{tab:a6a_sample_device_trends}
\scriptsize
\setlength{\tabcolsep}{3.5pt}
\renewcommand{\arraystretch}{0.90}
\begin{adjustbox}{width=\textwidth}
\begin{threeparttable}
\begin{tabular}{lrrrrrrrr}
\toprule
Year-half 
& Rows 
& \makecell{Active\\rows} 
& \makecell{Person-\\semesters} 
& \makecell{Share\\active} 
& \makecell{Device\\observed} 
& \makecell{Incomplete\\tail} 
& \makecell{AID\\on} 
& \makecell{Any\\device} \\
\midrule
2016H1 & 1,608 & 1,131 & 908.615   & 0.703 & 0.655 & 0 & 0.000 & 0.000 \\
2016H2 & 1,608 & 1,194 & 1,098.337 & 0.743 & 0.706 & 0 & 0.000 & 0.000 \\
2017H1 & 1,608 & 1,249 & 1,169.884 & 0.777 & 0.747 & 0 & 0.000 & 0.000 \\
2017H2 & 1,608 & 1,275 & 1,219.217 & 0.793 & 0.767 & 0 & 0.000 & 0.000 \\
2018H1 & 1,608 & 1,303 & 1,245.497 & 0.810 & 0.784 & 0 & 0.000 & 0.000 \\
2018H2 & 1,608 & 1,328 & 1,273.924 & 0.826 & 0.800 & 0 & 0.000 & 0.000 \\
2019H1 & 1,608 & 1,349 & 1,302.481 & 0.839 & 0.817 & 0 & 0.000 & 0.000 \\
2019H2 & 1,608 & 1,369 & 1,326.098 & 0.851 & 0.835 & 0 & 0.001 & 0.001 \\
2020H1 & 1,608 & 1,382 & 1,352.071 & 0.859 & 0.867 & 0 & 0.001 & 0.019 \\
2020H2 & 1,608 & 1,417 & 1,376.451 & 0.881 & 0.889 & 0 & 0.001 & 0.277 \\
2021H1 & 1,608 & 1,441 & 1,418.127 & 0.896 & 0.906 & 0 & 0.012 & 0.558 \\
2021H2 & 1,608 & 1,463 & 1,440.815 & 0.910 & 0.919 & 0 & 0.017 & 0.816 \\
2022H1 & 1,608 & 1,489 & 1,459.492 & 0.926 & 0.930 & 0 & 0.082 & 0.869 \\
2022H2 & 1,608 & 1,500 & 1,476.402 & 0.933 & 0.937 & 0 & 0.107 & 0.893 \\
2023H1 & 1,608 & 1,505 & 1,492.050 & 0.936 & 0.942 & 0 & 0.121 & 0.898 \\
2023H2 & 1,608 & 1,521 & 1,504.864 & 0.946 & 0.951 & 0 & 0.127 & 0.912 \\
2024H1 & 1,608 & 1,538 & 1,517.379 & 0.956 & 0.961 & 0 & 0.137 & 0.890 \\
2024H2 & 1,608 & 1,544 & 1,528.196 & 0.960 & 0.966 & 0 & 0.136 & 0.907 \\
2025H1 & 1,608 & 1,533 & 1,529.812 & 0.953 & 0.960 & 0 & 0.156 & 0.879 \\
2025H2 & 1,608 & 1,525 & 1,519.576 & 0.948 & 0.952 & 1 & 0.166 & 0.500 \\
\bottomrule
\end{tabular}
\begin{tablenotes}[flushleft]
\footnotesize
\item[] \textit{Notes:} Device variables are exposure-weighted means. The final semester may be incomplete.
\end{tablenotes}
\end{threeparttable}
\end{adjustbox}
\end{table}

\begin{table}[!htbp]
\centering
\footnotesize
\caption{Mortality and non-death attrition during the study window}
\label{tab:mortality_attrition_summary}
\begin{threeparttable}
\begin{tabular*}{\textwidth}{@{\extracolsep{\fill}}p{0.40\textwidth}rrr}
\toprule
Category & All patients & Ever-AID & Never-AID \\
\midrule
Patients in linked panel with active follow-up in the study window & 1,608 (100.0) & 283 (100.0) & 1,325 (100.0) \\
Died during study window & 83 (5.2) & 5 (1.8) & 78 (5.9) \\
Alive and observed through final usable semester (2025H1) & 1,525 (94.8) & 278 (98.2) & 1,247 (94.1) \\
\bottomrule
\end{tabular*}
\begin{tablenotes}[flushleft]
\footnotesize
\item Notes: Cell entries are N (\% of the column's patient group). Study window: 1 January 2020 to 30 June 2025. Final usable semester: 2025H1. All 1,608 patients in the linked panel have at least one active semester overlapping the study window, so the active-follow-up row coincides with the full panel. The two rows reported account for the entire panel in each column (e.g., 83 + 1,525 = 1,608); non-death attrition (active follow-up that ends before 2025H1 without a recorded death) and loss of all active follow-up are both exactly zero for every patient in every group, so every patient is observed either to die within the window or to survive to the final usable semester. Active follow-up uses c\_active when available; otherwise effective exposure and observed activity are used.
\end{tablenotes}
\end{threeparttable}
\end{table}


\begin{table}[!htbp]
\centering
\caption{Primary outcome definitions and availability}
\label{tab:main_outcomes}
\footnotesize
\setlength{\tabcolsep}{4pt}
\renewcommand{\arraystretch}{1.15}

\begin{threeparttable}
\begin{tabularx}{\textwidth}{
>{\raggedright\arraybackslash}p{0.24\textwidth}
>{\raggedright\arraybackslash}X
>{\centering\arraybackslash}p{0.14\textwidth}
>{\centering\arraybackslash}p{0.08\textwidth}
}
\toprule
Outcome & Definition & Active semesters & Mean \\
\midrule
Any diabetologist visit
& Equal to one if at least one diabetologist visit is recorded during the semester.
& 28,056 & 0.598 \\
HbA1c measured
& Equal to one if at least one HbA1c measurement is recorded during the semester.
& 28,056 & 0.490 \\
Composite process-of-care index (six components)
& Count of recommended diabetes-monitoring activities recorded during the semester:
HbA1c, LDL cholesterol, albuminuria, kidney-function testing, eye examination, and
diabetologist visit.
& 28,056 & 2.959 \\
\bottomrule
\end{tabularx}
\begin{tablenotes}[flushleft]
\scriptsize
\item Notes: The table reports the primary outcomes used to estimate the effect of AID adoption on routine outpatient engagement. Availability is computed among active patient-semester observations. For visit and measurement indicators, absence of a recorded activity is coded as zero; therefore, the mean equals the share of active semesters with the corresponding recorded activity. Means are weighted by effective person-semester exposure. \item This six-component index includes the visit indicator and excludes blood pressure and BMI. It is distinct from the seven-component clinical indicator count reported as a baseline descriptive in Table C1, which excludes visits and includes blood pressure and BMI; the two should not be read as the same construction.
\end{tablenotes}
\end{threeparttable}
\end{table}

\begin{table}[!htbp]
\centering
\caption{Balance before and after risk-set CEM}
\label{tab:cem_balance}
\scriptsize
\begin{threeparttable}
\begin{adjustbox}{max width=\textwidth}
\begin{tabular}{lrrrrrr}
\toprule
Variable & Treat before & Control before & Std. diff. before & Treat after & Control after & Std. diff. after \\
\midrule
\multicolumn{7}{l}{\textit{Panel A: Exact-match variables}}\\
Female & 0.393 & 0.578 & -0.376 & 0.425 & 0.425 & 0.000 \\
Age (continuous) & 42.216 & 51.489 & -0.624 & 42.381 & 42.915 & -0.039 \\
Age group (CEM matching cell) & 2.638 & 3.137 & -0.632 & 2.648 & 2.648 & 0.000 \\
\textit{Mean |std. diff.|} & & & \textit{0.544} & & & \textit{0.013} \\
\multicolumn{7}{l}{\textit{Panel B: Coarsened-exact-match clinical history}}\\
Mean HbA1c, pre-adoption window & 7.661 & 7.808 & -0.125 & 7.569 & 7.702 & -0.130 \\
Mean eGFR, pre-adoption window & 100.902 & 94.952 & 0.259 & 107.362 & 105.628 & 0.092 \\
Mean seven-item measured-process count, pre-adoption window & 2.594 & 2.942 & -0.225 & 2.528 & 2.566 & -0.024 \\
Mean seven-item covered-process count, pre-adoption window & 3.904 & 4.448 & -0.316 & 3.816 & 3.909 & -0.052 \\
Pump active, pre-adoption window & 0.082 & 0.014 & 0.322 & 0.036 & 0.036 & 0.000 \\
Any complication, pre-adoption window & 0.097 & 0.156 & -0.176 & 0.073 & 0.073 & 0.000 \\
Any acute care, pre-adoption window & 0.490 & 0.475 & 0.030 & 0.477 & 0.477 & 0.000 \\
\textit{Mean |std. diff.|} & & & \textit{0.208} & & & \textit{0.043} \\
\multicolumn{7}{l}{\textit{Panel C: Pre-adoption levels of outcome variables}}\\
Any visit, baseline semester & 0.549 & 0.542 & 0.013 & 0.565 & 0.610 & -0.092 \\
HbA1c measured, baseline semester & 0.374 & 0.500 & -0.256 & 0.316 & 0.361 & -0.095 \\
LDL measured, baseline semester & 0.564 & 0.551 & 0.028 & 0.549 & 0.549 & 0.001 \\
Blood pressure measured, baseline semester & 0.389 & 0.392 & -0.006 & 0.420 & 0.465 & -0.092 \\
BMI measured, baseline semester & 0.529 & 0.512 & 0.034 & 0.539 & 0.540 & -0.001 \\
Albuminuria measured, baseline semester & 0.393 & 0.395 & -0.005 & 0.347 & 0.312 & 0.075 \\
Kidney function measured, baseline semester & 0.381 & 0.496 & -0.233 & 0.321 & 0.367 & -0.097 \\
Eye exam, baseline semester\tnote{c} & 0.066 & 0.114 & -0.167 & 0.062 & 0.128 & -0.225 \\
\textit{Mean |std. diff.|} & & & \textit{0.093} & & & \textit{0.085} \\
\bottomrule
\end{tabular}
\end{adjustbox}
\makeatletter
\renewcommand{\TPT@hsize}{\hsize\textwidth \@parboxrestore}
\makeatother
\begin{tablenotes}[flushleft]
\footnotesize
\item Notes: The table reports balance between AID adopters and eligible comparison patients before and after risk-set coarsened exact matching. Before-matching statistics use all eligible not-yet-treated or never-treated controls in each adoption risk set and are unweighted. After-matching statistics restrict the sample to retained matched strata and weight controls so that, within each matched stratum, the weighted control distribution represents the treated patients. The sample is stacked by adoption risk set; therefore, a patient who has not yet adopted AID may appear as a control for more than one earlier adoption cohort. Panel A variables are matched exactly within coarsened cells, so remaining imbalance in continuous age reflects within-cell variation. Panel C reports baseline values of outcome variables; these variables are not direct matching targets. Standardized differences use pooled standard deviations, and mean absolute standardized differences summarize balance within each panel.
\end{tablenotes}
\end{threeparttable}
\end{table}

\clearpage
\begin{figure}[!htbp]
\centering
\includegraphics[width=0.90\textwidth]{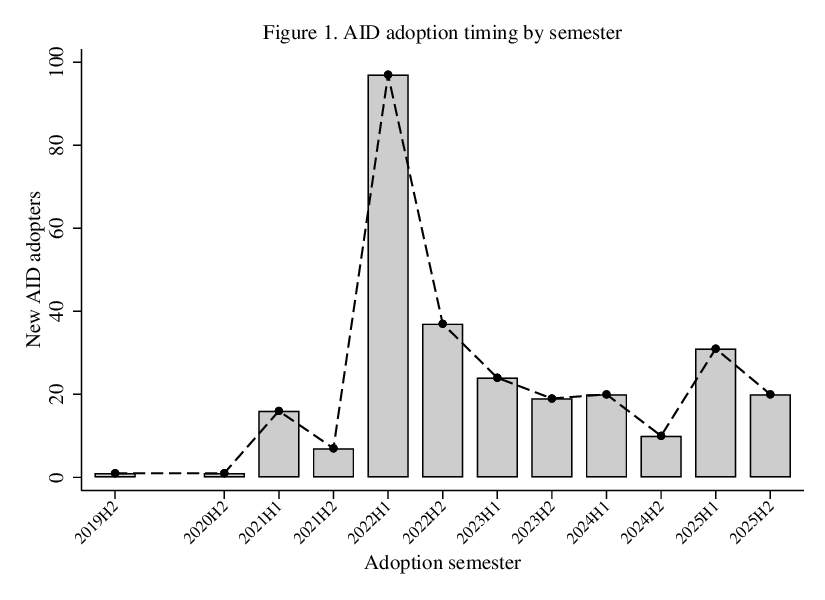}
\caption{AID adoption timing by semester}
\label{fig:aid_adoption_timing}
\end{figure}
\begin{figure}[!htbp]
\centering
\includegraphics[width=0.90\textwidth]{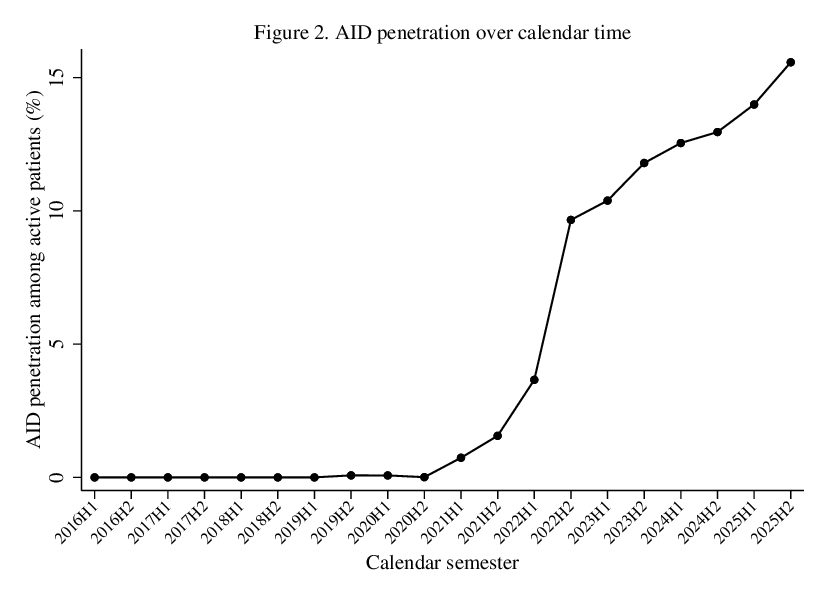}
\caption{AID penetration over calendar time}
\label{fig:aid_penetration_calendar}
\end{figure}
\begin{figure}[!htbp]
\centering
\includegraphics[width=0.95\textwidth]{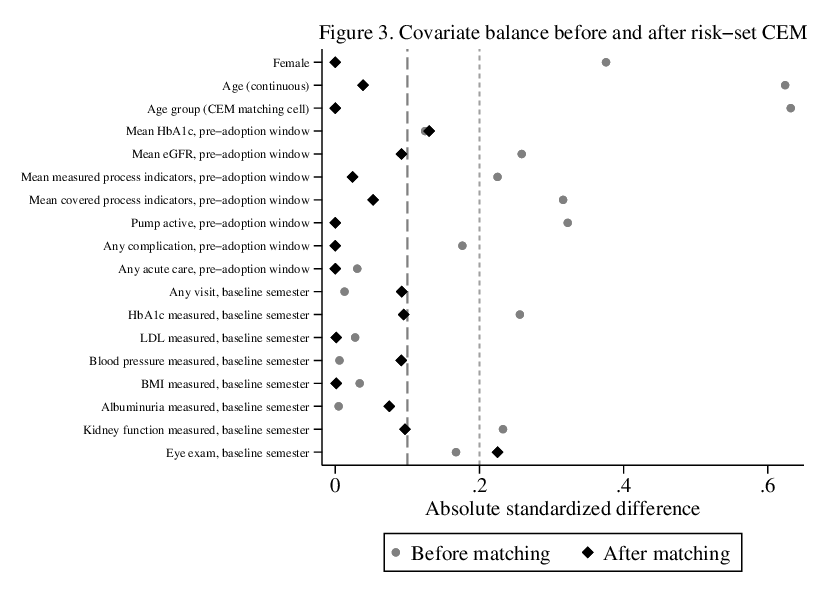}
\caption{Covariate balance before and after risk-set CEM}
\label{fig:loveplot_cem}
\end{figure}

\begin{figure}[!htbp]
\centering
\includegraphics[width=0.75\textwidth]{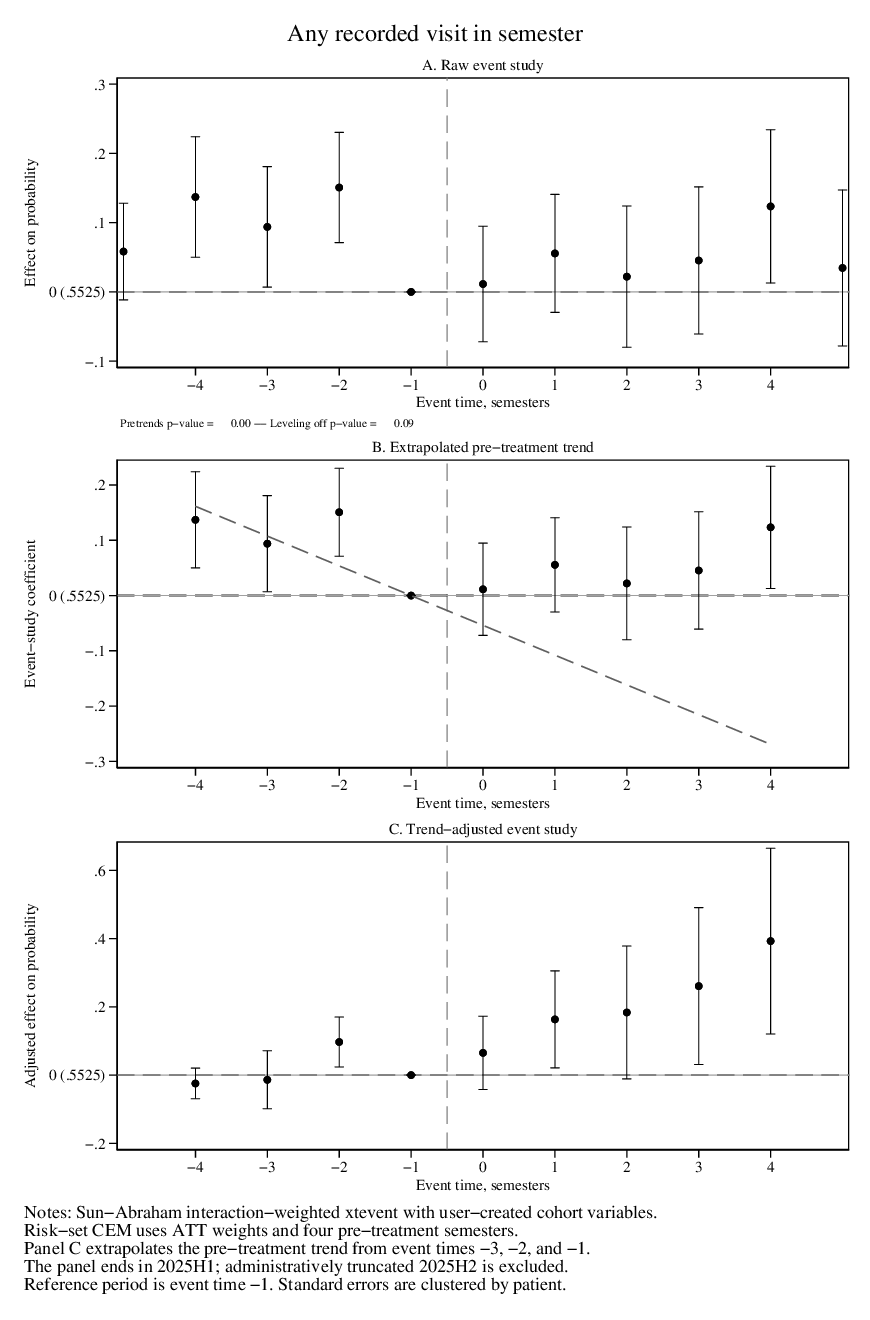}
\caption{Dynamic effect of AID adoption on recorded diabetologist visits}
\label{fig:aid_visit_eventstudy}
\begin{minipage}{0.90\textwidth}
\footnotesize
\emph{Notes:} The figure reports Sun--Abraham event-study estimates for AID
adoption and any recorded diabetologist visit. Event time is in semesters
relative to adoption, with $m=-1$ omitted. Panel A shows raw estimates,
Panel B the extrapolated pre-trend, and Panel C the trend-adjusted estimates.
Risk-set CEM weights are used; standard errors are clustered by patient.
\end{minipage}
\end{figure}

\begin{figure}[!htbp]
\centering
\includegraphics[width=0.90\textwidth]{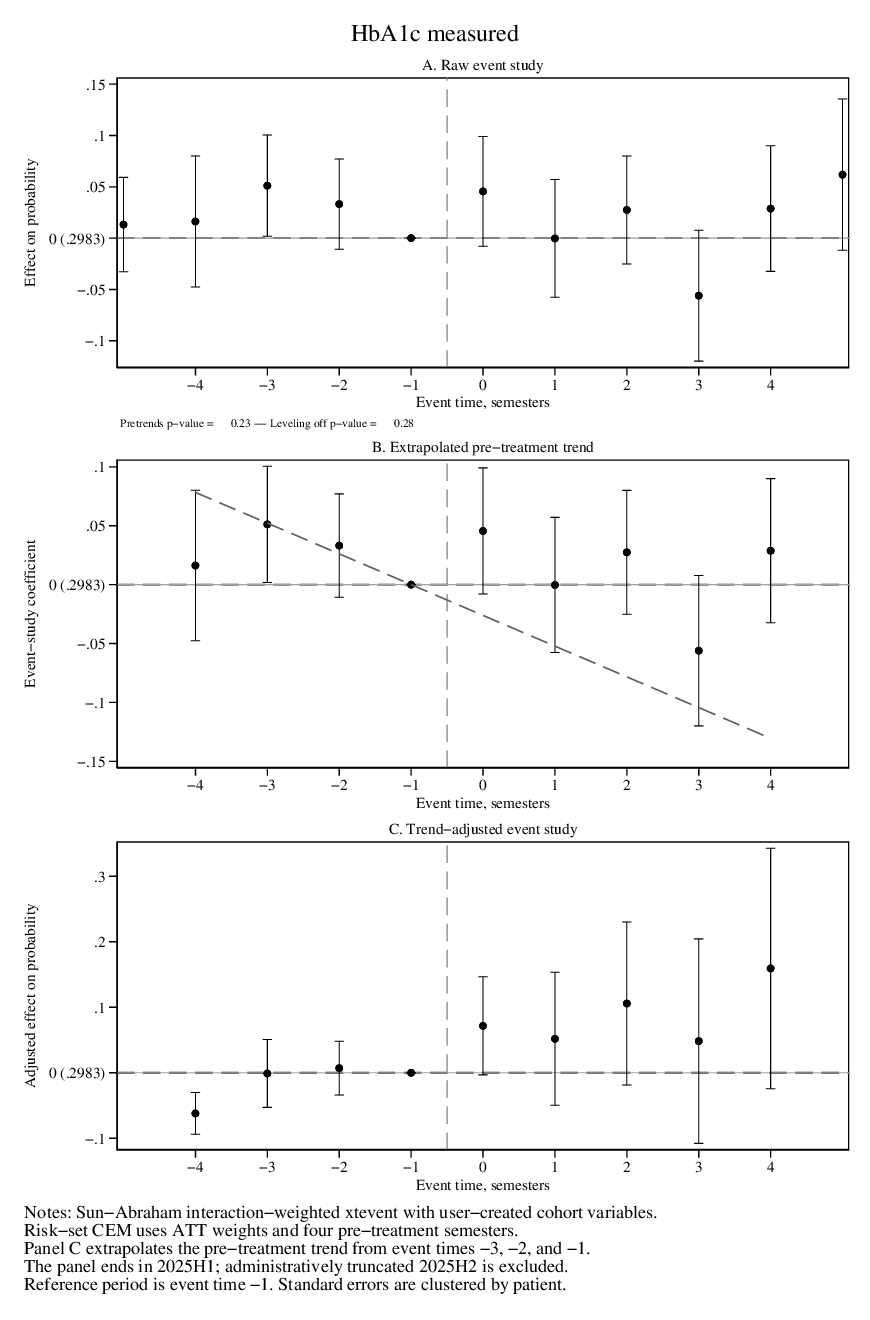}
\caption{Dynamic effect of AID adoption on HbA1c measurement}
\label{fig:aid_hba1c_measured_eventstudy}
\begin{minipage}{0.90\textwidth}
\footnotesize
\emph{Notes:} The figure reports Sun--Abraham event-study estimates for AID
adoption and HbA1c measurement. Event time is in semesters relative to
adoption, with $m=-1$ omitted. Panel A shows raw estimates, Panel B the
extrapolated pre-trend, and Panel C the trend-adjusted estimates. Risk-set
CEM weights are used; standard errors are clustered by patient.
\end{minipage}
\end{figure}

\begin{figure}[!htbp]
\centering
\includegraphics[width=0.75\textwidth]{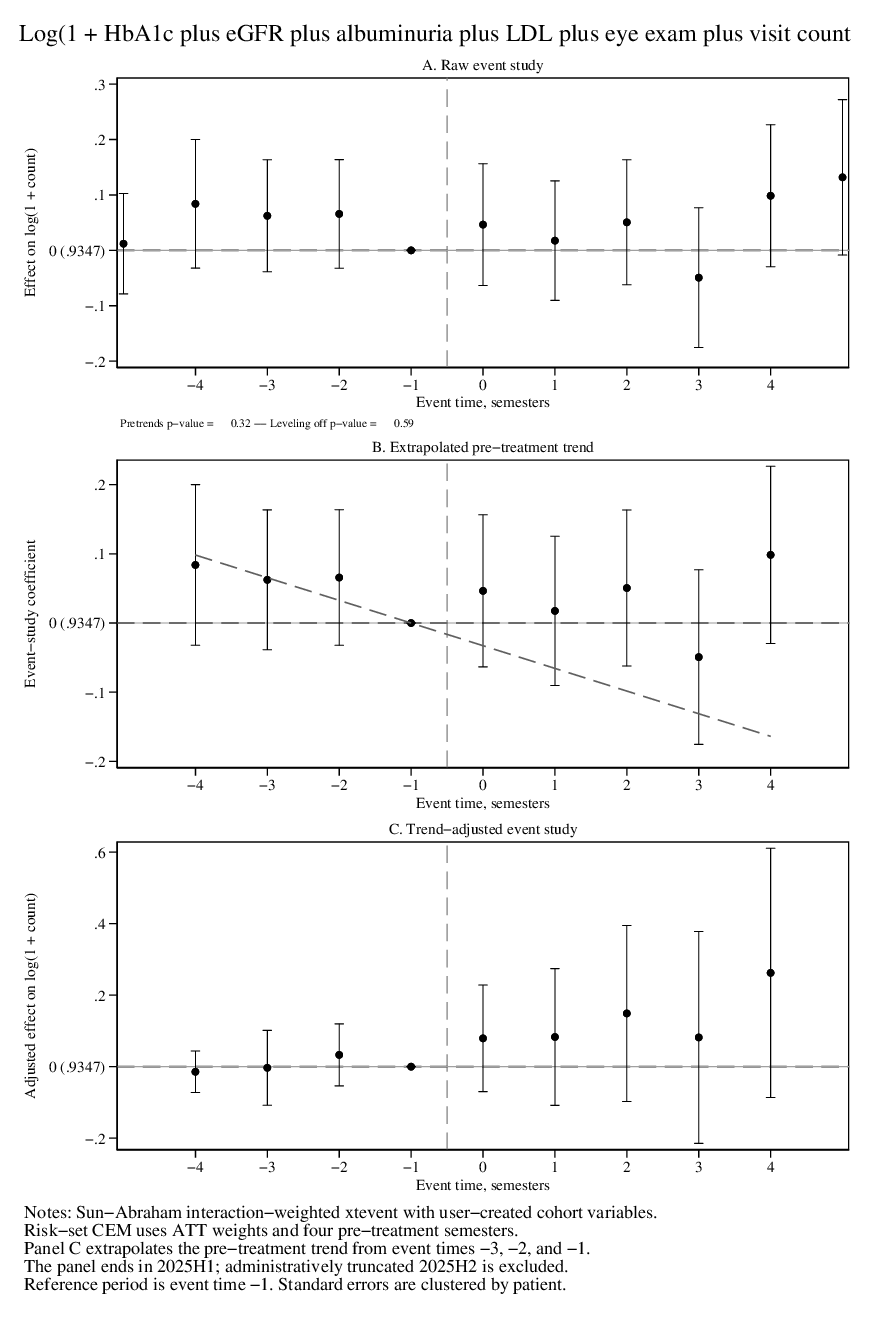}
\caption{Dynamic effect of AID adoption on the composite process-of-care index}
\label{fig:aid_process_index_eventstudy}
\begin{minipage}{0.90\textwidth}
\footnotesize
\emph{Notes:} The figure reports Sun--Abraham event-study estimates for AID
adoption and the composite process-of-care index, defined as
$\log(1+\text{recorded monitoring activities})$. Event time is in semesters
relative to adoption, with $m=-1$ omitted. Panel A shows raw estimates,
Panel B the extrapolated pre-trend, and Panel C the trend-adjusted estimates.
Risk-set CEM weights are used; standard errors are clustered by patient.
\end{minipage}
\end{figure}


\begin{figure}[!htbp]
\centering
\includegraphics[width=0.90\textwidth]{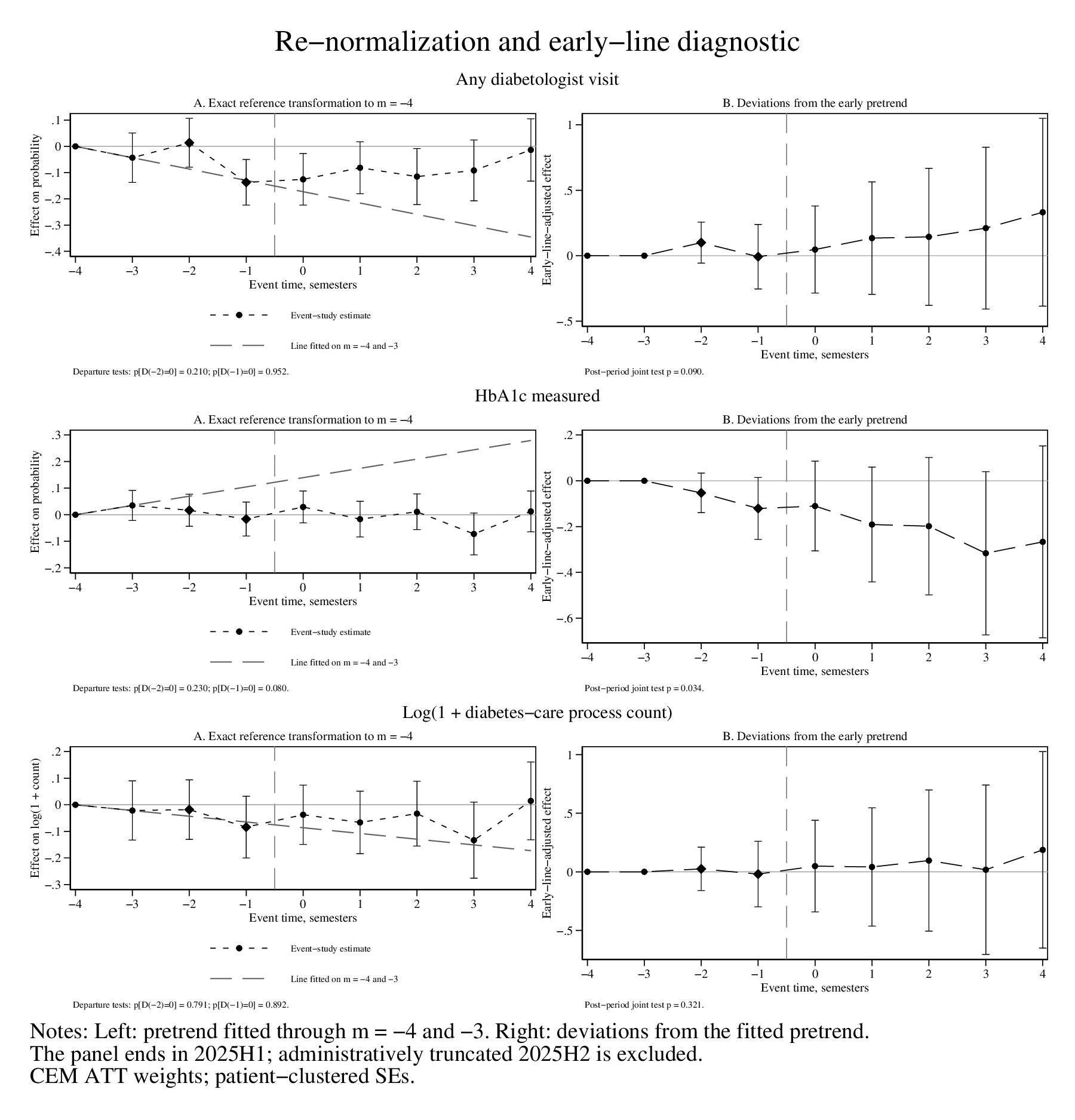}
\caption{Re-normalization and early-pre-trend diagnostic. Left panels re-normalize the event-study estimates to event time \(m=-4\) and overlay the linear trend fitted through \(m=-4\) and \(m=-3\). Right panels report deviations from this early pre-adoption trend. Rows correspond to diabetologist visits, HbA1c measurement, and the composite diabetes-care process index.}
\label{fig:renormalization_diagnostic}
\end{figure}

\begin{figure}[!htbp]
\centering
\includegraphics[width=0.90\textwidth]{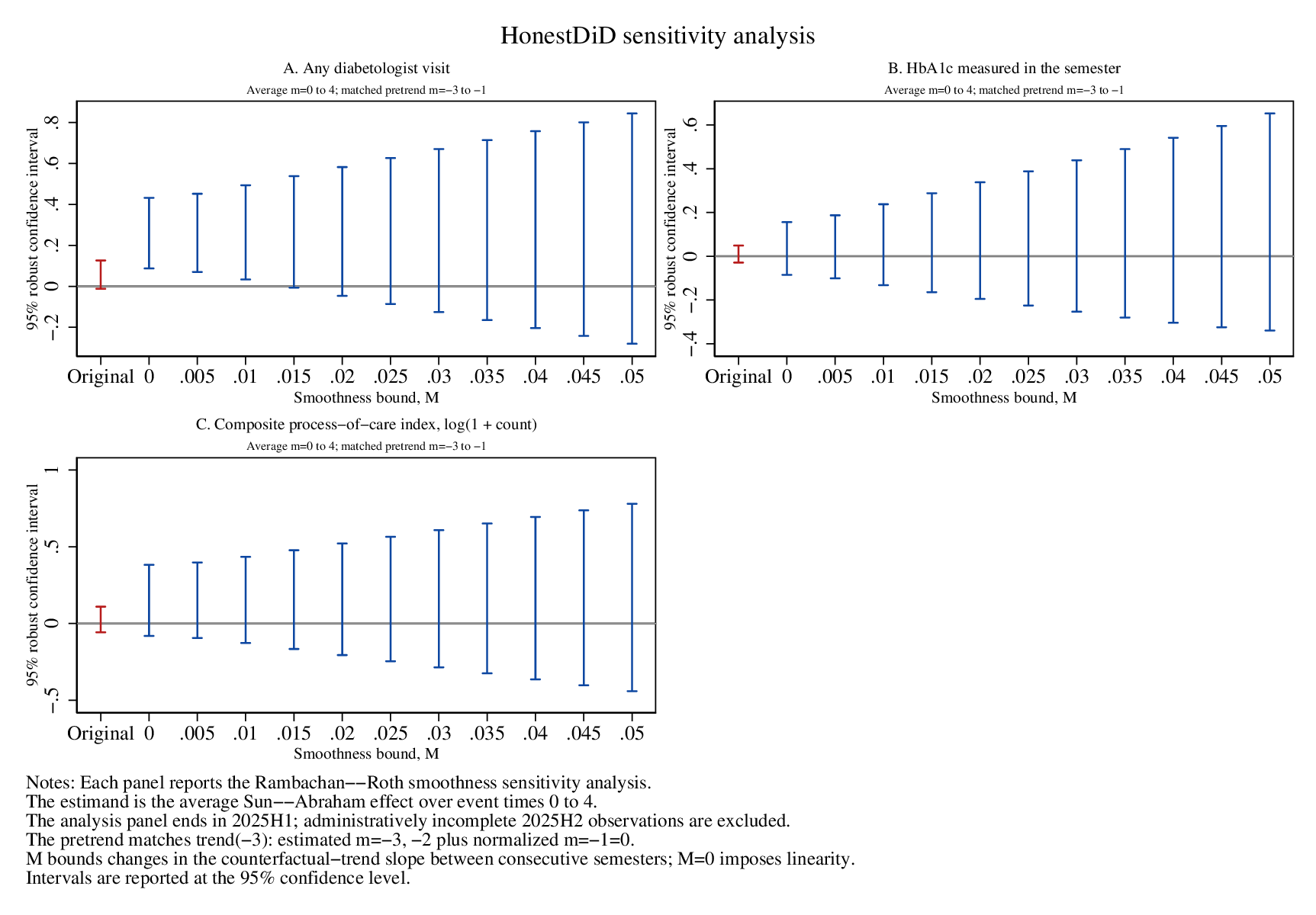}
\caption{HonestDiD sensitivity analysis of the effects of AID adoption on healthcare monitoring}
\label{fig:aid_honestdid_three_outcomes}
\begin{minipage}{0.90\textwidth}
\footnotesize
\emph{Notes:} The three panels report Rambachan--Roth smoothness-based
sensitivity analyses for any diabetologist visit, HbA1c measurement, and the
composite process-of-care index, respectively. The estimand is the average
Sun--Abraham treatment effect over event times $m=0,\ldots,4$. The
counterfactual pre-adoption trend is estimated using event times $m=-3$ and
$m=-2$, together with the normalized reference period $m=-1$. The sensitivity
parameter $M$ bounds the maximum change in the slope of the untreated potential
outcome between consecutive semesters; $M=0$ therefore imposes a linear
continuation of the pre-adoption trend. Points show estimated effects and bars
report 95\% confidence intervals. Estimates use risk-set CEM ATT weights, with
standard errors clustered at the patient level.
\end{minipage}
\end{figure}

\clearpage


\begin{figure}[!htbp]
\centering
\includegraphics[width=0.75\textwidth]{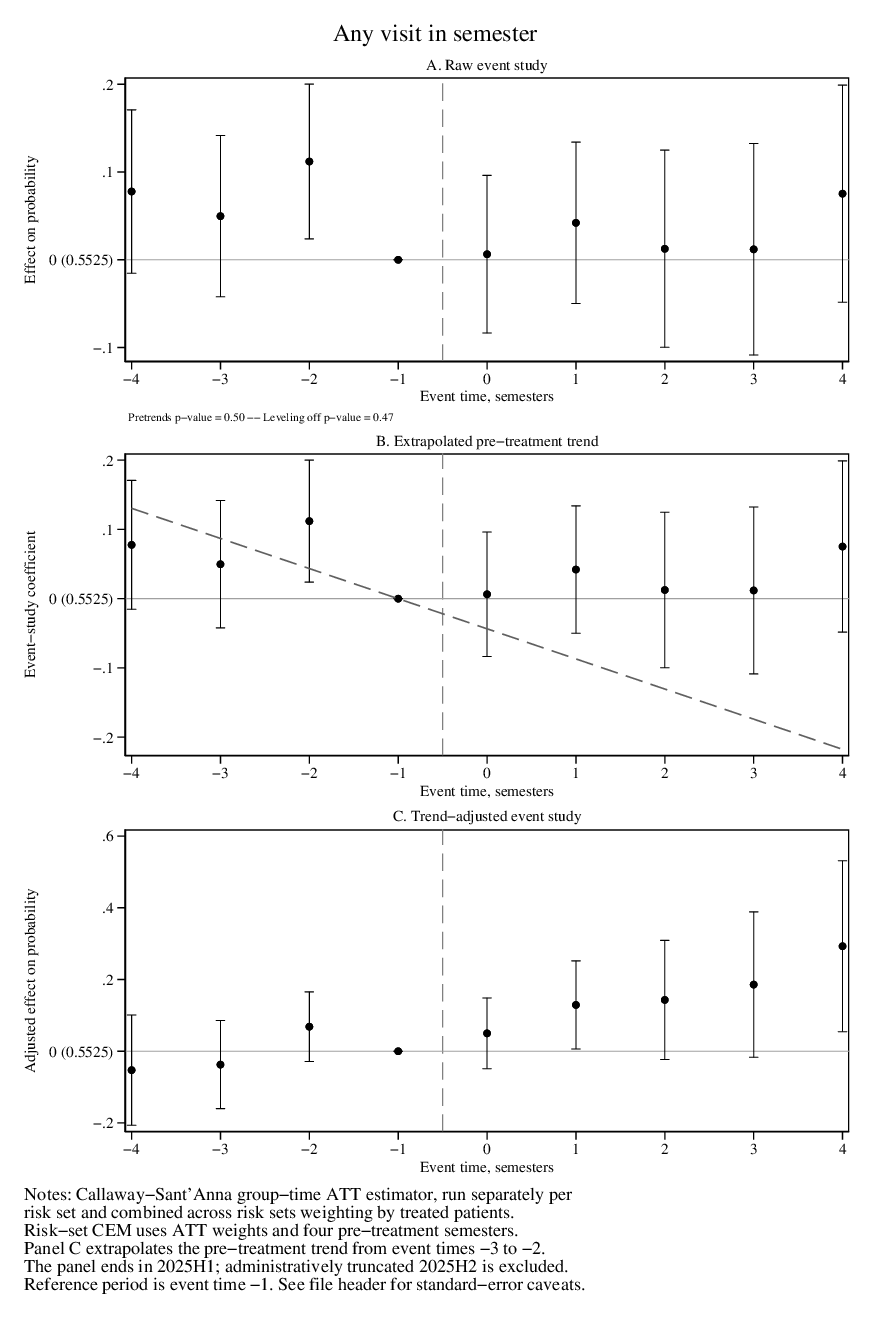}
\caption{Callaway--Sant'Anna estimates of the dynamic effect of AID adoption
on diabetologist visits}
\label{fig:cs_visit}
\begin{minipage}{0.90\textwidth}
\footnotesize
\emph{Notes:} The figure reports Callaway--Sant'Anna group-time ATT estimates for the effect of AID adoption on having any diabetologist visit during the semester. Event time is measured in semesters relative to adoption, with $m=-1$ omitted. Panel A presents the unadjusted estimates, Panel B shows the pre-treatment trend extrapolated from event times $m=-3$, $m=-2$ and $m=-1$, and Panel C presents the trend-adjusted estimates. Risk-set CEM weights are used throughout. Group-time ATTs are estimated separately within each risk set and then aggregated across risk sets using treated-patient counts as weights. Standard errors are clustered at the patient level.
\end{minipage}
\end{figure}


\begin{figure}[!htbp]
\centering
\includegraphics[width=0.75\textwidth]{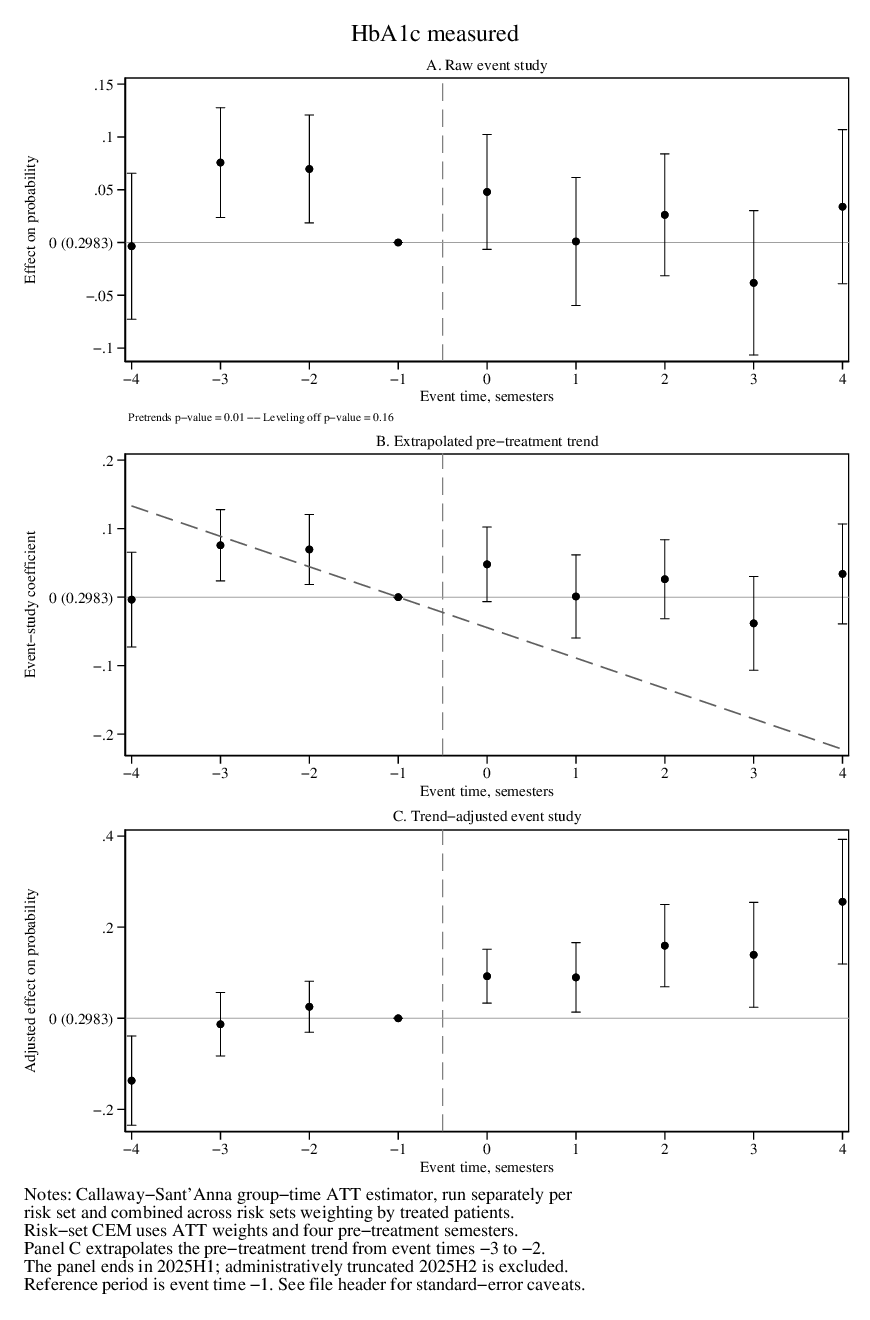}
\caption{Callaway--Sant'Anna estimates of the dynamic effect of AID adoption
on HbA1c measurement}
\label{fig:cs_hba}
\begin{minipage}{0.90\textwidth}
\footnotesize
\emph{Notes:} The figure reports Callaway--Sant'Anna group-time ATT estimates
for the effect of AID adoption on whether HbA1c was measured during the
semester. Event time is measured in semesters relative to adoption, with
$m=-1$ omitted. Panel A presents the unadjusted estimates, Panel B shows the
pre-treatment trend extrapolated from event times $m=-3$, $m=-2$ and $m=-1$, and Panel C presents the trend-adjusted estimates. Risk-set CEM weights are used throughout. Group-time ATTs are estimated separately within each risk set and then aggregated across risk sets using treated-patient counts as weights. Standard errors are clustered at the patient level.
\end{minipage}
\end{figure}


\begin{figure}[!htbp]
\centering
\includegraphics[width=0.75\textwidth]{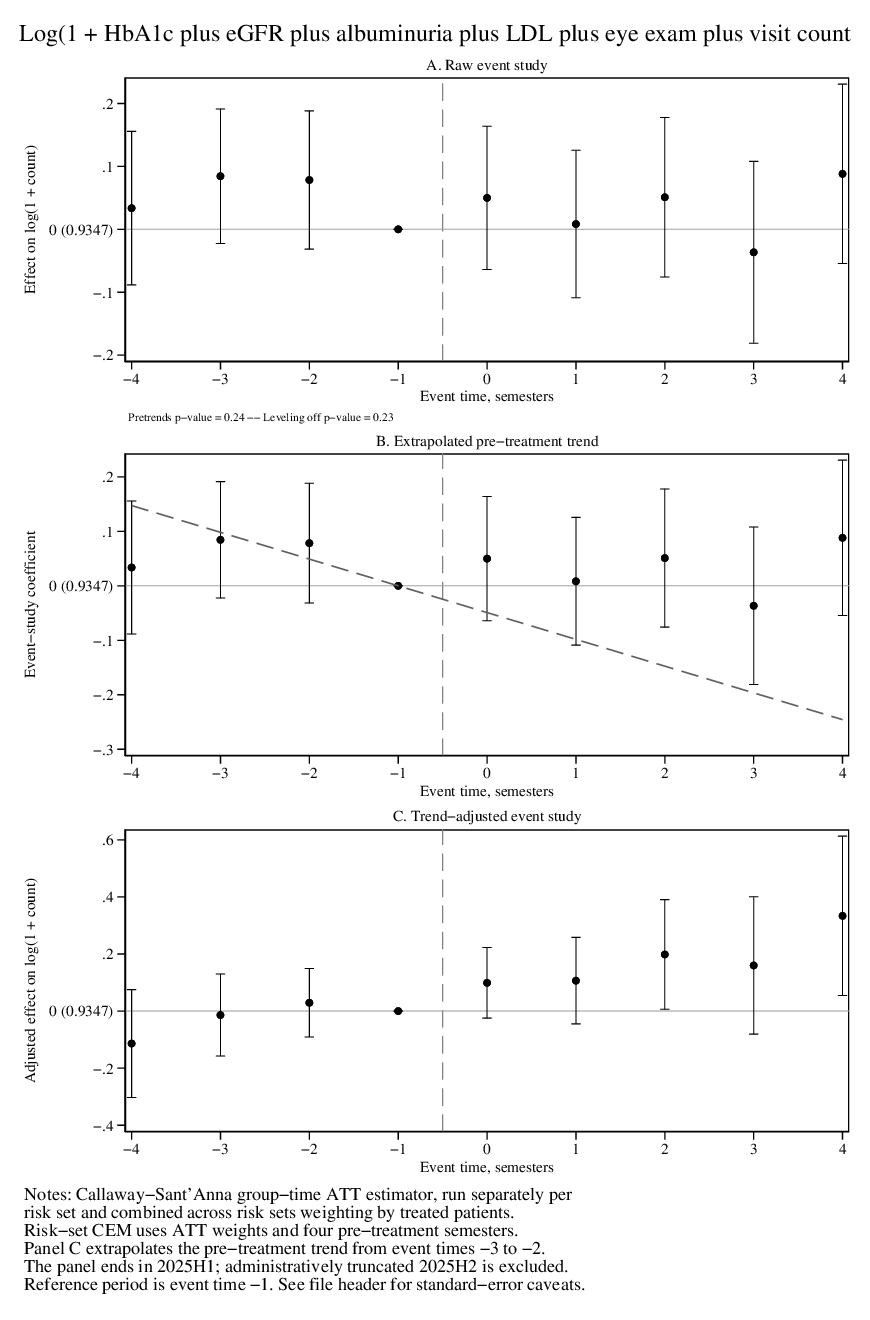}
\caption{Callaway--Sant'Anna estimates of the dynamic effect of AID adoption
on the composite process-of-care index}
\label{fig:cs_process}
\begin{minipage}{0.90\textwidth}
\footnotesize
\emph{Notes:} The figure reports Callaway--Sant'Anna group-time ATT estimates for the effect of AID adoption on the composite process-of-care index, defined as $\log(1+\text{recorded process-of-care activities})$. Event time is measured in semesters relative to adoption, with $m=-1$ omitted. Panel A presents the unadjusted estimates, Panel B shows the pre-treatment trend extrapolated from event times $m=-3$, $m=-2$ and $m=-1$, and Panel C presents the trend-adjusted estimates. Risk-set CEM weights are used throughout. Group-time ATTs are estimated separately within each risk set and then aggregated across risk sets using treated-patient counts as weights. Standard errors are clustered at the patient level.
\end{minipage}
\end{figure}




\begin{figure}[!htbp]
\centering
\includegraphics[width=0.75\textwidth]{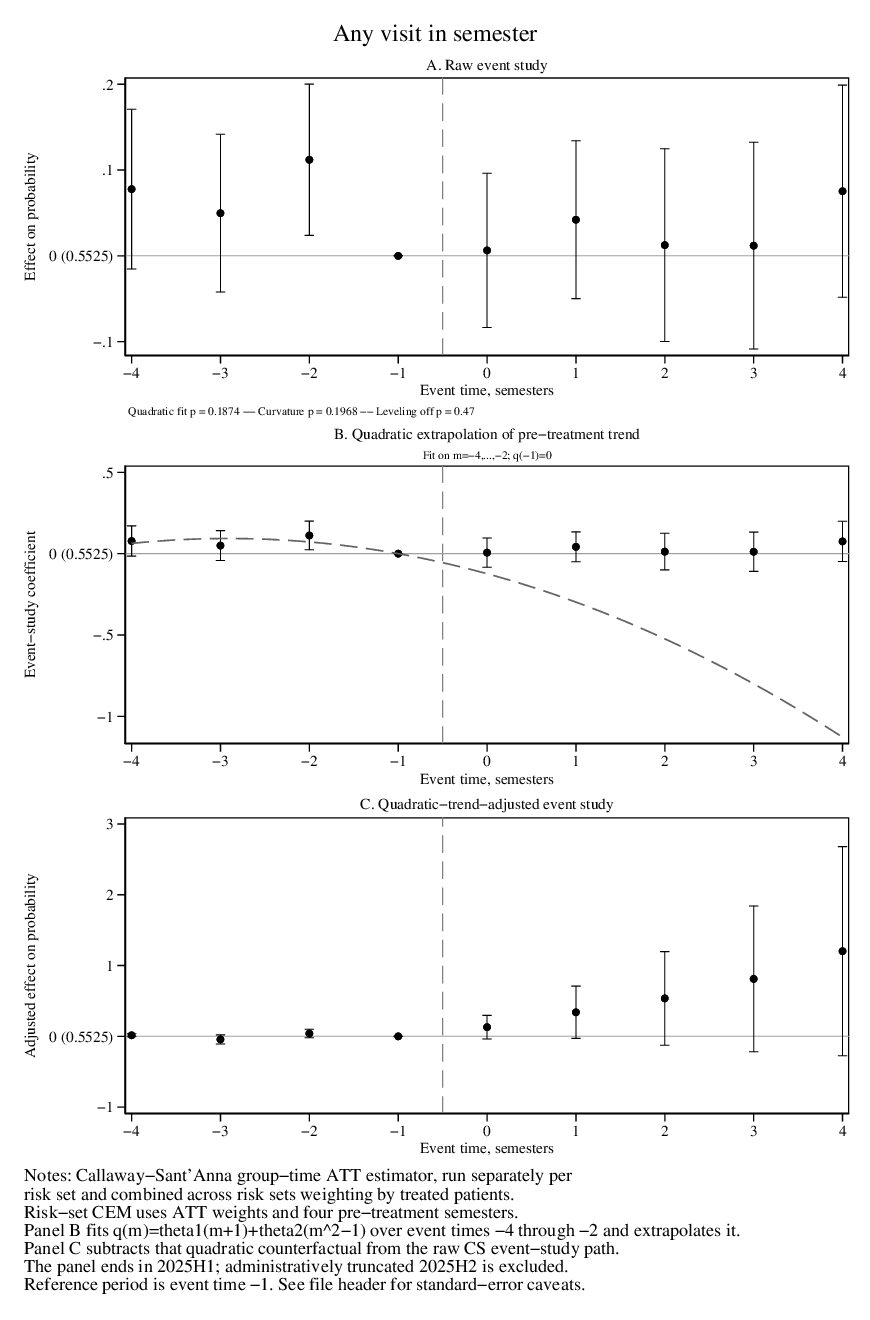}
\caption{Robustness to quadratic trend extrapolation: probability of a
diabetology visit}
\label{fig:cs_quad_visit}
\begin{minipage}{0.90\textwidth}
\footnotesize
\emph{Notes:} Callaway--Sant'Anna group-time ATT estimates for the
probability of a diabetology visit, by semesters relative to adoption
($m=-1$ omitted). Panel A: unadjusted estimates. Panel B: quadratic
pre-trend fitted over $m=-4,-3,-2$, normalized to zero at $m=-1$.
Panel C: quadratic-adjusted estimates. Risk-set CEM weights; standard
errors clustered at the patient level.
\end{minipage}
\end{figure}

\begin{figure}[!htbp]
\centering
\includegraphics[width=0.75\textwidth]{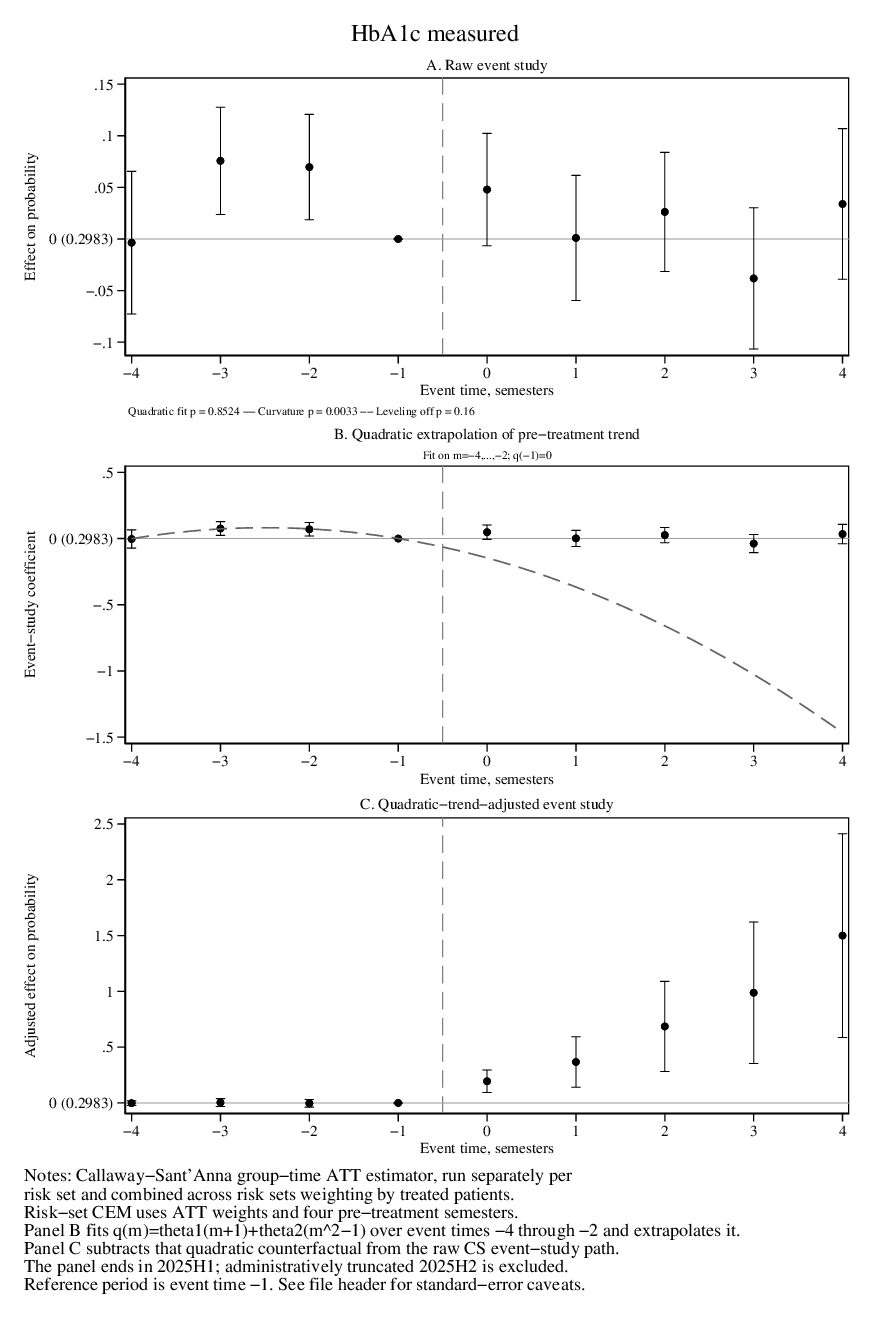}
\caption{Robustness to quadratic trend extrapolation: HbA1c measurement}
\label{fig:cs_quad_hba}
\begin{minipage}{0.90\textwidth}
\footnotesize
\emph{Notes:} Callaway--Sant'Anna group-time ATT estimates for HbA1c, by
semesters relative to adoption ($m=-1$ omitted). Panel A: unadjusted
estimates. Panel B: quadratic pre-trend fitted over $m=-4,-3,-2$,
normalized to zero at $m=-1$. Panel C: quadratic-adjusted estimates.
Risk-set CEM weights; standard errors clustered at the patient level.
\end{minipage}
\end{figure}

\begin{figure}[!htbp]
\centering
\includegraphics[width=0.75\textwidth]{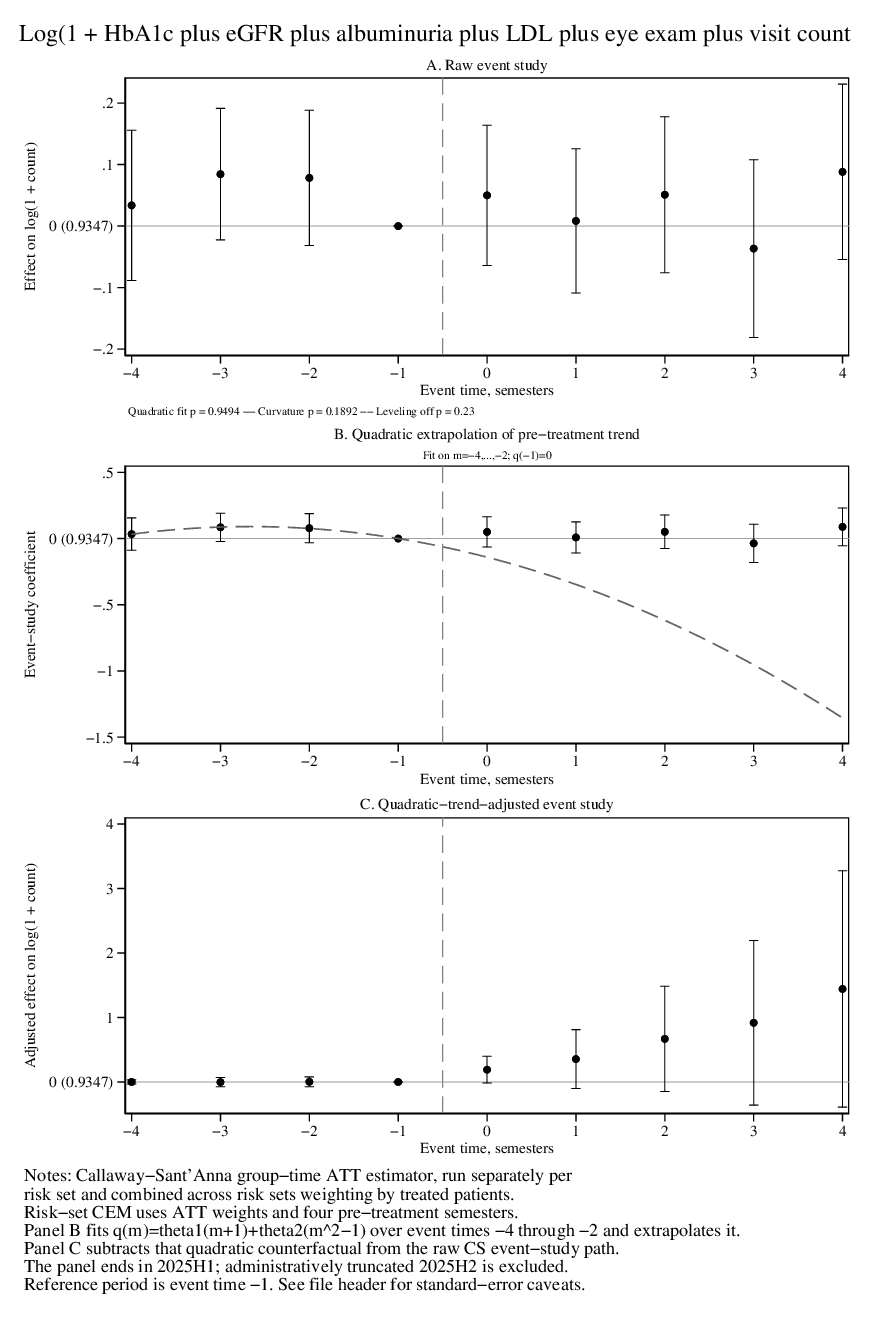}
\caption{Robustness to quadratic trend extrapolation: composite
process-of-care index}
\label{fig:cs_quad_process}
\begin{minipage}{0.90\textwidth}
\footnotesize
\emph{Notes:} Callaway--Sant'Anna group-time ATT estimates for the
composite process-of-care index ($\log(1+\text{recorded
process-of-care activities})$), by semesters relative to adoption
($m=-1$ omitted). Panel A: unadjusted estimates. Panel B: quadratic
pre-trend fitted over $m=-4,-3,-2$, normalized to zero at $m=-1$.
Panel C: quadratic-adjusted estimates. Risk-set CEM weights; standard
errors clustered at the patient level.
\end{minipage}
\end{figure}

\clearpage

\renewcommand{\thefigure}{F\arabic{figure}}
\renewcommand{\thetable}{F\arabic{table}}
\renewcommand{\theHfigure}{F\arabic{figure}}
\renewcommand{\theHtable}{F\arabic{table}}
\setcounter{figure}{0}
\setcounter{table}{0}


\appendix
\section*{Appendix}

\renewcommand{\thefigure}{A\arabic{figure}}
\renewcommand{\thetable}{A\arabic{table}}
\renewcommand{\theHfigure}{A\arabic{figure}}
\renewcommand{\theHtable}{A\arabic{table}}
\setcounter{figure}{0}
\setcounter{table}{0}


\renewcommand{\thesection}{\Alph{section}}
\section{Derivations for the task-based model}
\label{app:task_model}

This appendix derives the task-based framework summarized in
Section~\ref{subsec:conceptual_framework}: the clinical-demand
equation~\eqref{eq:clinical_demand_simple}, the
decomposition~\eqref{eq:net_clinical_demand}, the comparative statics
(Proposition~\ref{prop:comparative_statics}), and the dynamic prediction
(Corollary~\ref{cor:impact-dynamics}) tested in
Section~\ref{sec:empirical_strategy}.

\subsection*{Environment}
\label{app:environment}

A patient's glycemic control $H$, the inverse of glycemic risk, is
produced from a continuum of management tasks indexed by $x\in[N-1,N]$,
aggregated with a Cobb--Douglas technology,
\begin{equation}
H
=
\exp\left\{
\int_{N-1}^{N}\ln y(x)\,dx
\right\},
\label{eq:app_aggregator}
\end{equation}
so every task is essential under this aggregator: dosing, monitoring, and
periodic review enter as complements rather than substitutes. We take this
as a modeling assumption that captures the idea that diabetes management
tasks are complementary, not as a property established independently by
the clinical literature. Each task can be produced with device input $d(x)$ or
clinical input $m(x)$, the specialist visits and laboratory monitoring our
outcomes record, up to an automation frontier $I\in(N-1,N)$:
\begin{equation}
y(x)
=
\begin{cases}
A_D\,\gamma_D(x)\,d(x)+A_M\,\gamma_M(x)\,m(x),
& x\in[N-1,I],\\[4pt]
A_M\,\gamma_M(x)\,m(x),
& x\in(I,N],
\end{cases}
\label{eq:app_technology}
\end{equation}
where $A_D$, $A_M$ are input-specific productivities and $\gamma_D(x)$,
$\gamma_M(x)$ are task-specific productivity schedules. The patient's own
time is folded into the device side of \eqref{eq:app_technology}
deliberately, since the data record only the clinical half of the human
input. One consequence follows directly: because every task above $I$
draws on clinical input alone, the model has no channel through which a
rise in patient time-productivity acts within a fixed task range;
own-time substitution of the kind Grossman (1972) describes enters only
insofar as it shifts the frontier itself, $\Delta I$. The substitution
hypothesis is therefore tested through the net task-content and scale
margins in equation~\eqref{eq:net_clinical_demand}, not through a
separate patient-time input.

Throughout, we restrict attention to configurations and shifts that
preserve an interior frontier,
\begin{equation}
N-1<I<N,
\qquad\text{equivalently}\qquad
N-I\in(0,1),
\label{eq:app_interiority}
\end{equation}
at every point evaluated, including the shifted point $(I+\Delta I,\,
N+\Delta N)$ in Proposition~\ref{prop:comparative_statics}(c) and each
point $(I_m,N_m)$ along the adoption path in
Section~\ref{app:dynamics}, so that $c_H(I,N)$ in \eqref{eq:app_cH} and
the linearization underlying Proposition~\ref{prop:comparative_statics}
and Corollary~\ref{cor:impact-dynamics} remain well defined.

\begin{assumption}[Comparative advantage]
\label{ass:comparative_advantage}
The ratio $\gamma_M(x)/\gamma_D(x)$ is continuous and strictly increasing
in $x$.
\end{assumption}

Tasks are ordered by the judgment they require. Low-$x$ tasks are routine
and codifiable, such as basal-rate adjustments and correction doses.
High-$x$ tasks call for interpretation, such as reading glucose
variability and revising treatment in response to device data. Let $r_M$
denote the shadow cost of a unit of clinical input and $p_D$ the effective
per-unit cost of device services. Diabetes carries a chronic-disease
exemption from the standard NHS co-payment (\emph{ticket}) for services
tied to its diagnosis and management, so the patients in our sample face
little or no point-of-use monetary cost for routine diabetes care.
Consequently $r_M$ reflects clinic capacity and the opportunity cost of
specialist time rather than a household budget, not the absence of
point-of-use pricing across the NHS generally.

\begin{assumption}[Cost advantage below the frontier]
\label{ass:cost_advantage}
At the automation frontier,
\begin{equation}
\frac{p_D}{A_D\,\gamma_D(I)}
\;\leq\;
\frac{r_M}{A_M\,\gamma_M(I)}.
\label{eq:app_cost_advantage}
\end{equation}
\end{assumption}

Under Assumptions~\ref{ass:comparative_advantage}
and~\ref{ass:cost_advantage}, cost minimization assigns every task
$x\leq I$ to the device and every task $x>I$ to clinical input: the
device is cheaper whenever $\gamma_M(x)/\gamma_D(x)\leq r_M A_D/(p_D A_M)$,
a ratio that rises in $x$ by Assumption~\ref{ass:comparative_advantage}
and satisfies this inequality at $x=I$ by
Assumption~\ref{ass:cost_advantage}, hence for every $x\leq I$; tasks
above $I$ cannot be automated by \eqref{eq:app_technology} and use
clinical input instead. Unit task cost is therefore
\begin{equation}
q(x)
=
\begin{cases}
\dfrac{p_D}{A_D\,\gamma_D(x)},
& x\in[N-1,I],\\[8pt]
\dfrac{r_M}{A_M\,\gamma_M(x)},
& x\in(I,N].
\end{cases}
\label{eq:app_unit_task_cost}
\end{equation}

\subsection*{Cost, demand, and the decomposition of $\Delta\ln M$}
\label{app:cost_demand}

Minimizing $\int q(x)y(x)\,dx$ subject to the aggregator constraint gives
the pointwise condition $q(x)y(x)=\lambda$ for every $x$: expenditure is
equalized across tasks, the defining property of the Cobb--Douglas
aggregator. Substituting back and using the fact that the task range has
measure one,
\begin{equation}
\lambda=c_H\,H,
\qquad
c_H(I,N)\equiv\exp\left\{
\int_{N-1}^{I}\ln\frac{p_D}{A_D\gamma_D(x)}\,dx
+
\int_{I}^{N}\ln\frac{r_M}{A_M\gamma_M(x)}\,dx
\right\},
\label{eq:app_cH}
\end{equation}
so $c_H$ is the minimized unit cost of glycemic control.

For tasks $x\in(I,N]$, expenditure on clinical input is
$r_M\,m(x)=\lambda=c_HH$, so total clinical demand integrates over the
clinical task range:
\begin{equation}
M
=
\int_{I}^{N} m(x)\,dx
=
(N-I)\,\frac{c_H\,H}{r_M}.
\label{eq:app_M}
\end{equation}
Evaluated at the targeted level $H=H^{*}$, this is
equation~\eqref{eq:clinical_demand_simple} in the main text.

We interpret $H^{*}$ as the treatment intensity that the care team
chooses given the AUSL's resource constraints, weighing the clinical
benefit of tighter control against the resource cost of producing it,
rather than as a household demand decision by a patient who does not bear
$c_H$ directly at the point of use. Under this interpretation the
targeted level of control solves $B'(H^{*})=c_H$ for a strictly concave
benefit function $B(\cdot)$, which gives the elasticity
\begin{equation}
\varepsilon
\;\equiv\;
-\frac{d\ln H^{*}}{d\ln c_H}
=
-\frac{B'(H^{*})}{H^{*}B''(H^{*})}
\;>\;0,
\label{eq:app_epsilon}
\end{equation}
equation~\eqref{eq:control_elasticity_simple} in the main text: the pace
at which the care team tightens the plan once tighter targets become
attainable, not a household price response. Since $\varepsilon>0$ for any
strictly concave $B$, $H^{*}$ itself always rises when $c_H$ falls,
whatever the value of $\varepsilon$; what depends on whether
$\varepsilon\gtrless1$ is only whether the resource component $c_H H^{*}$,
and hence $M$, rises or falls. Under a constant-elasticity specification
$H^{*}=K\,c_H^{-\varepsilon}$, differencing $\ln M$ between the pre- and
post-adoption states gives
\begin{equation}
\Delta\ln M
=
\ln\!\left(\frac{N_1-I_1}{N_0-I_0}\right)
+
(1-\varepsilon)\,\Delta\ln c_H
-
\Delta\ln r_M,
\label{eq:app_decomposition}
\end{equation}
equation~\eqref{eq:net_clinical_demand} in the main text. The
decomposition is exact under constant elasticity and holds to first order
otherwise, with $\varepsilon$ evaluated at the pre-adoption cost.

\subsection*{Comparative statics}
\label{app:comparative_statics}

Define $\pi_I\equiv\ln[p_D/(A_D\gamma_D(I))]-\ln[r_M/(A_M\gamma_M(I))]\leq0$
and $\pi_N\equiv\ln[r_M/(A_M\gamma_M(N))]-\ln[p_D/(A_D\gamma_D(N-1))]$, the
marginal log-cost effects of the frontier and of the task range on unit
cost. We avoid calling $\pi_I,\pi_N$ elasticities, since $I$ and $N$ enter
in levels rather than logs; they are semi-elasticities of $c_H$ with
respect to the task boundaries. Here $\pi_I\leq0$
follows from Assumption~\ref{ass:cost_advantage}: automation
is adopted only where it weakly lowers cost. Neither
Assumption~\ref{ass:comparative_advantage} nor
Assumption~\ref{ass:cost_advantage} restricts $\pi_N$, since both compare
device and clinical input at the same task $x=I$, while $\pi_N$ compares
a clinical task at $N$ against a device task at $N-1$.

\begin{proposition}[Comparative statics of clinical demand]
\label{prop:comparative_statics}
\label{prop:displacement}
\label{prop:reinstatement}
\label{prop:net} 
(a) Holding $N$ fixed,
\begin{equation}
\frac{\partial\ln M}{\partial I}
=
-\frac{1}{N-I}
+
(1-\varepsilon)\,\pi_I,
\label{eq:app_prop1}
\end{equation}
negative without qualification except through the productivity term $(1-\varepsilon)\pi_I$, which can offset it when $\varepsilon>1$. (b) Holding $I$ fixed,
\begin{equation}
\frac{\partial\ln M}{\partial N}
=
\frac{1}{N-I}
+
(1-\varepsilon)\,\pi_N,
\label{eq:app_prop2}
\end{equation}
positive whenever $\varepsilon\leq1$ and $\pi_N\geq0$. More generally, define
\begin{equation}
\varepsilon^{*}
\equiv
1+\frac{1}{(N-I)\,\pi_N},
\qquad \pi_N\neq0.
\label{eq:app_epsilon_star}
\end{equation}
Expression~\eqref{eq:app_prop2} is positive if and only if $\varepsilon<\varepsilon^{*}$ when $\pi_N>0$, and if and only if $\varepsilon>\varepsilon^{*}$ when $\pi_N<0$; it is positive for every $\varepsilon$ when $\pi_N=0$. Since $\pi_N$ is unsigned by Assumptions~\ref{ass:comparative_advantage}--\ref{ass:cost_advantage} (Section~\ref{subsec:conceptual_framework}), part (b) is not positive without qualification: raising $N$ retires the lowest-index task in the range, at $N-1$, and replaces it with a new highest-index task at $N$, where clinical input holds the comparative advantage; whether that compositional shift raises clinical demand depends jointly on the sign of $\pi_N$ and on $\varepsilon$. (c) For a simultaneous shift $(\Delta I,\Delta N)>0$, to first order
\begin{equation}
\Delta\ln M
=
\frac{\Delta N-\Delta I}{N-I}
+
(1-\varepsilon)\bigl(\pi_I\,\Delta I+\pi_N\,\Delta N\bigr).
\label{eq:app_prop3}
\end{equation}
\end{proposition}

\begin{proof}
Differentiate $\ln M$ term by term. The task-content term gives $\mp1/(N-I)$ depending on which boundary moves. Leibniz's rule applied to \eqref{eq:app_cH} gives $\partial\ln c_H/\partial I=\pi_I$ and $\partial\ln c_H/\partial N=\pi_N$: moving $I$ reassigns the marginal task between device and clinical production, while moving $N$ retires the lowest-index task from the range and admits a new highest-index task in its place. The chain rule gives $\partial\ln H^{*}/\partial I=-\varepsilon\pi_I$ and $\partial\ln H^{*}/\partial N=-\varepsilon\pi_N$. Summing gives (a) and (b); (c) sums (a) weighted by $\Delta I$ and (b) weighted by $\Delta N$. For (b), solving $1/(N-I)+(1-\varepsilon)\pi_N>0$ for $\varepsilon$ gives threshold~\eqref{eq:app_epsilon_star} and the two sign conditions stated above, obtained by dividing through by $\pi_N$ and flipping the inequality when $\pi_N<0$. Since $N>I$ by construction, $N-I>0$, so $\varepsilon^{*}>1$ whenever $\pi_N>0$, and the sufficient condition $\varepsilon\leq1$ together with $\pi_N\geq0$ always lies inside the region where~\eqref{eq:app_prop2} is positive.
\end{proof}

To first order, part (c) is positive if and only if task creation, net
of the scale response through $c_H$, outweighs displacement, a statement
about the linear approximation's sign, not the exact, finite-shift
change. The omitted terms are second order in $(\Delta I,\Delta N)$ and
shrink only as the shift itself vanishes, not merely because it is
small: the exact task-content change is $\ln[(N_1-I_1)/(N_0-I_0)]$
rather than $(\Delta N-\Delta I)/(N-I)$, and the exact change in $c_H$
integrates $\pi_I$ and $\pi_N$ along the path from $(I_0,N_0)$ to
$(I_1,N_1)$ rather than evaluating them once at the pre-adoption point.
We treat the linear approximation as reasonable for a shift reallocating
one task among the many in the unit task range, abrupt in calendar time
though it is, but not as exact. Because $\gamma_D(\cdot)$ and
$\gamma_M(\cdot)$ are unobserved, the linearization cannot be tested
directly.

\subsection*{Adoption dynamics}
\label{app:dynamics}

Proposition~\ref{prop:comparative_statics}(c) compares two fixed points.
The event study in Section~\ref{sec:empirical_strategy} recovers a full
dynamic path $\{\tau_m\}_{m\geq 0}$, so the model needs to say what that
path looks like. The automation margin moves the moment the algorithm is
switched on, $I_m=I_0+\Delta I$ for every $m\geq 0$, while the clinical
margin moves more slowly, so that
\begin{equation}
N_m
=
N_0
+
\Delta N^{\mathrm{impact}}
+
\Delta N^{\mathrm{learning}}
\left(1-e^{-\Psi m}\right),
\qquad m\geq 0,
\quad \Psi>0,
\label{eq:app_Nm}
\end{equation}
with $I_m=I_0$ and $N_m=N_0$ for $m<0$. The parameter $\Psi$ governs the
rate at which complementary clinical tasks accumulate after adoption. We
refer to this as a learning process for convenience, since converting a
continuous glucose and insulin trace into treatment changes is a skill
clinicians plausibly build through repeated encounters with a given
patient's device data, but the reduced form in (25) does not by itself identify clinician learning as opposed to other channels operating at a similar pace, such as administrative reauthorization or the gradual accumulation of patient-specific monitoring history. In the empirical implementation, $m=0$ corresponds to the first observed AID-related device supply rather than to confirmed activation of the control algorithm, so the model's instantaneous automation margin is an upper bound on how sharply exposure begins in the data.

\paragraph{Formal statement of Corollary~\ref{cor:impact-dynamics}.}
Define
\[
A
=
\frac{1}{N_0-I_0}
+
(1-\varepsilon)\pi_N,
\qquad
D
=
\frac{1}{N_0-I_0}
-
(1-\varepsilon)\pi_I,
\]
where $A$ is the marginal contribution of a newly created task and $D$
is the marginal contribution of an automated task. The two sufficient
conditions referenced in Corollary~\ref{cor:impact-dynamics} are:
\begin{enumerate}[label=(C\arabic*)]
    \item $A>0$, the condition under which
    Proposition~\ref{prop:reinstatement} delivers a positive net effect
    of task creation; and
    \item
    \[
    A\,\Delta N^{\mathrm{impact}} \leq D\,\Delta I,
    \]
    so that task creation at initiation is small relative to
    displacement.
\end{enumerate}
Under $\varepsilon\leq 1$,
$
D
\geq
\tfrac{1}{N_0-I_0}
>0 .
$
If, in addition, $\pi_I+\pi_N\leq 0$, so that an equal marginal expansion
of the frontier and the task range ($dI=dN$) weakly reduces the unit cost
of control, then condition~\textnormal{(C2)} is implied by the simpler
restriction
$
\Delta N^{\mathrm{impact}}\leq\Delta I .
$
In the special case $\varepsilon=1$ the scale term disappears entirely,
$A=D$, and condition~\textnormal{(C2)} reduces directly to
$
\Delta N^{\mathrm{impact}}\leq\Delta I .
$

\begin{proof}[Proof of Corollary~\ref{cor:impact-dynamics}]
Linearizing Proposition~\ref{prop:comparative_statics}(c) around the
pre-adoption state $(I_0,N_0)$ and substituting the dynamic path
$\Delta N_m=N_m-N_0$ from equation~\eqref{eq:app_Nm} together with the
constant $\Delta I=I_m-I_0$ gives
\begin{equation}
\tau_m
\;\approx\;
A\,\bigl[\Delta N^{\mathrm{impact}}
+\Delta N^{\mathrm{learning}}\bigl(1-e^{-\Psi m}\bigr)\bigr]
\;-\;
D\,\Delta I,
\label{eq:app_taupath}
\end{equation}
with $A$ and $D$ as defined above, fixed at their pre-adoption values,
and $\Delta I>0$.

\emph{Impact.} At $m=0$ the learning term vanishes and
$\tau_0=A\Delta N^{\mathrm{impact}}-D\Delta I$. Condition (C2) states
directly that $A\Delta N^{\mathrm{impact}}\leq D\Delta I$, hence
$\tau_0\leq 0$.

\emph{Monotonicity and convergence.} Differentiating
equation~\eqref{eq:app_taupath},
\[
\frac{\partial\tau_m}{\partial m}
=
\Psi\,e^{-\Psi m}\,\Delta N^{\mathrm{learning}}\,A
>0,
\]
by (C1) and $\Delta N^{\mathrm{learning}}>0$. Since $m$ indexes semesters,
the relevant object is the discrete difference, which carries the same
sign,
\[
\tau_{m+1}-\tau_m
=
A\,\Delta N^{\mathrm{learning}}\,e^{-\Psi m}\bigl(1-e^{-\Psi}\bigr)
>0,
\]
so the derivative above is shorthand for this discrete monotonicity
result. The path rises strictly in $m$ and converges exponentially at
rate $\Psi$ to the long-run level
$A\bigl(\Delta N^{\mathrm{impact}}+\Delta N^{\mathrm{learning}}\bigr)
-D\Delta I$ implied by the first-order approximation to
Proposition~\ref{prop:comparative_statics}(c).

\emph{Sufficient conditions.} Assumption~\ref{ass:cost_advantage} gives
$\pi_I\leq 0$, so $\varepsilon\leq 1$ implies
$D\geq \tfrac{1}{N_0-I_0}>0$. Moreover,
\[
D-A
=
-(1-\varepsilon)\bigl(\pi_I+\pi_N\bigr)
\;\geq\; 0
\qquad\text{whenever }\varepsilon\leq 1
\text{ and }\pi_I+\pi_N\leq 0.
\]
In that case $\Delta N^{\mathrm{impact}}<\Delta I$ yields
$A\Delta N^{\mathrm{impact}}<A\Delta I\leq D\Delta I$, so (C2) holds. When
$\varepsilon=1$, $A=D=\tfrac{1}{N_0-I_0}$ and (C2) reduces directly to
$\Delta N^{\mathrm{impact}}\leq\Delta I$.
\end{proof}
\renewcommand{\thefigure}{B\arabic{figure}}
\renewcommand{\thetable}{B\arabic{table}}
\renewcommand{\theHfigure}{B\arabic{figure}}
\renewcommand{\theHtable}{B\arabic{table}}
\setcounter{figure}{0}
\setcounter{table}{0}
\renewcommand{\thesection}{\Alph{section}}

\section{Identification with event-time trend extrapolation}
\label{app:identification}

This appendix states why an unrestricted time-varying confound leaves the post-adoption path unidentified, how Assumption~\ref{ass:event_time_extrapolation} restores identification, how an onboarding effect in the reference period enters the trend-adjusted estimates, and how the diagnostics reported in Section~\ref{subsec:diagnostics} are defined. Throughout, $\beta_{gm}$, $\tau_{gm}$, $\widetilde w_{gm}$, $\theta_m$, $\tau_m$, and $c(m)$ are the population objects defined in Section~\ref{subsec:trend_extrapolation}.

\subsection{Non-identification under an unrestricted confound}
\label{app:unrestricted_not_identified}

Consider the stacked event-study regression \eqref{eq:eventstudy}. The term $C_{it}$ is not observed, so the population coefficient on $D^m_{itg}$ absorbs any part of $C_{it}$ that evolves differently for cohort-$g$ adopters and their CEM-weighted controls. Let $\mathrm E_\omega$ denote an expectation taken with the CEM weights $\omega_{ig}$ over the risk-set observation set $\mathcal O_g$. Define the residual bias
\begin{equation}
b_{gm}
=
\mathrm E_\omega\!\big[C_{i,g+m}-C_{i,g-1}\,\big|\,E_i=g,\ i\in\mathcal M_g\big]
-
\mathrm E_\omega\!\big[C_{i,g+m}-C_{i,g-1}\,\big|\,i\in\mathcal R_g,\ E_i>g+m\big],
\label{eq:app_bias}
\end{equation}
the difference between adopters and controls in how the confound changes from the reference semester to event time $m$. Then $\beta_{gm}=\tau_{gm}+b_{gm}$. Aggregating with the cohort shares $\widetilde w_{gm}$ gives equation~\eqref{eq:decomposition},
\begin{equation*}
\theta_m=\tau_m+c(m),
\qquad
c(m)=\sum_{g\in\mathcal G_m}\widetilde w_{gm}\,b_{gm}.
\end{equation*}

The data identify $\{\theta_m\}$ but not its two components separately. To see this, take any candidate treatment path $\{\tau^\ast_m\}$ and set $c^\ast(m)=\theta_m-\tau^\ast_m$. The pair $(\tau^\ast,c^\ast)$ produces exactly the same $\theta_m$ at every $m$ as the true pair $(\tau,c)$. Because nothing in the data distinguishes the two pairs, the path $\{\tau_m:m\ge0\}$ is not identified unless $c(\cdot)$ is restricted. Identification is therefore an assumption about the confound and not a property of the event-study design.

\subsection{Identification under the event-time restriction}
\label{app:low_dimensional_confounding}

Two restrictions together identify the post-adoption path.

\paragraph{No anticipation.} $\tau_{gm}=0$ for all $g$ and all $m<0$. Adoption cannot move outcomes before it happens.

\paragraph{Event-time extrapolation.} Assumption~\ref{ass:event_time_extrapolation}: $c(m)=\phi'\widetilde f(m)$ with $\widetilde f(m)=f(m)-f(-1)$. The baseline uses $f(m)=m$, so $\widetilde f(m)=m+1$ and
\begin{equation}
c(m)=\rho\,(m+1).
\label{eq:app_linear}
\end{equation}

\noindent\textit{Step 1 (normalization).} Since $\beta_{g,-1}=0$ and, under no anticipation, $\tau_{g,-1}=0$, we have $b_{g,-1}=0$ for every $g$ and hence $c(-1)=0$. Equation~\eqref{eq:app_linear} satisfies this for any $\rho$, so the normalization places no restriction on $\rho$.

\noindent\textit{Step 2 (pre-adoption periods identify $\rho$).} For $m<0$, no anticipation gives $\tau_m=0$, so $\theta_m=c(m)=\rho(m+1)$. Any single pre-adoption period with $m\le-2$ identifies $\rho=\theta_m/(m+1)$, because $m+1\neq0$. With several such periods, $\rho$ is overidentified.

\noindent\textit{Step 3 (post-adoption periods).} For $m\ge0$, $\theta_m=\tau_m+\rho(m+1)$. Subtracting the known quantity $\rho(m+1)$ gives $\tau_m=\theta_m-\rho(m+1)$, which is identified.

\noindent\textit{Estimation.} The baseline fits $\rho$ on $\theta_{-3}$ and $\theta_{-2}$, with the line constrained to equal zero at $m=-1$, jointly with the event-study coefficients by GMM (Section~\ref{subsec:trend_extrapolation}). Let $x=(x_{-3},x_{-2})'=(-2,-1)'$ collect the values of $m+1$ at the fitted points, $\widehat\theta_{\mathrm{pre}}=(\widehat\theta_{-3},\widehat\theta_{-2})'$, and $W$ the $2\times2$ weighting matrix implied by the estimator. The slope estimate is
\begin{equation}
\widehat\rho=\frac{x'W\widehat\theta_{\mathrm{pre}}}{x'Wx},
\label{eq:app_rho}
\end{equation}
and the trend-adjusted coefficient at any modeled event time is
\begin{equation}
\widehat\tau_m=\widehat\theta_m-\widehat\rho\,(m+1),
\qquad m\neq-1.
\label{eq:app_tau}
\end{equation}
For $m\ge0$ this is equation~\eqref{eq:trendadjusted}. For $m<0$ it is the residual used in the adjusted pre-period test of Section~\ref{app:diagnostic_tests}. The coefficient at $m=-4$ does not enter \eqref{eq:app_rho}, so $\widehat\tau_{-4}$ is a held-out residual.

\subsection{An onboarding effect in the reference period}
\label{app:admin_effect}

Suppose training and documentation before device supply raise the outcome of adopters by $\delta$ in the semester before recorded adoption, and in no other semester. Because onboarding adds visits and tests, $\delta\ge0$. This is a violation of no anticipation at $m=-1$ only. Maintain no anticipation at $m\le-2$ and a linear confound. We derive how $\delta$ enters $\widehat\tau_m$.

\noindent\textit{Step 1 (coefficients relative to a contaminated reference).} Every coefficient is a contrast with $m=-1$. Adopters' reference-period level is raised by $\delta$, so the contrast at every $m\neq-1$ falls by $\delta$:
\begin{equation}
\theta_m=\tau_m+\rho(m+1)-\delta,
\qquad m\neq-1,
\label{eq:app_admin_theta}
\end{equation}
where $\rho$ is the true slope of the confound. The coefficient at $m=-1$ remains zero by normalization.

\noindent\textit{Step 2 (pre-adoption coefficients).} For $m\le-2$, $\tau_m=0$, so $\theta_{\mathrm{pre}}=\rho\,x-\delta\,\iota$, with $\iota=(1,1)'$.

\noindent\textit{Step 3 (fitted slope).} Substituting into \eqref{eq:app_rho},
\begin{equation*}
\widehat\rho
=
\frac{x'W(\rho x-\delta\iota)}{x'Wx}
=
\rho-\delta\,\frac{x'W\iota}{x'Wx}
=
\rho+\kappa\,\delta,
\qquad
\kappa\equiv-\frac{x'W\iota}{x'Wx}.
\end{equation*}

\noindent\textit{Step 4 (bias in the adjusted estimates).} For $m\ge0$, combine \eqref{eq:app_tau}, \eqref{eq:app_admin_theta}, and Step~3:
\begin{equation}
\widehat\tau_m
=
\tau_m+\rho(m+1)-\delta-(\rho+\kappa\delta)(m+1)
=
\tau_m-\delta\big[1+\kappa(m+1)\big].
\label{eq:app_admin_bias}
\end{equation}

\noindent\textit{Step 5 (the value of $\kappa$).} With a diagonal weighting matrix $W=\mathrm{diag}(w_{-3},w_{-2})$, $w_{-3},w_{-2}>0$, we have $x'Wx=4w_{-3}+w_{-2}$ and $x'W\iota=-(2w_{-3}+w_{-2})$, so
\begin{equation*}
\kappa=\frac{2w_{-3}+w_{-2}}{4w_{-3}+w_{-2}}.
\end{equation*}
The numerator is smaller than the denominator, so $\kappa<1$. Twice the numerator, $4w_{-3}+2w_{-2}$, exceeds the denominator, so $\kappa>1/2$. With equal weights, $\kappa=3/5$ and the bias equals $-\delta(8+3m)/5$. For a non-diagonal $W$, $\kappa$ follows from the estimated weighting matrix.

\noindent\textit{Implications.} Whenever $\kappa>0$, the bias in \eqref{eq:app_admin_bias} is non-positive and grows in absolute value with $m$. An onboarding effect in the reference period therefore pushes the trend-adjusted estimates down, by more at longer horizons, and cannot by itself generate a positive, rising post-adoption path. For a given $\kappa$, the size of the bias is governed by $\delta$.

\paragraph{When the effect cancels.} The effect cancels only if the reference-period point plays no role in the fit. Fit an unconstrained line $a+s\,m$ on periods $m\le-2$ and take $\widetilde\tau_m=\widehat\theta_m-(\widehat a+\widehat s\,m)$. By \eqref{eq:app_admin_theta}, every fitted point equals $\rho(m+1)-\delta=(\rho-\delta)+\rho\,m$, which lies exactly on a line with intercept $\rho-\delta$ and slope $\rho$. Any fit with at least two distinct points therefore returns $\widehat a=\rho-\delta$ and $\widehat s=\rho$ in the population. For $m\ge0$, $\widetilde\tau_m=\tau_m+\rho(m+1)-\delta-(\rho-\delta)-\rho m=\tau_m$, so $\delta$ cancels exactly. The baseline differs because its line is forced through zero at $m=-1$, a point that is itself shifted by $\delta$ relative to the others.

\paragraph{Estimating $\delta$ by re-normalization.} Re-estimate the raw event study with reference period $m=-4$. Coefficients are now contrasts with $m=-4$, which onboarding does not affect. Under no anticipation at $m\le-2$ and a linear confound, the coefficient at $m$ is $\theta^{(-4)}_m=\rho(m+4)+\delta\,\mathbf 1\{m=-1\}$ for $m<0$. A line through zero at $m=-4$, fitted on periods before $m=-2$, recovers $\rho$. The departure of $\theta^{(-4)}_{-1}$ from that line estimates $\delta$, and the departure of $\theta^{(-4)}_{-2}$ should be zero. A nonzero departure at $m=-2$ is inconsistent with an effect confined to $m=-1$. It could instead reflect curvature in the confound or anticipation beyond the reference period, and the exercise cannot separate these two explanations. If a departure appears only at $m=-1$, \eqref{eq:app_admin_bias} with the estimated $\delta$ gives the implied bias under the baseline window. If a departure appears at $m=-2$ as well, the fitting window should move to periods before $m=-2$ with a free intercept, where the cancellation result above applies.

\subsection{Diagnostic test definitions}
\label{app:diagnostic_tests}

This subsection defines the diagnostics described in Section~\ref{subsec:trend_extrapolation} and reported in Sections~\ref{sec:results} and~\ref{subsec:diagnostics}. All tests are Wald tests based on the patient-clustered variance-covariance matrix of the jointly estimated coefficients.

\textit{Raw pre-trends test.} $H_0:\theta_m=0$ for $m\in\{-4,-3,-2\}$. Flat pre-trends are compatible with Assumption~\ref{ass:event_time_extrapolation}, as the case $\rho=0$, but the assumption does not require them. Failing to reject is not sufficient for the design to be valid, since it says nothing about the post-adoption continuation of the confound. Rejecting shows that the unadjusted event study would be misleading. It does not by itself show that the extrapolation design is invalid.

\textit{Adjusted joint pre-period test.} $H_0:\tau_m=0$ for $m\in\{-4,-3,-2\}$, evaluated at $\widehat\tau_m$ from \eqref{eq:app_tau}, with the variance-covariance matrix propagated through the trend-fitting step. One slope is estimated, so under $H_0$ the statistic has two degrees of freedom. Within the fitting window the test has little power against curvature, because one slope is fitted to two coefficients, and much of its content comes from the held-out residual at $m=-4$. A rejection means the constrained linear path does not describe the pre-adoption coefficients. That could reflect curvature in the confound, an onboarding effect in the reference period, or a departure confined to the held-out period. The Rambachan and Roth (2023) analysis in Section~\ref{subsec:diagnostics} complements this test by reporting how large a change in the slope of the differential trend the post-adoption conclusions can tolerate. It is applied to the raw coefficients and does not bound the specific departure detected here.

\textit{Leveling-off test.} $H_0:\bar\tau_{\mathcal L}=\bar\tau_{\mathcal E}$, where $\bar\tau_{\mathcal E}$ and $\bar\tau_{\mathcal L}$ are the averages of $\widehat\tau_m$ over an early set $\mathcal E=\{0,1\}$ and a late set $\mathcal L=\{3,4\}$ of post-adoption event times. Failing to reject is consistent with a plateau, but confidence intervals widen at longer horizons and the test loses power there. A failure to reject is therefore inconclusive and does not confirm that the path has flattened. This test concerns the dynamic prediction of the task-based model and not the validity of the extrapolated counterfactual.

\renewcommand{\thefigure}{C\arabic{figure}}
\renewcommand{\thetable}{C\arabic{table}}
\renewcommand{\theHfigure}{C\arabic{figure}}
\renewcommand{\theHtable}{C\arabic{table}}
\setcounter{figure}{0}
\setcounter{table}{0}


\section*{Supplementary Tables and Figures}

\begin{table}[!htbp]
\centering
\caption{Baseline characteristics and pre-adoption care levels by ever-AID status}
\label{tab:baseline_characteristics}
\scriptsize
\begin{threeparttable}
\begin{adjustbox}{max width=\textwidth}
\begin{tabular}{lrrrrrr}
\toprule
Variable & N never & N ever & Mean never & Mean ever & Diff. & Std. diff. \\
\midrule
\multicolumn{7}{l}{\textit{Panel A: Demographic and geographic covariates}} \\
Female & 1,325 & 283 & 0.601 & 0.392 & -0.209 & -0.426 \\
Age & 1,325 & 283 & 49.491 & 41.816 & -7.675 & -0.510 \\
Age 65+ & 1,325 & 283 & 0.193 & 0.032 & -0.161 & -0.528 \\
Age 75+ & 1,325 & 283 & 0.075 & 0.000 & -0.075 & -0.404 \\
Mountain area & 1,306 & 282 & 0.006 & 0.007 & 0.001 & 0.012 \\
Hill area & 1,306 & 282 & 0.116 & 0.149 & 0.033 & 0.096 \\
Plain area & 1,306 & 282 & 0.877 & 0.844 & -0.034 & -0.097 \\
Urban city & 1,306 & 282 & 0.377 & 0.472 & 0.095 & 0.193 \\
Suburb / town & 1,306 & 282 & 0.495 & 0.447 & -0.048 & -0.096 \\
Rural area & 1,306 & 282 & 0.129 & 0.082 & -0.047 & -0.154 \\
\midrule
\multicolumn{7}{l}{\textit{Panel B: Clinical risk and prior technology use}} \\
Any Meteda complication & 1,325 & 283 & 0.041 & 0.035 & -0.005 & -0.028 \\
Any hospitalization & 1,325 & 283 & 0.071 & 0.053 & -0.018 & -0.074 \\
Any ER/UCC visit & 1,325 & 283 & 0.198 & 0.166 & -0.032 & -0.082 \\
Any acute use & 1,325 & 283 & 0.215 & 0.191 & -0.024 & -0.060 \\
Hospitalizations plus ER/UCC contacts & 1,325 & 283 & 0.371 & 0.279 & -0.091 & -0.114 \\
Pump active\tnote{a} & 1,276 & 282 & 0.001 & 0.074 & 0.074 & 0.394 \\
Any device active\tnote{a} & 1,276 & 282 & 0.085 & 0.755 & 0.670 & 1.845 \\
\midrule
\multicolumn{7}{l}{\textit{Panel C: Pre-adoption levels of outcome variables}\tnote{b}} \\
Any visit in semester & 1,325 & 283 & 0.395 & 0.544 & 0.149 & 0.301 \\
HbA1c measured & 1,325 & 283 & 0.401 & 0.364 & -0.037 & -0.076 \\
LDL measured & 1,325 & 283 & 0.399 & 0.548 & 0.148 & 0.300 \\
Blood pressure measured & 1,325 & 283 & 0.226 & 0.396 & 0.170 & 0.373 \\
BMI measured & 1,325 & 283 & 0.372 & 0.523 & 0.151 & 0.307 \\
Albuminuria measured & 1,325 & 283 & 0.255 & 0.382 & 0.127 & 0.274 \\
eGFR measured & 1,325 & 283 & 0.368 & 0.378 & 0.010 & 0.020 \\
Eye exam performed & 1,325 & 283 & 0.066 & 0.064 & -0.002 & -0.008 \\
Seven-item measured-process count & 1,325 & 283 & 2.087 & 2.654 & 0.567 & 0.284 \\
\bottomrule
\end{tabular}
\end{adjustbox}
\begin{tablenotes}[flushleft]
\footnotesize
\item Notes: Baseline is the nearest active pre-AID semester for ever-AID patients and the first active device-era semester for never-AID patients. Standardized differences use pooled standard deviations.
\item[a] Pump and device activity at baseline are mechanically elevated among AID adopters because AID is typically prescribed to patients already on pump or CGM therapy (Section~\ref{subsec:background}); these rows should not be read as evidence of an unbalanced confound.
\item[b] \textit{Notes:} Panel C reports pre-adoption levels of individual
process-of-care measures and the broader seven-item measured-process
count. The seven-item count is used to characterize pre-adoption care
intensity and differs from the six-component composite outcome defined
in Section~\ref{subsec:outcomes}: it includes blood-pressure and BMI measurement
and excludes diabetologist visits..
\end{tablenotes}
\end{threeparttable}
\end{table}

\begin{table}[!htbp]
\centering
\caption{Event-time support among treated observations}
\label{tab:a5_eventtime_support}
\scriptsize
\renewcommand{\arraystretch}{0.90}
\begin{adjustbox}{width=\textwidth, max totalheight=0.82\textheight, keepaspectratio}
\begin{threeparttable}
\begin{tabular}{rrrr}
\toprule
Event time & Rows & Person-semesters & Main window \\
\midrule
-19 & 10  & 7.753   & 0 \\
-18 & 34  & 27.152  & 0 \\
-17 & 42  & 35.424  & 0 \\
-16 & 60  & 51.912  & 0 \\
-15 & 77  & 66.825  & 0 \\
-14 & 95  & 85.283  & 0 \\
-13 & 125 & 114.797 & 0 \\
-12 & 203 & 184.472 & 0 \\
-11 & 218 & 205.504 & 0 \\
-10 & 238 & 224.300 & 0 \\
-9  & 244 & 236.298 & 0 \\
-8  & 247 & 243.224 & 0 \\
-7  & 248 & 245.440 & 0 \\
-6  & 255 & 248.280 & 0 \\
-5  & 254 & 251.538 & 0 \\
-4  & 262 & 254.985 & 1 \\
-3  & 267 & 260.039 & 1 \\
-2  & 272 & 267.619 & 1 \\
-1  & 276 & 274.394 & 1 \\
0   & 283 & 280.881 & 1 \\
1   & 263 & 263.000 & 1 \\
2   & 232 & 231.354 & 1 \\
3   & 221 & 220.228 & 1 \\
4   & 200 & 198.841 & 1 \\
5   & 179 & 179.000 & 0 \\
6   & 155 & 154.196 & 0 \\
7   & 117 & 117.000 & 0 \\
8   & 23  & 23.000  & 0 \\
9   & 17  & 17.000  & 0 \\
10  & 2   & 2.000   & 0 \\
11  & 1   & 1.000   & 0 \\
12  & 1   & 1.000   & 0 \\
\bottomrule
\end{tabular}
\begin{tablenotes}[flushleft]
\footnotesize
\item[] \textit{Notes:} Main window equals one for event times from -4 to +4 semesters.
\end{tablenotes}
\end{threeparttable}
\end{adjustbox}
\end{table}

\begin{table}[!htbp]
\centering
\caption{Year-half process and utilization trends}
\label{tab:a6b_process_utilization_trends}
\scriptsize
\setlength{\tabcolsep}{4pt}
\renewcommand{\arraystretch}{0.90}
\begin{adjustbox}{width=\textwidth}
\begin{threeparttable}
\begin{tabular}{lrrrrrr}
\toprule
Year-half 
& \makecell{Person-\\semesters} 
& \makecell{Covered\\indicators} 
& \makecell{Measured\\indicators} 
& Hospitalization 
& ER/UCC 
& \makecell{Any acute\\use} \\
\midrule
2016H1 & 908.615   & 4.034 & 3.581 & 0.045 & 0.169 & 0.184 \\
2016H2 & 1,098.337 & 4.607 & 3.227 & 0.045 & 0.176 & 0.193 \\
2017H1 & 1,169.884 & 4.839 & 3.423 & 0.055 & 0.175 & 0.198 \\
2017H2 & 1,219.217 & 4.841 & 3.272 & 0.039 & 0.185 & 0.194 \\
2018H1 & 1,245.497 & 4.851 & 3.283 & 0.054 & 0.184 & 0.205 \\
2018H2 & 1,273.924 & 4.776 & 3.234 & 0.052 & 0.178 & 0.197 \\
2019H1 & 1,302.481 & 4.792 & 3.311 & 0.050 & 0.172 & 0.189 \\
2019H2 & 1,326.098 & 4.673 & 2.904 & 0.059 & 0.194 & 0.211 \\
2020H1 & 1,352.071 & 4.370 & 2.026 & 0.044 & 0.141 & 0.152 \\
2020H2 & 1,376.451 & 4.067 & 2.085 & 0.056 & 0.154 & 0.170 \\
2021H1 & 1,418.127 & 3.788 & 2.640 & 0.050 & 0.139 & 0.151 \\
2021H2 & 1,440.815 & 3.985 & 2.805 & 0.063 & 0.169 & 0.184 \\
2022H1 & 1,459.492 & 4.261 & 3.166 & 0.061 & 0.168 & 0.191 \\
2022H2 & 1,476.402 & 4.342 & 3.003 & 0.056 & 0.168 & 0.182 \\
2023H1 & 1,492.050 & 4.477 & 3.252 & 0.063 & 0.183 & 0.200 \\
2023H2 & 1,504.864 & 4.488 & 3.027 & 0.073 & 0.169 & 0.197 \\
2024H1 & 1,517.379 & 4.553 & 3.344 & 0.057 & 0.193 & 0.212 \\
2024H2 & 1,528.196 & 4.527 & 3.126 & 0.056 & 0.223 & 0.236 \\
2025H1 & 1,529.812 & 4.565 & 3.084 & 0.061 & 0.203 & 0.220 \\
2025H2 & 1,519.576 & 4.338 & 1.819 & 0.042 & 0.226 & 0.241 \\
\bottomrule
\end{tabular}
\begin{tablenotes}[flushleft]
\footnotesize
\item[] \textit{Notes:} Means are weighted by effective person-semester exposure.
\end{tablenotes}
\end{threeparttable}
\end{adjustbox}
\end{table}

\begin{table}[!htbp]
\centering
\caption{AID adoption cohorts: size and demographics}
\label{tab:aid_cohort_demographics}
\scriptsize
\setlength{\tabcolsep}{4.5pt}
\renewcommand{\arraystretch}{0.95}
\begin{threeparttable}
\begin{tabular*}{\textwidth}{@{\extracolsep{\fill}}lrrrrr}
\toprule
AID cohort 
& New adopters 
& Cumulative adopters 
& Female (\%) 
& Age mean 
& Age $<18$ (\%) \\
\midrule
Jul-Dec 2019 & 1   & 1   & 100.0 & 39.50 & 0.0 \\
Jul-Dec 2020 & 1   & 2   & 0.0   & 28.50 & 0.0 \\
Jan-Jun 2021 & 16  & 18  & 43.8  & 42.50 & 0.0 \\
Jul-Dec 2021 & 7   & 25  & 28.6  & 34.07 & 14.3 \\
Jan-Jun 2022 & 97  & 122 & 33.0  & 42.78 & 0.0 \\
Jul-Dec 2022 & 37  & 159 & 35.1  & 41.04 & 0.0 \\
Jan-Jun 2023 & 24  & 183 & 33.3  & 42.42 & 0.0 \\
Jul-Dec 2023 & 19  & 202 & 26.3  & 38.87 & 0.0 \\
Jan-Jun 2024 & 20  & 222 & 50.0  & 43.85 & 0.0 \\
Jul-Dec 2024 & 10  & 232 & 70.0  & 46.10 & 10.0 \\
Jan-Jun 2025 & 31  & 263 & 61.3  & 44.68 & 0.0 \\
Jul-Dec 2025 & 20  & 283 & 35.0  & 41.95 & 0.0 \\
All adopters & 283 & 283 & 39.2  & 42.31 & 0.7 \\
\bottomrule
\end{tabular*}
\begin{tablenotes}
\scriptsize
\item Notes: A cohort is defined by the first semester in which a patient is observed receiving automated insulin delivery (AID). Percentages are computed among non-missing observations.
\end{tablenotes}
\end{threeparttable}
\end{table}


\begin{table}[!htbp]
\centering
\caption{AID adoption cohorts: technology use and preventive-care coverage}
\label{tab:aid_cohort_preventive_care}
\scriptsize
\setlength{\tabcolsep}{4.0pt}
\renewcommand{\arraystretch}{0.95}
\begin{threeparttable}
\begin{adjustbox}{max width=\textwidth}
\begin{tabular}{lrrrrrrr}
\toprule
AID cohort 
& Pump active 
& AID share 
& Proc. index 
& HbA1c 
& LDL 
& BP 
& Eye \\
& (\%) 
& (\%) 
& mean 
& cov. (\%) 
& cov. (\%) 
& cov. (\%) 
& cov. (\%) \\
\midrule
Jul-Dec 2019 & 0.0  & .     & 5.00 & 100.0 & 0.0  & 100.0 & 0.0 \\
Jul-Dec 2020 & 0.0  & 11.41 & 2.00 & 0.0   & 0.0  & 100.0 & 0.0 \\
Jan-Jun 2021 & 0.0  & 58.98 & 4.62 & 68.8  & 93.8 & 62.5  & 25.0 \\
Jul-Dec 2021 & 0.0  & 79.89 & 4.86 & 85.7  & 57.1 & 85.7  & 14.3 \\
Jan-Jun 2022 & 0.0  & 39.65 & 3.16 & 38.1  & 68.0 & 46.4  & 12.4 \\
Jul-Dec 2022 & 0.0  & 67.05 & 3.57 & 43.2  & 78.4 & 59.5  & 16.2 \\
Jan-Jun 2023 & 8.3  & 28.13 & 3.50 & 50.0  & 66.7 & 50.0  & 20.8 \\
Jul-Dec 2023 & 0.0  & 50.97 & 3.58 & 47.4  & 68.4 & 52.6  & 15.8 \\
Jan-Jun 2024 & 10.0 & 50.03 & 4.15 & 40.0  & 90.0 & 70.0  & 30.0 \\
Jul-Dec 2024 & 0.0  & 44.62 & 3.50 & 30.0  & 80.0 & 70.0  & 20.0 \\
Jan-Jun 2025 & 22.6 & 43.49 & 4.26 & 51.6  & 90.3 & 67.7  & 22.6 \\
Jul-Dec 2025 & 5.0  & 67.53 & 4.85 & 65.0  & 90.0 & 70.0  & 30.0 \\
All adopters & 4.2  & 48.34 & 3.72 & 46.6  & 76.0 & 57.6  & 18.4 \\
\bottomrule
\end{tabular}
\end{adjustbox}
\begin{tablenotes}[flushleft]
\scriptsize
\item[] \textit{Notes}: Pump active and AID share are measured in the adoption semester. Previous-semester coverage variables are measured immediately before AID adoption. The preventive-care coverage index counts covered process indicators among HbA1c, LDL, blood pressure, BMI, albuminuria, eGFR, and eye-exam coverage.
\end{tablenotes}
\end{threeparttable}
\end{table}

\begin{table}[!htbp]
\centering
\caption{AID adoption cohorts: prior health risk}
\label{tab:aid_cohort_prior_risk}
\scriptsize
\setlength{\tabcolsep}{5.5pt}
\renewcommand{\arraystretch}{0.95}
\begin{threeparttable}
\begin{tabular*}{\textwidth}{@{\extracolsep{\fill}}lrrr}
\toprule
AID cohort 
& New adopters 
& Prior acute use (\%) 
& Prior complication (\%) \\
\midrule
Jul-Dec 2019 & 1   & 100.0 & 0.0 \\
Jul-Dec 2020 & 1   & 100.0 & 0.0 \\
Jan-Jun 2021 & 16  & 81.2  & 18.8 \\
Jul-Dec 2021 & 7   & 57.1  & 0.0 \\
Jan-Jun 2022 & 97  & 82.5  & 32.0 \\
Jul-Dec 2022 & 37  & 89.2  & 48.6 \\
Jan-Jun 2023 & 24  & 66.7  & 20.8 \\
Jul-Dec 2023 & 19  & 68.4  & 15.8 \\
Jan-Jun 2024 & 20  & 85.0  & 20.0 \\
Jul-Dec 2024 & 10  & 80.0  & 40.0 \\
Jan-Jun 2025 & 31  & 90.3  & 35.5 \\
Jul-Dec 2025 & 20  & 100.0 & 20.0 \\
All adopters & 283 & 82.7  & 29.3 \\
\bottomrule
\end{tabular*}
\begin{tablenotes}[flushleft]
\scriptsize
\item[] \textit{Notes}: Prior acute use and prior complications are measured before the first AID semester. Prior acute use includes any hospitalization or ER/UCC contact before adoption.
\end{tablenotes}
\end{threeparttable}
\end{table}


\begin{table}[!htbp]
\centering
\scriptsize
\caption{Timing of deaths and non-death exits by semester}
\label{tab:death_attrition_timing}
\begin{threeparttable}
\begin{tabular*}{\textwidth}{@{\extracolsep{\fill}}llrrr}
\toprule
Semester & Event & All patients & Ever AID & Never AID \\
\midrule
2020H2 & Death & 2 & 0 & 2 \\
2021H1 & Death & 7 & 0 & 7 \\
2021H2 & Death & 6 & 0 & 6 \\
2022H1 & Death & 9 & 1 & 8 \\
2022H2 & Death & 13 & 0 & 13 \\
2023H1 & Death & 6 & 0 & 6 \\
2023H2 & Death & 6 & 1 & 5 \\
2024H1 & Death & 14 & 2 & 12 \\
2024H2 & Death & 12 & 1 & 11 \\
2025H1 & Death & 8 & 0 & 8 \\
\bottomrule
\end{tabular*}
\begin{tablenotes}[flushleft]
\footnotesize
\item Notes: Deaths are assigned to the semester containing the registry death date. Non-death attrition is assigned to the patient's last active follow-up semester. Non-death attrition excludes patients who die during the study window.
\end{tablenotes}
\end{threeparttable}
\end{table}

\clearpage
\renewcommand{\thefigure}{D\arabic{figure}}
\renewcommand{\thetable}{D\arabic{table}}
\renewcommand{\theHfigure}{D\arabic{figure}}
\renewcommand{\theHtable}{D\arabic{table}}
\setcounter{figure}{0}
\setcounter{table}{0}

\begin{figure}[!htbp]
\centering
\includegraphics[width=0.95\textwidth]{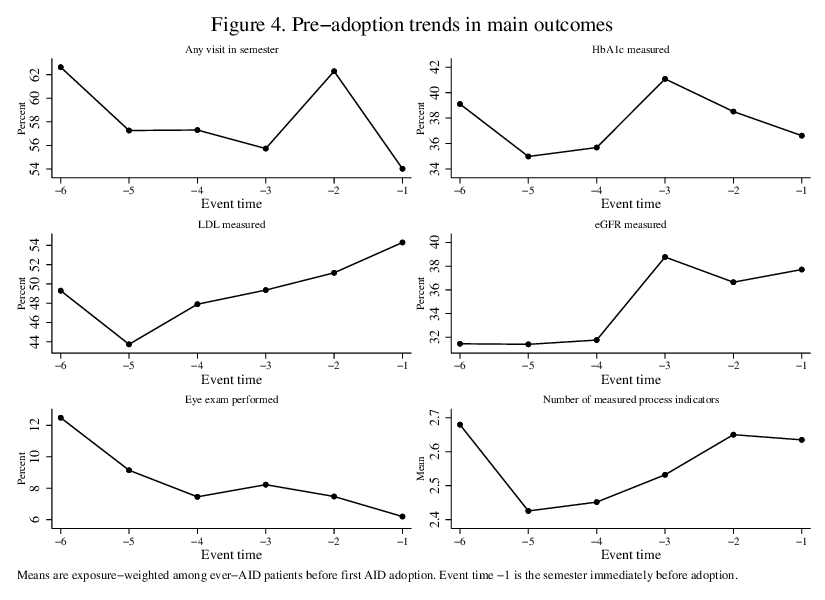}
\caption{Pre-adoption trends in main outcomes}
\label{fig:preadoption_trends}
\end{figure}
\clearpage

\renewcommand{\thefigure}{E\arabic{figure}}
\renewcommand{\thetable}{E\arabic{table}}
\renewcommand{\theHfigure}{E\arabic{figure}}
\renewcommand{\theHtable}{E\arabic{table}}
\setcounter{figure}{0}
\setcounter{table}{0}

\end{document}